\documentclass[12pt]{article}
\usepackage{amsmath}
\usepackage{graphicx}
\usepackage{enumerate}
\usepackage{natbib}
\usepackage{url} 

\newcommand{\blind}{1}

\usepackage{amssymb, amsthm, ascmac, caption, subcaption, comment, ascmac, bm}
\usepackage{latexsym, mathabx, mathrsfs, mathtools, psfrag, setspace}
\usepackage{newtxtext, xcolor}
\usepackage[stable]{footmisc}
\usepackage{multirow}
\usepackage[colorlinks=true, linkcolor=blue, filecolor=blue, urlcolor=blue, citecolor=blue, backref=page]{hyperref}
\usepackage{xr}
\usepackage{clipboard}
\newclipboard{Manuscript}
\DeclareMathOperator*{\bbE}{\mathbb{E}}

\DeclareMathOperator*{\Cov}{\mathrm{Cov}}
\DeclareMathOperator*{\Var}{\mathrm{Var}}
\DeclareMathOperator*{\Lip}{\mathrm{Lip}}
\renewcommand{\hat}{\widehat}
\renewcommand{\tilde}{\widetilde}
\renewcommand{\check}{\widecheck}

\theoremstyle{definition}
\newtheorem{theorem}{Theorem}[section]
\newtheorem{assumption}{Assumption}[section]
\newtheorem{definition}{Definition}[section]

\newtheorem{example}{Example}[section]
\newtheorem{lemma}{Lemma}[section]
\newtheorem{proposition}{Proposition}[section]
\newtheorem{remark}{Remark}[section]

\numberwithin{equation}{section}

\begin{document}

\def\spacingset#1{\renewcommand{\baselinestretch}%
	{#1}\small\normalsize} \spacingset{1}


\if1\blind
{
	\title{\bf Causal Inference with Noncompliance and Unknown Interference}
	\author{Tadao Hoshino\thanks{1-6-1 Nishi-Waseda, Shinjuku-ku, Tokyo 169-8050, Japan. Email: \href{mailto:thoshino@waseda.jp}{thoshino@waseda.jp}}\hspace{.2cm}\\
		School of Political Science and Economics, Waseda University \\
		and \\
		Takahide Yanagi\thanks{Corresponding author: Yoshida Honmachi, Sakyo, Kyoto, 606-8501, Japan. Email: \href{mailto:yanagi@econ.kyoto-u.ac.jp}{yanagi@econ.kyoto-u.ac.jp}} \\
		Graduate School of Economics, Kyoto University}
	\maketitle
} \fi

\if0\blind
{
	\bigskip
	\bigskip
	\bigskip
	\begin{center}
		{\LARGE\bf Causal Inference with Noncompliance and Unknown Interference}
	\end{center}
	\medskip
} \fi

\bigskip
\begin{abstract}
	We consider a causal inference model in which individuals interact in a social network and they may not comply with the assigned treatments.
	In particular, we suppose that the form of network interference is unknown to researchers.
	To estimate meaningful causal parameters in this situation, we introduce a new concept of exposure mapping, which summarizes potentially complicated spillover effects into a fixed dimensional statistic of instrumental variables.
	We investigate identification conditions for the intention-to-treat effects and the average treatment effects for compliers, while explicitly considering the possibility of misspecification of exposure mapping.
	Based on our identification results, we develop nonparametric estimation procedures via inverse probability weighting.
	Their asymptotic properties, including consistency and asymptotic normality, are investigated using an approximate neighborhood interference framework.
	For an empirical illustration, we apply our method to experimental data on the anti-conflict intervention school program.
	The proposed methods are readily available with the companion {\ttfamily R} package \href{https://tkhdyanagi.github.io/latenetwork/}{{\ttfamily latenetwork}}.
\end{abstract}

\noindent%
{\it Keywords:} exposure mapping, instrumental variables, local average treatment effect, network interference, spillover effects.
\vfill

\newpage
\spacingset{1.9} 


\section{Introduction} \label{sec:introduction}

Estimating causal effects under cross-unit interference has become increasingly important in various fields.
When individuals interact with each other, using the conventional potential outcome framework of \cite{rubin1980discussion} based on the \textit{stable unit treatment value assumption} (SUTVA) is inappropriate.
To address the potential interference, there has been a rapidly growing number of studies that attempt to mitigate SUTVA by replacing it with some weaker restrictions.

A common approach to dealing with interference is to assume the existence of a low-dimensional \textit{exposure mapping} which serves as a sufficient statistic of spillover effects in that others' treatments affect one's outcomes only through this function (e.g., \citealp{hong2006evaluating}; \citealp{hudgens2008toward}; \citealp{manski2013identification}; \citealp{aronow2017estimating}; \citealp{li2019randomization}; \citealp{forastiere2021identification}; \citealp{li2022random}).
Some frequently-used forms of exposure mapping include, for example, simply extracting the neighbors' treatments from the entire treatment vector or calculating the proportion of treated neighbors.
The exposure mapping is useful for summarizing potentially complicated spillover effects, but there is an inherent difficulty in how to choose the ``right'' functional form.
Thus, some recent studies investigate under what conditions one can estimate meaningful causal parameters even under unknown interference (\citealp{savje2021average}; \citealp{leung2022causal}; \citealp{savje2023causal}).

Including the aforementioned studies, much of the research on causal inference with interference assumes the availability of experimental data where the individuals fully comply with their assigned treatments.
However, this should be restrictive in many applications (e.g., \citealp{miguel2004worms}; \citealp{dupas2014short}; \citealp{zelizer2019position}).
As a real example, consider the experiment on social norms and behaviors of adolescents conducted by \cite{paluck2016changing}.
They randomly selected students to participate in the anti-conflict intervention program where the participants were encouraged to take on leadership roles to reduce conflicts in school.
The authors were interested in assessing the effectiveness of the intervention against one's own behavior, as well as whether the participants influence their peers through their friendship network.
Unfortunately, a certain proportion of the selected students did not join the intervention program, which led the authors to compromise with an intention-to-treat (ITT) analysis.

Although the coexistence of spillovers and noncompliance should be prevalent in empirical applications, only a few studies have explicitly tackled this issue.
\cite{sobel2006randomized} shows that the conventional methods, such as the two-stage least squares estimator, may not admit causal interpretations when ignoring spillover effects.
While there are studies that deal with both spillovers and noncompliance using an instrumental variable (IV) method by extending the local average treatment effect (LATE) framework of \cite{imbens1994identification} (e.g., \citealp{kang2016peer}; \citealp{kang2018spillover}; \citealp{imai2020causal}; \citealp{ditraglia2023identifying}; \citealp{vazquez2023causal}), they rule out interactions on a large network and, more importantly, do not explicitly consider the misspecification of exposure mappings.

Taken these points together, it should be of primary importance to understand what causal parameters we can identify (if any) and how to perform statistical inference on them under the possibility of noncompliance and network interference of unknown form, which is the objective of this study.
We consider a model in which individuals are connected through a single large network and they may self-select their treatment status.
To account for the noncompliance issue and network interference, we employ the IV method and introduce a new concept of exposure mapping, which we call \textit{instrumental exposure mapping} (IEM).
The IEM is similar to the conventional exposure mapping in that it is a function summarizing the spillover effects into low-dimensional variables, but it differs in that it is a function of IVs.

We begin by considering the ITT analysis, wherein the estimands of interest are the average direct effect (ADE) and average indirect effect (AIE) of the IV on the outcome and those on the treatment choice.
We show that these estimands have clear causal interpretation even with a misspecified IEM.
Next, we focus on identifying the average direct and indirect effects for \textit{compliers} who comply with their assigned treatments, which we call the \textit{local average direct effect} (LADE) and \textit{local average indirect effect} (LAIE), respectively.
Under certain conditions, these LATE-type parameters capture the direct and indirect effects of the treatment receipt on the outcome for compilers, and thus, they should be more interpretable and policy-relevant than the simple ITT parameters.
The technical difficulty in identifying these LATE-type parameters is that the standard identification argument in the LATE literature cannot be directly applied if no additional restrictions on the interaction structure are given.
To address this problem, we extend the \textit{restricted interference} assumption in \cite{imai2020causal} to our situation.
It is shown that the LADE and LAIE parameters are identifiable from Wald-type estimands under certain restricted interference assumptions.
Importantly, these identification results imply that the ITT analysis disregarding the noncompliance may underestimate the direct and indirect treatment effects.

We propose nonparametric estimation procedures via inverse probability weighting.
Our estimators are easy to implement, but their statistical properties are non-trivial because of unknown interference.
We impose two key assumptions to show that our estimators are consistent and asymptotically normal.
The first is the \textit{approximate neighborhood interference} (ANI) assumption of \cite{leung2022causal}, which is suitable for many empirical situations where spillover effects from distant units are weaker than those from close ones.
The second assumption is that the network structure is ``sparse'' such that the number of link connections is sufficiently small for each unit.

\paragraph{Related literature}

Our identification results for the ADE parameters build particularly on \cite{imai2020causal}, who study the identification of average causal effects for compliers in two-stage randomized experiments under noncompliance.
A crucial assumption underlying their model is that interference is restricted within disjoint groups (i.e., partial interference).
More importantly, they a priori assume a stratified interference mechanism in which spillovers are determined only through the number of treatment assignments within each cluster.
In contrast, this paper addresses the network interference of unknown form leveraging a potentially misspecified IEM.
See Remark \ref{remark:imaietal} for further comparison.

Our definitions and identification arguments for the AIE parameters extend \cite{hu2022average} to the case of noncompliance.
As in their paper, we define the AIE parameters based on the \textit{interference graph} (cf. \citealp{aronow2021spillover}).
Compared to \cite{hu2022average}, we explicitly consider a potential misspecification of \textit{interference set} -- the set of units who are affected by a focal unit.
Specifically, we allow for the possibility that the effects originated from a focal unit spillover beyond its interference set.

Another closely related study is \cite{leung2022causal}.
He proposes an ANI model in experimental situations with perfect compliance and develops inferential methods for average treatment effects while explicitly allowing for the misspecification of exposure mapping.
The major distinction between \cite{leung2022causal} and ours is that not only the spillover effects of treatments on the outcome but also that of others' IVs on own treatment choice are considered in our study.

\paragraph{Paper organization}

Section \ref{sec:model} presents our model setup.
Sections \ref{sec:identification} and \ref{sec:estimation} provide the identification and estimation results, respectively.
Section \ref{sec:numerical} reports the numerical results.
The companion {\ttfamily R} package \href{https://tkhdyanagi.github.io/latenetwork/}{{\ttfamily latenetwork}} is available from the authors' websites.

\section{Model} \label{sec:model}

Consider a finite population of $n \in \mathbb{N}$ units $N_n \coloneqq \{1, 2, \dots, n\}$.
The units form an undirected network represented by the $n \times n$ symmetric adjacency matrix $\bm{A} = (A_{ij})_{i,j \in N_n}$, where $A_{ij} \in \{ 0, 1 \}$ indicates whether or not $i$ and $j$ are connected.
We assume that there are no self-links so that $A_{ii} = 0$ for all $i \in N_n$.
Denote the set of possible adjacency matrices of $n$ units as $\mathcal{A}_n$.

In a later section, we study asymptotic theory under the condition that the network size $n$ grows to infinity.
This means that we consider a sequence of networks $\{ \bm{A}_m \}$ for $m = 1, 2, \dots$.
The observed adjacency matrix $\bm{A}$ with no subscript is regarded as an $n$-th element of the sequence.
Generally, networks $\bm{A}_{m_1}$ and $\bm{A}_{m_2}$ are completely unrelated, such that the members of the former and latter do not overlap.
Meanwhile, it is possible to create a new network $\bm{A}_{m_3}$ ($m_3 = m_1 + m_2$) as a block diagonal matrix with the diagonal submatrices $\bm{A}_{m_1}$ and $\bm{A}_{m_2}$.
Furthermore, the distributions of variables such as treatment assignments may be specific to each network in general; that is, they form a triangular array defined along with the network sequence.
However, for notational simplicity, we suppress the dependence of variables on the network.

Let $Y_i \in \mathbb{R}$ be an observed outcome variable and $D_i \in \{ 0, 1 \}$ an indicator of the treatment receipt for $i$.
In observational studies or randomized experiments with possible noncompliance, individuals may self-select their treatment status and the existing methods under perfect compliance may not be applicable.
To address this problem, suppose that there is a binary IV, $Z_i \in \{ 0, 1 \}$.
In an experimental setup, $Z_i$ is typically an indicator of initial treatment recommendation for $i$.
Denote the $n$-dimensional vector of realized treatments as $\bm{D} = (D_i)_{i \in N_n}$, and similarly let $\bm{Z} = (Z_i)_{i \in N_n}$.
We write the support of $\bm{D}$ and that of $\bm{Z}$ as $\mathcal{D}_n = \{ 0, 1 \}^n$ and $\mathcal{Z}_n = \{ 0, 1 \}^n$, respectively.
For each $\bm{d} \in \mathcal{D}_n$ and $\bm{z} \in \mathcal{Z}_n$, we denote the potential outcome of unit $i$ when $\bm{D} = \bm{d}$ and $\bm{Z} = \bm{z}$ as $Y_i(\bm{d}, \bm{z})$.
Similarly, the potential treatment status given $\bm{Z} = \bm{z}$ is written as $D_i(\bm{z})$.
Let $\bm{D}(\bm{z}) = (D_i(\bm{z}))_{i \in N_n}$ be the $n$-dimensional vector of potential treatments.
By construction, we have $Y_i = Y_i(\bm{D}, \bm{Z})$, $D_i = D_i(\bm{Z})$, and $\bm{D} = \bm{D}(\bm{Z})$.
Hence, we can further write $y_i(\bm{z}) = Y_i(\bm{D}(\bm{z}), \bm{z})$ for some function $y_i: \mathcal{Z}_n \to \mathbb{R}$, and we have $Y_i = y_i(\bm{Z}) = \sum_{\bm{z} \in \mathcal{Z}_n} \bm{1}\{ \bm{Z} = \bm{z} \} y_i(\bm{z})$.
Denoting $\bm{z}_{-i} = (z_k)_{k \neq i}$, we write the potential outcome of unit $i$ given $Z_j = z_j$ and $\bm{Z}_{-j} = \bm{z}_{-j}$ as $y_i(Z_j = z_j, \bm{Z}_{-j} = \bm{z}_{-j})$, where $i$ may differ from $j$.
When $i = j$, we use both $y_i(z_i, \bm{z}_{-i})$ and $y_i(\bm{z})$ interchangeably depending on the situation.
The same notation applies to the other functions of $\bm{z}$.

Because we can observe only one realization from $(y_i(\bm{z}), D_i(\bm{z}))_{\bm{z} \in \mathcal{Z}_n}$ for each unit, it is generally impossible to define identifiable causal estimands without introducing some restrictions.
Here, we consider a pre-specified function $T: N_n \times \{0,1\}^{n-1} \times \mathcal{A}_n \to \mathcal{T}$, where $\mathcal{T} \subset \mathbb{R}^{\dim(T)}$ does not depend on $i$ and $n$, and $\dim(T)$ is a fixed positive integer.\footnote{
	In the literature, \cite{forastiere2021identification} consider an exposure mapping whose range may be heterogeneous across $i$ and $n$.
	Although our results would hold with minor modifications even in the presence of heterogeneity in exposure mappings, we consider this aspect beyond this study's scope because such a generalization substantially complicates asymptotic theory.
	Nevertheless, the common range assumption should not be too restrictive in practice, considering that in our framework, researchers can arbitrarily specify the form of IEM.
}
We call the function $T$ the {\it instrumental exposure mapping} (IEM) and its realization $T_i = T(i, \bm{Z}_{-i}, \bm{A})$ the {\it instrumental exposure}.
When no confusion exists, we suppress the dependence of $T$ on $\bm{A}$, that is, $T(i, \bm{z}_{-i}) = T(i, \bm{z}_{-i}, \bm{A})$.
Without loss of generality, we may suppose that the functional form of $T$ does not depend on $i$'s own $Z_i$.\footnote{
	Even when the level of exposure changes depending on unit's own IV, we can consider the exposures when $Z_i = 1$ and $Z_i = 0$ separately, say $T_1(i, \bm{z}_{-i})$ and $T_0(i, \bm{z}_{-i})$, respectively, and re-define $T(i, \bm{z}_{-i}) \coloneqq (T_0(i, \bm{z}_{-i}), T_1(i, \bm{z}_{-i}))$.
	We thank the referee for suggesting this point.
}
We say that the IEM is {\it correctly specified} on $\bm{A}$ if
\begin{align} \label{eq:correctIEM}
	T(i, \bm{z}_{-i}) = T(i, \bm{z}_{-i}') \Longrightarrow D_i(z_i, \bm{z}_{-i}) = D_i(z_i, \bm{z}_{-i}')\;\; \text{and} \;\; y_i(z_i, \bm{z}_{-i}) = y_i(z_i, \bm{z}_{-i}')
\end{align}
for all $i \in N_n$, $z_i \in \{ 0, 1 \}$, and $\bm{z}_{-i}, \bm{z}_{-i}' \in \{ 0, 1 \}^{n-1}$.\footnote{
	Note that the definition in \eqref{eq:correctIEM} does not imply the uniqueness of correct IEM.
	For example, if having at least one treatment-eligible neighbor is a correct IEM for $i$ (i.e., $T_i^\text{max} = \max\{Z_j: A_{ij} = 1\}$), so is $i$'s neighborhood average $T_i^\text{ave} = \sum_{j \neq i} A_{ij} Z_j/\sum_{j \neq i} A_{ij}$, because $T_i^\text{ave} = 0$ and $T_i^\text{max} = 0$ are equivalent and $T_i^\text{ave} = t$ for any $t > 0$ is only a special case of $T_i^\text{max} = 1$.
	That is, $T_i^\text{ave}$ is a ``finer'' exposure than $T_i^\text{max}$.
	Based on such a hierarchical structure of different exposure mappings, \cite{hoshino2023randomization} propose a randomization test for the specification of exposure mappings.
	If any spillover effects do not exist in the first place, any IEM is correct.
} 

If the IEM is correctly specified, it serves as a fixed dimensional sufficient statistic that summarizes potentially high-dimensional spillover effects.
That is, the potential treatment status and the potential outcome of unit $i$ can be fully characterized by $i$'s own IV $Z_i$ and her instrumental exposure $T_i$, and there exist functions $\tilde d_i: \{ 0, 1 \} \times \mathcal{T} \to \{ 0, 1 \}$ and $\tilde y_i: \{ 0, 1 \} \times \mathcal{T} \to \mathbb{R}$ satisfying $\tilde d_i(z_i, T(i, \bm{z}_{-i})) = D_i(z_i, \bm{z}_{-i})$ and $\tilde y_i(z_i, T(i, \bm{z}_{-i})) = y_i(z_i, \bm{z}_{-i})$ for any $z_i \in \{ 0, 1 \}$ and $\bm{z}_{-i} \in \{ 0, 1 \}^{n-1}$.
Then, $\tilde y_i(z, t)$ and $\tilde d_i(z, t)$ represent the potential outcome and the potential treatment status, respectively, given $Z_i = z$ and $T_i = t$.
In this way, a properly specified IEM alleviates the complexity of handling general spillover effects and greatly simplifies the estimation of causal parameters under interference.
However, in reality, the user-specified IEM is generally incorrect.
In this case, $\tilde y_i(z, t)$ and $\tilde d_i(z, t)$ are no longer well-defined.

Throughout the paper, following the recent literature on causal inference with interference, we focus on a design-based uncertainty framework where the randomness comes only from $\bm{Z}$.
That is, we treat the potential outcomes, potential treatments, and adjacency matrix as non-stochastic components.
The design-based approach is suitable for randomized experiments where researchers can design the random assignment mechanism for treatment eligibility or initial recommendations, as in \cite{dupas2014short} and \cite{paluck2016changing}.
Even in observational studies, the design-based approach should be relevant when we can observe the entire population or most of the finite population (cf. \citealp{abadie2020sampling}).
It should be noted that we can also view our framework as a random design on which everything other than $\bm{Z}$ is conditioned.

We assume that the IV can affect the outcome only through the treatment.

\begin{assumption} [Exclusion restriction] \label{as:exclusion}
	$Y_i(\bm{d}, \bm{z}) = Y_i(\bm{d}, \bm{z}')$ for all $i \in N_n$, $\bm{d} \in \mathcal{D}_n$, and $\bm{z}, \bm{z}' \in \mathcal{Z}_n$.
\end{assumption}

Under Assumption \ref{as:exclusion}, we can reduce the potential outcome when $\bm{D} = \bm{d}$ to $Y_i(\bm{d}) = Y_i(\bm{d}, \bm{z})$, and we have $y_i(\bm{z}) = Y_i(\bm{D}(\bm{z}))$.
Note that this assumption is not essentially necessary in terms of the ITT analysis, but it can greatly improve the causal interpretation of our parameters.

We provide two specific examples that can be effectively analyzed within our model.

\begin{example}\label{example:linear1}
	Suppose that the outcome is generated as $Y_i = \beta_{0i} + \beta_1 D_i + \beta_2 \cdot \bm{1}\{ \sum_{j \neq i} A_{ij} D_j > c \}$, where $\beta_{0i}$ is an idiosyncratic intercept term, $\beta_1$ and $\beta_2$ indicate the direct and spillover effects, respectively, and $c$ is a given threshold.
	Assume that the treatment status of each unit is determined only by her own IV: $D_i(z_i) = D_i(z_i, \bm{z}_{-i})$.
	Then, the potential outcome when $\bm{Z} = \bm{z}$ is $y_i(\bm{z}) = \beta_{0i} + \beta_1 D_i(z_i) + \beta_2 \cdot \bm{1} \{ \sum_{j \neq i} A_{ij} D_j(z_j) > c \}$.
	A correctly specified IEM is, for example, $T(i, \bm{Z}_{-i}) = \bm{1}\{ \sum_{j \neq i} A_{ij} D_j(Z_j) > c \}$ with $\mathcal{T} = \{ 0, 1 \}$.
	In the literature, this type of exposure mapping is used, for example, in \cite{hong2006evaluating} and \cite{leung2022causal}.
\end{example}

\begin{example}\label{example:latent1}
	Suppose that Assumption \ref{as:exclusion} holds and that no interference exists in the outcome.
	For the treatment choice, consider the latent index model $D_i = \bm{1}\{ \gamma_{0i} + \gamma_{1i} Z_i + \gamma_{2i} \cdot \bm{1} \{ \sum_{j \neq i} A_{ij} Z_j > c \} > 0 \}$, where $\gamma_{0i}$ is the preference heterogeneity for the treatment, and $\gamma_{1i}$ and $\gamma_{2i}$ respectively capture the direct and spillover effects of the IV.
	In this situation, the potential outcome when $\bm{Z} = \bm{z}$ is given by $y_i(\bm{z}) = \beta_{0i} + \beta_{1i} D_i(\bm{z})$.
	As such, the model is a simple binary treatment model with potentially many IVs.
	It is straightforward to find that we can estimate a LATE-type parameter using the two-stage least squares method under a monotonicity condition between $D_i$ and $Z_i$ (e.g., $\gamma_{1i} \ge 0$ for all $i$), ignoring the spillover effect in the treatment choice model.
	If we set $T(i, \bm{Z}_{-i}) = \bm{1}\{ \sum_{j \neq i} A_{ij} Z_j > c \}$, this is clearly a correct IEM.
\end{example}

\section{Identification} \label{sec:identification}

To begin with, we introduce the following assumption:

\begin{assumption}[Independence] \label{as:independence}
	$\{Z_i\}_{i \in N_n}$ are mutually independent.
\end{assumption}

This assumption would be reasonable for many empirical situations.
For example, the assumption is satisfied in a randomized experiment where the treatment eligibility is independently assigned to each unit with a given probability.
Another example is an observational study in which the IV of each unit is determined independently from the other units.

Assumption \ref{as:independence} will be used for both the identification analysis and asymptotic investigations.
If the IVs are correlated with each other, this causes non-trivial difficulties in proceeding the subsequent analysis.
For the same reason, several studies in the literature of network interference predominantly focus on Bernoulli experiments (e.g., \citealp{hu2022average}; \citealp{li2022random}).
Investigating more general treatment assignment mechanisms is left for future research.

\subsection{Intention-to-treat estimands and causal parameters of interest} \label{subsec:parameters}

We define the ITT estimands and the causal parameters of interest.
To this end, consider a non-random sub-population $S_n \subseteq N_n$.
Throughout the paper, we assume $S_n$ to be non-empty, and consider estimating causal parameters specific to this sub-population.
For an example of $S_n$, let $S_n(\delta)$ be the set of units whose degrees are $\delta$: $S_n(\delta) = \{ i \in N_n: \sum_{j \neq i} A_{ij} = \delta \}$.
In this case, we can examine whether the causal impacts vary across individuals with different centrality by comparing the estimates obtained with different $\delta$'s.
See Remark \ref{remark:subpopulation} for further discussion.

To define the ADE estimands, let $\mu_i^Y(z, t) \coloneqq \bbE[ Y_i \mid Z_i = z, T_i = t]$ for $z \in \{ 0, 1 \}$ and $t \in \mathcal{T}$.
Here, the expectation is taken with respect to the distribution of $\bm{Z}$; that is, $\mu_i^Y(z, t) 
= \sum_{\bm{z} \in \mathcal{Z}_n} \Pr[\bm{Z} = \bm{z} \mid Z_i = z, T_i = t] y_i(\bm{z})$.
Due to the heterogeneity in the potential outcomes, generally we cannot obtain a consistent estimator of $\mu_i^Y(z, t)$ in the design-based approach.
Denoting $\bar \mu_{S_n}^Y(z, t) \coloneqq |S_n|^{-1} \sum_{i \in S_n} \mu_i^Y(z, t)$, where $|S_n|$ is the cardinality of $S_n$, the ADE of the IV is defined by $\mathrm{ADEY}_{S_n}(t) \coloneqq \bar \mu_{S_n}^Y(1, t) - \bar \mu_{S_n}^Y(0, t)$.
Similarly, we define $\mu_i^D(z, t) \coloneqq \bbE[ D_i \mid Z_i = z, T_i = t]$, $\bar \mu_{S_n}^D(z, t) \coloneqq |S_n|^{-1} \sum_{i \in S_n} \mu_i^D(z, t)$, and $\mathrm{ADED}_{S_n}(t) \coloneqq \bar \mu_{S_n}^D(1, t) - \bar \mu_{S_n}^D(0, t)$.

Next, we turn to the AIE estimands.
Let $\ell_{\bm{A}}(i, j)$ denote the path distance (defined on the whole population $N_n$) between units $i$ and $j$.\footnote{
	The {\it path distance} between $i$ and $j$ is the length of the shortest sequence of neighboring edges connecting them.
	As convention, we set $\ell_{\bm{A}}(i, j) = \infty$ when no path exists between $i$ and $j$ in $\bm{A}$ and 0 if $i = j$.
}
Suppose that each $i$'s instrumental exposure $T_i$ depends only on unit $j$'s such that $1 \le \ell_{\bm{A}}(i, j) \le K$ with some constant $K \ge 1$ (see Assumption \ref{as:exposure}).
Given this, we define the \textit{interference graph} $\bm{E} = ( E_{ij} )_{i,j \in N_n}$, where $E_{ij} \coloneqq \bm{1}\{1 \le \ell_{\bm{A}}(i, j) \le K\}$.
For each $i \in S_n$, let $\mathcal{E}_i \coloneqq \{ j \in N_n: E_{ij} = 1 \}$, which we call \textit{$i$'s interference set}.
Note that $i \notin \mathcal{E}_i$ because $\ell_{\bm{A}}(i, i) = 0$.
Write $\mu_{ji}^Y(z) \coloneqq \bbE[ Y_j \mid Z_i = z ]$ and $\bar \mu_{S_n}^{Y}(z ; \mathcal{E}) \coloneqq |S_n|^{-1} \sum_{i \in S_n} \sum_{j \in \mathcal{E}_i} \mu_{ji}^Y(z)$. 
Then, we define $\mathrm{AIEY}_{S_n} \coloneqq \bar \mu_{S_n}^Y(1; \mathcal{E}) - \bar \mu_{S_n}^Y(0; \mathcal{E})$.
Similarly, we define $\mu_{ji}^D(z) \coloneqq \bbE[ D_j \mid Z_i = z ]$, $\bar \mu_{S_n}^D(z; \mathcal{E}) \coloneqq |S_n|^{-1} \sum_{i \in S_n} \sum_{j \in \mathcal{E}_i} \mu_{ji}^D(z)$, and $\mathrm{AIED}_{S_n} \coloneqq \bar \mu_{S_n}^D(1; \mathcal{E}) - \bar \mu_{S_n}^D(0; \mathcal{E})$.

We will show later that these ITT estimands have clear causal interpretation as a certain weighted average of the effect of the IV on the outcome and as that on the treatment receipt.
However, as will be highlighted, the ITT parameters may underestimate the effects of the treatment.
To depart from the ITT analysis, we extend the notions of {\it compliers} and LATE to our setting.
For each given $\bm{z}_{-i} \in \{ 0, 1 \}^{n-1}$, let $\mathcal{C}_i(\bm{z}_{-i}) \coloneqq \mathbf{1}\{ D_i(1, \bm{z}_{-i}) = 1, D_i(0, \bm{z}_{-i}) = 0 \}$ indicate a complier who takes the treatment only when $Z_i = 1$ conditional on $\bm{Z}_{-i} = \bm{z}_{-i}$.
Notably, the compliance status may change with the others' IV values.
Denote the realized compliance status as $\mathcal{C}_i = \mathcal{C}_i(\bm{Z}_{-i})$.
Letting $\pi_i(\bm{z}_{-i}, t) \coloneqq \Pr[ \bm{Z}_{-i} = \bm{z}_{-i} \mid T_i = t ]$, the expected compliance status conditional on $T_i = t$ is $\bbE[ \mathcal{C}_i \mid T_i = t] 
= \sum_{\bm{z}_{-i} \in \{ 0, 1 \}^{n-1}} \mathcal{C}_i(\bm{z}_{-i}) \pi_i(\bm{z}_{-i}, t)$.
Then, the LADE is defined by the weighted average of $y_i(1, \bm{z}_{-i}) - y_i(0, \bm{z}_{-i})$ over the compliers:
\begin{align*}
	\mathrm{LADE}_{S_n}(t)
	& \coloneqq \sum_{i \in S_n} \sum_{\bm{z}_{-i} \in \{0,1\}^{n-1}} \left\{ y_i(1, \bm{z}_{-i}) - y_i(0, \bm{z}_{-i}) \right\} \frac{\mathcal{C}_i(\bm{z}_{-i}) \pi_i(\bm{z}_{-i}, t)}{\sum_{i \in S_n} \bbE[ \mathcal{C}_i \mid T_i = t]}.
\end{align*}
Similarly, noting that $\bbE[ \mathcal{C}_i ] = \sum_{\bm{z}_{-i} \in \{ 0, 1 \}^{n-1}} \mathcal{C}_i(\bm{z}_{-i}) \pi_i(\bm{z}_{-i})$ with $\pi_i(\bm{z}_{-i}) \coloneqq \Pr[ \bm{Z}_{-i} = \bm{z}_{-i} ]$, the LAIE is defined by the weighted average of $\sum_{j \in \mathcal{E}_i} \{ y_j(Z_i = 1, \bm{Z}_{-i} = \bm{z}_{-i}) - y_j(Z_i = 0, \bm{Z}_{-i} = \bm{z}_{-i}) \}$ over the compliers:
\begin{align*}
	\mathrm{LAIE}_{S_n}
	& \coloneqq \sum_{i \in S_n} \sum_{\bm{z}_{-i} \in \{ 0, 1\}^{n-1}} \sum_{j \in \mathcal{E}_i} \{ y_j(Z_i = 1, \bm{Z}_{-i} = \bm{z}_{-i}) - y_j(Z_i = 0, \bm{Z}_{-i} = \bm{z}_{-i}) \} \frac{ \mathcal{C}_i(\bm{z}_{-i}) \pi_i(\bm{z}_{-i}) }{\sum_{i \in S_n} \bbE[ \mathcal{C}_i ] }.
\end{align*}

\begin{remark}[ADE parameters with a constant IEM] \label{remark:constantIEM1}
	The identification and estimation results presented below are still valid even when $T(i, \bm{z}_{-i})$ is constant for all $i$ and $\bm{z}_{-i}$.
	In this case, $\mathrm{ADEY}_{S_n}(t)$ reduces to $\mathrm{ADEY}_{S_n} \coloneqq \bar \mu_{S_n}^Y(1) - \bar \mu_{S_n}^Y(0)$, where $\bar \mu_{S_n}^{Y}(z) \coloneqq |S_n|^{-1} \sum_{i \in S_n} \mu_i^Y(z)$ with $\mu_i^Y(z) \coloneqq \bbE[ Y_i \mid Z_i = z ]$.
	Analogously, $\mathrm{ADED}_{S_n}(t) = \mathrm{ADED}_{S_n} \coloneqq \bar \mu_{S_n}^D(1) - \bar \mu_{S_n}^D(0)$, where $\bar \mu_{S_n}^{D}(z) \coloneqq |S_n|^{-1} \sum_{i \in S_n} \mu_i^D(z)$, and $\mu_i^D(z) \coloneqq \bbE[ D_i \mid Z_i = z ]$.
	Further, $\mathrm{LADE}_{S_n}(t)$ reduces to
	\begin{align*}
		\mathrm{LADE}_{S_n}
		\coloneqq \sum_{i \in S_n} \sum_{\bm{z}_{-i} \in \{ 0, 1 \}^{n-1}} \{ y_i(1, \bm{z}_{-i}) - y_i(0, \bm{z}_{-i}) \} \frac{ \mathcal{C}_i(\bm{z}_{-i}) \pi_i(\bm{z}_{-i}) }{ \sum_{i \in S_n} \bbE[ \mathcal{C}_i ] }.
	\end{align*}
	These parameters are natural extensions of those in \cite{hu2022average} to the case of noncompliance.
	Compared to the ADE parameters discussed above, the parameters defined here do not depend on the instrumental exposure explicitly.
	Thus, the latter parameters would be easier to interpret than the former with incorrect IEM.
	Meanwhile, if a correct IEM is used, the former parameters are more informative for understanding the treatment effect heterogeneity for the exposure level.
\end{remark}

\subsection{Average direct effects} \label{subsec:direct1}

\subsubsection{Intention-to-treat analysis} \label{subsec:ITT1}

The following proposition presents the causal interpretation of $\mathrm{ADEY}_{S_n}(t)$ and $\mathrm{ADED}_{S_n}(t)$.

\begin{proposition} \label{prop:ITT}
	Under Assumption \ref{as:independence},
	\begin{align*}
		\mathrm{ADEY}_{S_n}(t) 
		& = \frac{1}{|S_n|} \sum_{i \in S_n} \sum_{\bm{z}_{-i} \in \{0, 1\}^{n-1}} \{ y_i(1, \bm{z}_{-i}) - y_i(0, \bm{z}_{-i}) \} \pi_i(\bm{z}_{-i}, t), \\
		\mathrm{ADED}_{S_n}(t) 
		& = \frac{1}{|S_n|} \sum_{i \in S_n} \sum_{\bm{z}_{-i} \in \{0, 1\}^{n-1}} \{ D_i(1, \bm{z}_{-i}) - D_i(0, \bm{z}_{-i}) \} \pi_i(\bm{z}_{-i}, t).
	\end{align*}
\end{proposition}

Proposition \ref{prop:ITT} shows that $\mathrm{ADEY}_{S_n}(t)$ and $\mathrm{ADED}_{S_n}(t)$ have clear causal interpretation as the weighted average of $y_i(1, \bm{z}_{-i}) - y_i(0, \bm{z}_{-i})$ and as that of $D_i(1, \bm{z}_{-i}) - D_i(0, \bm{z}_{-i})$, respectively.
If we additionally impose Assumption \ref{as:exclusion}, the result for $\mathrm{ADEY}_{S_n}(t)$ can be further well interpreted.
To see this, let $\bm{D}_{-i}(Z_i = z_i, \bm{Z}_{-i} = \bm{z}_{-i}) = (D_k(Z_i = z_i, \bm{Z}_{-i} = \bm{z}_{-i}))_{k \neq i}$.
By the definition of $y_i(z_i, \bm{z}_{-i})$ and Assumption \ref{as:exclusion}, we can observe that
\begin{align*}
	& y_i(1, \bm{z}_{-i}) - y_i(0, \bm{z}_{-i}) \\
	& = \underset{\text{effect of changing $i$'s treatment through $i$'s IV}}{Y_i(D_i(1, \bm{z}_{-i}), \bm{D}_{-i}(Z_i = 1, \bm{Z}_{-i} = \bm{z}_{-i})) - Y_i(D_i(0, \bm{z}_{-i}), \bm{D}_{-i}(Z_i = 1, \bm{Z}_{-i} = \bm{z}_{-i}))} \\
	& \quad + \underset{\text{effect of changing others' treatments through $i$'s IV}}{Y_i(D_i(0, \bm{z}_{-i}), \bm{D}_{-i}(Z_i = 1, \bm{Z}_{-i} = \bm{z}_{-i})) - Y_i(D_i(0, \bm{z}_{-i}), \bm{D}_{-i}(Z_i = 0, \bm{Z}_{-i} = \bm{z}_{-i}))}.
\end{align*}
That is, $y_i(1, \bm{z}_{-i}) - y_i(0, \bm{z}_{-i})$ comprises of the direct effect of changing $i$'s own treatment status from $D_i(0, \bm{z}_{-i})$ to $D_i(1, \bm{z}_{-i})$ and the spillover effect by changing the others' treatments from $\bm{D}_{-i}(Z_i = 0, \bm{Z}_{-i} = \bm{z}_{-i})$ to $\bm{D}_{-i}(Z_i = 1, \bm{Z}_{-i} = \bm{z}_{-i})$.
Hence, Proposition \ref{prop:ITT} can be read as that $\mathrm{ADEY}_{S_n}(t)$ consists of the sum of the average direct effect from own IV and the average spillover effect caused by changing the unit's own IV.

Proposition \ref{prop:ITT} is also useful for interpreting the ADEs obtained with different IEMs.
For example, suppose that two IEMs $T$ and $T'$ generate the same ADEY value at $t$ and at $t'$, respectively:
\begin{align*}
	0 & = \mathrm{ADEY}_{S_n}(t)|_{\text{IEM}=T} -   \mathrm{ADEY}_{S_n}(t')|_{\text{IEM}=T'} \\ 
	& = \frac{1}{|S_n|} \sum_{i \in S_n} \sum_{\bm{z}_{-i} \in \{0, 1\}^{n-1}} \{ y_i(1, \bm{z}_{-i}) - y_i(0, \bm{z}_{-i}) \} \cdot \{\pi_i(\bm{z}_{-i}, t) - \pi_i'(\bm{z}_{-i}, t') \}.
\end{align*}
Thus, if this equality holds, it is a strong indication that $y_i(1, \bm{z}_{-i}) - y_i(0, \bm{z}_{-i})$ is homogeneous with respect to $\bm{z}_{-i}$ for all individuals.
Notably, if it is indeed homogeneous, the above equality must hold for any combination of IEMs, which is testable from the data.

\subsubsection{Local average direct effect} \label{subsec:LADE}

The following two conditions are analogous to the IV relevance condition and the monotonicity condition for the standard LATE estimation without interference.

\begin{assumption}[Relevance 1]\label{as:relevance1}
	$|S_n|^{-1} \sum_{i \in S_n} \bbE[ \mathcal{C}_i \mid T_i  = t ] \ge c$ for a constant $c > 0$.
\end{assumption}

\begin{assumption}[Monotonicity 1]\label{as:monotone1}
	$D_i(1, \bm{z}_{-i}) \ge D_i(0, \bm{z}_{-i})$ for all $i \in S_n$ and $\bm{z}_{-i} \in \{0, 1\}^{n - 1}$ such that $\pi_i(\bm{z}_{-i}, t) > 0$.
\end{assumption}

Assumption \ref{as:relevance1} states that there is a non-negligible proportion of units among those with $T_i = t$ whose treatment status is positively affected by the IV.
The assumption is necessary to well-define the LADE.
Assumption \ref{as:monotone1} requires that there do not exist {\it defiers}, whose treatment status is negatively affected by the IV.
This assumption limits the heterogeneity in treatment choice in that the response to the IV must be uniform.
For instance, the treatment choice equation in Example \ref{example:latent1} satisfies the condition if $\gamma_{1i} \ge 0$ for all $i$.
Note that these assumptions do not have to be fulfilled uniformly in $t \in \mathcal{T}$ as the LADE of interest is conditioned on $T_i = t$ at a given $t$.

Under Assumption \ref{as:monotone1}, conditional on $\bm{Z}_{-i}$, each individual $i$ can be classified into one of the following three latent types: {\it compliers}, {\it always takers} (those who always take the treatment), and {\it never takers} (those who never take the treatment).
More precisely, we define the indicator for always takers as $\mathcal{AT}_i = \mathcal{AT}_i(\bm{Z}_{-i}) \coloneqq \bm{1}\{ D_i(1, \bm{Z}_{-i}) = D_i(0, \bm{Z}_{-i}) = 1 \}$.
Similarly, the indicator for never takers is $\mathcal{NT}_i = \mathcal{NT}_i(\bm{Z}_{-i}) \coloneqq \bm{1}\{ D_i(1, \bm{Z}_{-i}) = D_i(0, \bm{Z}_{-i}) = 0 \}$.
Under Assumption \ref{as:monotone1}, only one of $\mathcal{C}_i$, $\mathcal{AT}_i$, and $\mathcal{NT}_i$ equals one for each $i \in S_n$.\footnote{
	While $\mathcal{C}_i$, $\mathcal{AT}_i$, and $\mathcal{NT}_i$ are determined from $D_i(1, \bm{Z}_{-i})$ and $D_i(0, \bm{Z}_{-i})$ conditional on $\bm{Z}_{-i}$, it is possible to more finely classify the units by incorporating other potential treatment responses $\{ D_i(\bm{z}) \}_{\bm{z} \in \mathcal{Z}_n}$. 
	As such an example, \cite{vazquez2023causal} classifies the units into always takers, never takers, compliers, \textit{social-interaction compliers}, and \textit{group compliers}.
	See Appendix F for further discussion.
}

Unlike conventional identification results without interference, the set of the exclusion restriction, relevance condition, and monotonicity condition does not suffice to identify the LADE.
This is because we need to account for two potential interference channels at the same time: one is the spillover effect of the IV on the treatment receipt, and the other is the spillover effect of the treatment on the outcome.
As in the conventional method, we use the variation in the instrumental value to identify the LADE, but in the present situation, the effect of shifting the IV can be amplified in two steps by the two different spillovers.
Therefore, to facilitate the identification of the LADE, some additional restriction on the interference structure is needed.
In this study, similar to \cite{imai2020causal}, we require the potential outcome $y_i(z_i, \bm{z}_{-i})$ of noncompliers to be insensitive to their own instrumental value $z_i$.

\begin{assumption}[Restricted interference 1]\label{as:restrict1}
	For all $i \in S_n$ and $\bm{z}_{-i} \in \{ 0, 1 \}^{n - 1}$ such that $\pi_i(\bm{z}_{-i}, t) > 0$, $y_i(1, \bm{z}_{-i}) = y_i(0, \bm{z}_{-i})$ holds whenever $D_i(1, \bm{z}_{-i}) = D_i(0, \bm{z}_{-i})$.
\end{assumption}

Here, we provide three empirically relevant sufficient conditions for this assumption.
The first condition is no spillovers between the IV and treatment choice:
\begin{align} \label{eq:sufficient1}
	D_i(z_i, \bm{z}_{-i}) = D_i(z_i, \bm{z}_{-i}')
	\quad
	\text{for any $z_i \in \{ 0, 1 \}$ and $\bm{z}_{-i}, \bm{z}_{-i}' \in \{ 0, 1 \}^{n-1}$}.
\end{align}
This corresponds to the {\it personalized encouragement} assumption of \cite{kang2016peer}, which states that an incentive to take the treatment must be personalized to everyone.
Under this condition, we can define the potential treatment status as $D_i(z_i) = D_i(z_i, \bm{z}_{-i})$. 
Then, the potential outcome satisfies $y_i(z_i, \bm{z}_{-i}) = Y_i(D_i(z_i), (D_k(z_k))_{k \neq i})$, implying Assumption \ref{as:restrict1}.

The second situation in which Assumption \ref{as:restrict1} holds is when there is no treatment spillover effect on the outcome; that is,
\begin{align} \label{eq:sufficient2}
	Y_i(d_i, \bm{d}_{-i}) = Y_i(d_i, \bm{d}_{-i}')
	\quad
	\text{for any $d_i \in \{ 0, 1 \}$ and $\bm{d}_{-i}, \bm{d}_{-i}' \in \{ 0, 1 \}^{n-1}$}.
\end{align}
Then, we may write the potential outcome given $D_i = d_i$ as $Y_i(d_i)$.
It is easy to see that \eqref{eq:sufficient2} implies Assumption \ref{as:restrict1}.

The third sufficient condition for Assumption \ref{as:restrict1} is that the IV of any noncomplier does not affect the treatment status of all other units; specifically, for any $\bm{z}_{-i} \in \{ 0, 1 \}^{n-1}$,
\begin{align} \label{eq:sufficient3}
	\bm{D}_{-i}(Z_i = 1, \bm{Z}_{-i} = \bm{z}_{-i}) = \bm{D}_{-i}(Z_i = 0, \bm{Z}_{-i} = \bm{z}_{-i})
	\quad
	\text{whenever $\mathcal{C}_i(\bm{z}_{-i}) = 0$}.
\end{align}
If this condition holds, the potential outcome of unit $i$ with $D_i(1, \bm{z}_{-i}) = D_i(0, \bm{z}_{-i})$ satisfies $y_i(1, \bm{z}_{-i}) = y_i(0, \bm{z}_{-i})$, which implies Assumption \ref{as:restrict1}.

Although it is difficult to directly verify conditions \eqref{eq:sufficient1} -- \eqref{eq:sufficient3} from data alone, they have different testable implications (see Appendix F).
There would also be cases where the experimental design suggests which is more likely to hold than the others.
For example, consider the empirical analysis of the anti-conflict intervention program, where $Y_i$ is an outcome variable relating to anti-conflict norms and behaviors, $D_i$ is an indicator for participating in the program, and $Z_i$ indicates whether receiving an invitation for the program.
In this experiment, the participation in the program was not compulsory and the students without an invitation were not able to attend.
Thus, there are only compliers and never takers in this empirical example, and joining in the intervention program means that the student is a complier.
It is plausible to imagine that the never takers were unlikely to affect the participation of others, as the never takers would be those who were not interested in the program.
Thus, the sufficient condition in \eqref{eq:sufficient3} would be met here.

The causal interpretation of the LADE can be different depending on which sufficient condition the researcher considers for Assumption \ref{as:restrict1} to hold.
To see this, when $i$ is a complier, we have
\begin{align*}
	y_i(1, \bm{z}_{-i}) - y_i(0, \bm{z}_{-i})
	& = \underset{\text{effect of changing $i$'s treatment}}{Y_i(1, \bm{D}_{-i}(Z_i = 1, \bm{Z}_{-i} = \bm{z}_{-i})) - Y_i(0, \bm{D}_{-i}(Z_i = 1, \bm{Z}_{-i} = \bm{z}_{-i}))} \\
	& \quad + \underset{\text{effect of changing others' treatments through $i$'s IV}}{Y_i(0, \bm{D}_{-i}(Z_i = 1, \bm{Z}_{-i} = \bm{z}_{-i})) - Y_i(0, \bm{D}_{-i}(Z_i = 0, \bm{Z}_{-i} = \bm{z}_{-i}))}.
\end{align*}
Here, under \eqref{eq:sufficient1} or \eqref{eq:sufficient2}, the second line vanishes.
Therefore, the LADE purely captures the direct treatment effect for the compliers.
On the contrary, \eqref{eq:sufficient3} admits that the IV of a complier may affect the treatment status of others.
In this case, the second line is generally nonzero, and the LADE contains the average spillover effect for the compliers as well.

The following theorem shows that the LADE is identifiable from a Wald-type estimand.  

\begin{theorem} \label{thm:LADE}
	Under Assumptions \ref{as:independence} -- \ref{as:restrict1}, $\mathrm{LADE}_{S_n}(t) = \mathrm{ADEY}_{S_n}(t) / \mathrm{ADED}_{S_n}(t)$.
\end{theorem}

\begin{remark}[Failure of Assumption \ref{as:restrict1}]
	The Wald-type estimand does not generally admit valid causal interpretation without Assumption \ref{as:restrict1} due to potential spillovers through noncompliers (see Appendix A.3 for details).
	The assumption fails if a noncomplier's IV affects the treatment choices of other units, which further influence on her own outcome.
	In our empirical setting, this occurs when students who are not interested in the anti-conflict campaign discourage other students from participating in the program, and this further reinforces their negative attitudes toward the school climate.
\end{remark}

\begin{remark}[\cite{imai2020causal}] \label{remark:imaietal}
	The IV relevance condition, monotonicity condition, and restricted interference in Assumptions \ref{as:relevance1} -- \ref{as:restrict1} are essentially the same as those in \cite{imai2020causal}.
	Consequently, Theorem \ref{thm:LADE} follows from nearly the same argument as in Theorem 2(1) of \cite{imai2020causal}.
	However, the two papers have two notable differences in terms of the identification of ADEs.
	First, we impose the mutual independence of IVs in Assumption \ref{as:independence}, while \cite{imai2020causal} focus on a two-stage randomized experiment.
	Second and more importantly, the two papers consider different network structures and asymptotic frameworks to establish the estimability for the target parameters.
	Our paper achieves this by showing that $\mathrm{ADEY}_{S_n}(t)$ and $\mathrm{ADED}_{S_n}(t)$ can be consistently estimated under the ANI framework on a large network and certain restrictions on the denseness of the network (see Section \ref{subsec:asymptotics}).
	By contrast, \cite{imai2020causal} prove the consistency of their estimators in the setting of partial interference assuming that both the cluster size and number of clusters grow to infinity.
\end{remark}

\subsection{Average indirect effects} \label{subsec:indirect1}

\subsubsection{Intention-to-treat analysis} \label{subsec:ITT2}

The following proposition presents the causal interpretation of $\mathrm{AIEY}_{S_n}$ and $\mathrm{AIED}_{S_n}$.

\begin{proposition} \label{prop:AIE_ITT}
	Under Assumption \ref{as:independence},
	\begin{align*}
		\mathrm{AIEY}_{S_n}
		& = \frac{1}{|S_n|} \sum_{i \in S_n} \sum_{\bm{z}_{-i} \in \{ 0, 1 \}^{n-1}} \sum_{j \in \mathcal{E}_i} \{ y_j(Z_i = 1, \bm{Z}_{-i} = \bm{z}_{-i}) - y_j(Z_i = 0, \bm{Z}_{-i} = \bm{z}_{-i}) \} \pi_i(\bm{z}_{-i}), \\
		\mathrm{AIED}_{S_n}
		& = \frac{1}{|S_n|} \sum_{i \in S_n} \sum_{\bm{z}_{-i} \in \{ 0, 1 \}^{n-1}} \sum_{j \in \mathcal{E}_i} \{ D_j(Z_i = 1, \bm{Z}_{-i} = \bm{z}_{-i}) - D_j(Z_i = 0, \bm{Z}_{-i} = \bm{z}_{-i}) \} \pi_i(\bm{z}_{-i}).
	\end{align*}
\end{proposition}

In view of Proposition \ref{prop:AIE_ITT}, these estimands measure the weighted average of the effect of $i$'s IV on the sum of $j$'s ($j \in \mathcal{E}_i$) outcomes and that on the sum of $j$'s treatments with the weight equal to $\pi_i(\bm{z}_{-i})$.
Furthermore, under Assumption \ref{as:exclusion}, we can see that $\mathrm{AIEY}_{S_n}$ captures both the effect of changing $i$'s treatment through $i$'s IV on the sum of $j$'s outcomes and that of changing others' treatments (including $j$'s treatment) .
This result immediately follows from the same arguments as in the previous subsection, once we notice that
\begin{align} \label{eq:yjpotential}
	y_j(Z_i = z_i, \bm{Z}_{-i} = \bm{z}_{-i})
	= Y_j(D_i(z_i, \bm{z}_{-i}), \bm{D}_{-i}(Z_i = z_i, \bm{Z}_{-i} = \bm{z}_{-i})).
\end{align}

While these ITT estimands help infer the spillover effects even with a misspecified IEM, whether they can correctly capture the spillover effects from distant units depends on the correctness of the selected IEM.
To see this, let $T_i^*$ denote the true instrumental exposure defined by $j$'s such that $1 \le \ell_{\bm{A}}(i, j) \le K^*$ with some $K^*$.
We write $i$'s interference set based on $T_i^*$ as $\mathcal{E}_i^* \coloneqq \{ j \in N_n : E_{ij}^* = 1 \}$.
If $K < K^*$ so that $\mathcal{E}_i \subset \mathcal{E}_i^*$ for some $i \in S_n$, there are some spillover effects that are not accounted for by the AIEY under $K$.
Meanwhile, if $K \ge K^*$ so that $\mathcal{E}_i \supseteq \mathcal{E}_i^*$ for all $i \in S_n$, the AIEY defined by $K$ coincides with the one defined by $K^*$.
This result immediately follows from Proposition \ref{prop:AIE_ITT} and the fact that $y_j(Z_i = 1, \bm{Z}_{-i} = \bm{z}_{-i}) = y_j(Z_i = 0, \bm{Z}_{-i} = \bm{z}_{-i})$ for all $j \notin \mathcal{E}_i^*$.
Considering this, one might expect that choosing a large $K$ is practically desirable.
However, a large $K$ implies a strong network dependence between distant units, and might deteriorate the estimation precision (see Lemma B.7 for a related result).
Studying how to balance such a trade-off is beyond the scope of this paper.

\subsubsection{Local average indirect effect} \label{subsec:LAIE}

We introduce the following three assumptions, which are parallel to Assumptions \ref{as:relevance1}, \ref{as:monotone1}, and \ref{as:restrict1}, respectively.

\begin{assumption}[Relevance 2]\label{as:relevance2}
	$|S_n|^{-1} \sum_{i \in S_n} \bbE[ \mathcal{C}_i ] \ge c$ for a constant $c > 0$.
\end{assumption}

\begin{assumption}[Monotonicity 2]\label{as:monotone2}
	$D_i(1, \bm{z}_{-i}) \ge D_i(0, \bm{z}_{-i})$ for all $i \in S_n$ and $\bm{z}_{-i} \in \{0, 1\}^{n - 1}$ such that $\pi_i(\bm{z}_{-i}) > 0$.
\end{assumption}

\begin{assumption}[Restricted interference 2] \label{as:restrict2}
	For all $i \in S_n$, $j \in \mathcal{E}_i$, and $\bm{z}_{-i} \in \{ 0, 1 \}^{n-1}$ such that $\pi_i(\bm{z}_{-i}) > 0$, $y_j(Z_i = 1, \bm{Z}_{-i} = \bm{z}_{-i}) = y_j(Z_i = 0, \bm{Z}_{-i} = \bm{z}_{-i})$ holds whenever $D_i(1, \bm{z}_{-i}) = D_i(0, \bm{z}_{-i})$.
\end{assumption}

As in Assumption \ref{as:restrict1}, Assumption \ref{as:restrict2} restricts the interference structure, but they are different in that the latter limits (not $i$'s own but) $j$'s outcome value when $i$ is a noncomplier.
In the same manner as in the previous subsection, we can see that \eqref{eq:sufficient1} or \eqref{eq:sufficient3} fulfills Assumption \ref{as:restrict2}.
Meanwhile, \eqref{eq:sufficient2} is not sufficient for Assumption \ref{as:restrict2}.
The causal interpretation of the LAIE varies with which sufficient condition holds for Assumption \ref{as:restrict2}.
Specifically, for a complier $i$, \eqref{eq:yjpotential} leads to the following decomposition:
\begin{align*}
	& y_j(Z_i = 1, \bm{Z}_{-i} = \bm{z}_{-i}) - y_j(Z_i = 0, \bm{Z}_{-i} = \bm{z}_{-i}) \\
	& = \underset{\text{effect of changing $i$'s treatment}}{ Y_j(D_i = 1, \bm{D}_{-i} = \bm{D}_{-i}(Z_i = 1, \bm{Z}_{-i} = \bm{z}_{-i})) - Y_j(D_i = 0, \bm{D}_{-i} = \bm{D}_{-i}(Z_i = 1, \bm{Z}_{-i} = \bm{z}_{-i})) } \\
	& \quad + \underset{\text{effect of changing others' treatments through $i$'s IV}}{ Y_j(D_i = 0, \bm{D}_{-i} = \bm{D}_{-i}(Z_i = 1, \bm{Z}_{-i} = \bm{z}_{-i})) - Y_j(D_i = 0, \bm{D}_{-i} = \bm{D}_{-i}(Z_i = 0, \bm{Z}_{-i} = \bm{z}_{-i})) }.
\end{align*}
Under \eqref{eq:sufficient1}, the second line vanishes, and the LAIE captures the average effect of complier $i$'s treatment on the sum of $j$'s outcomes.
By contrast, the second line remains in the case of \eqref{eq:sufficient3}, and the LAIE recovers the sum of the average effect of complier $i$'s treatment on the others' outcomes and that of others' treatments caused by changing $i$'s IV.

The next theorem shows that the LAIE is identifiable from a Wald-type estimand.

\begin{theorem} \label{thm:LAIE1}
	Under Assumptions \ref{as:independence} and \ref{as:relevance2} -- \ref{as:restrict2}, $\mathrm{LAIE}_{S_n} = \mathrm{AIEY}_{S_n} / \mathrm{ADED}_{S_n}$.
\end{theorem}

\begin{remark}[Another Wald-type estimand]
	We can naturally think of another Wald-type estimand $\mathrm{AIEY}_{S_n} / \mathrm{AIED}_{S_n}$.
	A causal interpretation of this parameter can be derived but with a set of assumptions that appears somewhat restrictive in practice.
	See Appendix F for further discussion.
\end{remark}

\section{Estimation and Asymptotic Theory} \label{sec:estimation}

\subsection{Estimators} \label{subsec:estimator}

We consider the following data generating process (DGP):
\begin{assumption}[DGP] \label{as:dgp}
	(i) Assumption \ref{as:independence} holds.
	(ii) $\{ (Z_i, T_i) \}_{i \in S_n}$ (resp. $\{ Z_i \}_{i \in S_n}$) are identically distributed across $i \in S_n$ for estimating the parameters conditioned on $(Z_i, T_i) = (z, t)$ (resp. $Z_i = z$).
	(iii) $\mathcal{T}$ is a finite subset of $\mathbb{R}^{\dim(T)}$.
	(iv) $|S_n| \to \infty$.
\end{assumption}

We reintroduce Assumption \ref{as:independence} here for the sake of self-containedness of this section.
The identical distribution of $\{ (Z_i, T_i) \}_{i \in S_n}$ in Assumption \ref{as:dgp}(ii) can be justified by appropriately choosing IEM $T$ and sub-population $S_n$.
For example, this assumption holds when $T_i = \bm{1}\{ \sum_{j \neq i} A_{ij} Z_j > c \}$ and $S_n = \{ i \in N_n: \sum_{j \neq i} A_{ij} = \delta \}$ for some $c$ and $\delta$, provided that $\{ Z_i \}_{i \in N_n}$ are IID.
We require this assumption to construct a consistent estimator of $\Pr[Z_i = z, T_i = t]$ for those $i \in S_n$ and prove a weak dependence property of some variables in the estimation of the ADE parameters.
Similarly, the identical distribution of $\{ Z_i \}_{i \in S_n}$ in Assumption \ref{as:dgp}(ii) is used to consistently estimate $\Pr[Z_i = z]$ for those $i \in S_n$ and to analyze the dependency structure in the AIE estimation.
Assumption \ref{as:dgp}(iii) is for simplicity.

Under Assumption \ref{as:dgp}, for all $i \in S_n$, we can write $p_{S_n}(z, t) \coloneqq \Pr[Z_i = z, T_i = t]$ and $p_{S_n}(z) \coloneqq \Pr[Z_i = z]$.
We estimate these by $\hat p_{S_n}(z, t) \coloneqq |S_n|^{-1} \sum_{i \in S_n} \bm{1}\{ Z_i = z, T_i = t \}$ and $\hat p_{S_n}(z) \coloneqq |S_n|^{-1} \sum_{i \in S_n} \bm{1}\{ Z_i = z \}$.\footnote{
	If one knows the experimental design completely, $p_{S_n}(z, t)$ and $p_{S_n}(z)$ can be exactly computed without estimation, as in \cite{leung2022causal}.
	In this case, we can achieve unbiased estimation of $\bar \mu_{S_n}^Y(z, t)$ and $\bar \mu_{S_n}^Y(z; \mathcal{E})$.
	In this study, for generality, we investigate the case where they need to be estimated.
}
Then, $\bar \mu_{S_n}^Y(z, t)$ and $\bar \mu_{S_n}^Y(z; \mathcal{E})$ can be estimated by
\begin{align*}
	\hat \mu_{S_n}^Y(z, t) 
	\coloneqq \frac{1}{|S_n|} \sum_{i \in S_n} \frac{Y_i \cdot \bm{1}\{ Z_i = z, T_i = t \}}{\hat p_{S_n}(z, t)},
	\qquad 
	\hat \mu_{S_n}^{Y}(z; \mathcal{E})
	\coloneqq \frac{1}{|S_n|} \sum_{i \in S_n} \sum_{j \in \mathcal{E}_i} \frac{ Y_j \cdot \bm{1}\{ Z_i = z \} }{ \hat p_{S_n}(z) }.
\end{align*}
With these estimators, we define $\hat{\mathrm{ADEY}}_{S_n}(t) \coloneqq \hat \mu_{S_n}^Y(1, t) - \hat \mu_{S_n}^Y(0, t)$ and $\hat{\mathrm{AIEY}}_{S_n} \coloneqq \hat \mu_{S_n}^Y(1; \mathcal{E}) - \hat \mu_{S_n}^Y(0; \mathcal{E})$.
Similarly, we can obtain the estimators for the treatment receipt, namely, $\hat \mu_{S_n}^D(z, t)$, $\hat \mu_{S_n}^D(z)$, $\hat \mu_{S_n}^{D}(z; \mathcal{E})$, $\hat{\mathrm{ADED}}_{S_n}(t)$, $\hat{\mathrm{ADED}}_{S_n}$, and $\hat{\mathrm{AIED}}_{S_n}$.
Finally, we have $\hat{\mathrm{LADE}}_{S_n}(t) \coloneqq \hat{\mathrm{ADEY}}_{S_n}(t) / \hat{\mathrm{ADED}}_{S_n}(t)$ and $\hat{\mathrm{LAIE}}_{S_n} \coloneqq \hat{\mathrm{AIEY}}_{S_n} / \hat{\mathrm{ADED}}_{S_n}$.

\begin{remark}[Choice of the subpopulation] \label{remark:subpopulation}
	To estimate the parameters conditioned on $(Z_i, T_i) = (z, t)$, we need to construct $S_n$ appropriately to satisfy Assumption \ref{as:dgp}(ii).
	When all units in $N_n$ have the same degree (e.g., a ring network where every unit connects only to the two adjacent units), we can set $S_n = N_n$ (if $\{ Z_i \}_{i \in N_n}$ are IID).
	For a more general network, the assumption requires us to select $S_n$ as a proper subset of $N_n$ to ensure the homogeneity of degrees over $S_n$.
\end{remark}

\subsection{Asymptotic properties} \label{subsec:asymptotics}

\subsubsection{Average direct effects} \label{subsec:direct2}

We impose the following conditions.

\begin{assumption} [Bounded outcome]\label{as:outcome}
	There exists a constant $\bar y$ such that $|y_i(\bm{z})| \le \bar y < \infty$ for all $i \in S_n$ and $\bm{z} \in \mathcal{Z}_n$.
\end{assumption}

\begin{assumption}[Overlap]\label{as:overlap}
	There exist constants $\underline{p}, \bar p \in (0, 1)$ such that $p_{S_n}(z, t) \in [\underline{p}, \bar{p}]$ and $p_{S_n}(z) \in [\underline{p}, \bar{p}]$ for all $z \in \{ 0, 1 \}$ and a given $t \in \mathcal{T}$.
\end{assumption}

Assumption \ref{as:overlap} depends on the specification of IEM $T$, choice of sub-population $S_n$, distribution of $\bm{Z}$, and structure of network $\bm{A}$.
For example, the assumption is fulfilled for each $t \in \{ 0, 1 \}$ if $T_i = \bm{1} \{ \sum_{j \neq i} A_{ij} Z_j > 0 \}$, $S_n = \{ i \in N_n: \sum_{j \neq i} A_{ij} = \delta \}$ for some $\delta \ge 1$ is non-empty, and $Z_i \stackrel{\mathrm{IID}}{\sim} \mathrm{Bernoulli}(q)$ for all $i \in N_n$ and some fixed $q \in (0, 1)$.

For a non-negative integer $s \ge 0$, let $N_{\bm{A}}(i, s) \coloneqq \{ j \in N_n: \ell_{\bm{A}}(i, j) \le s \}$ be the set of units within $s$ distance from unit $i$; namely, unit $i$'s $s$-neighborhood.
Note that $i \in N_{\bm{A}}(i, s)$ for all $s \ge 0$.
We write the sub-vector of $\bm{z} \in \mathcal{Z}_n$ restricted on $N_{\bm{A}}(i, s)$ as $\bm{z}_{N_{\bm{A}}(i, s)} \coloneqq (z_j)_{j \in N_{\bm{A}}(i, s)}$.
Similarly, let $\bm{A}_{N_{\bm{A}}(i, s)} = (A_{kl})_{k,l \in N_{\bm{A}}(i, s)}$ denote the sub-matrix of $\bm{A}$ restricted on $N_{\bm{A}}(i, s)$.

\begin{assumption}[IEM]\label{as:exposure}
	There exists a known positive integer $K \in \mathbb{N}$ such that, for all $i \in S_n$, $\bm{A}, \bm{A}' \in \mathcal{A}_n$, and $\bm{z}, \bm{z}' \in \mathcal{Z}_n$,
	\begin{align*}
		\text{$N_{\bm{A}}(i, K) = N_{\bm{A}'}(i, K)$, $\bm{A}_{N_{\bm{A}}(i, K)} = \bm{A}'_{N_{\bm{A}'}(i, K)}$, and $\bm{z}_{N_{\bm{A}}(i, K)} = \bm{z}'_{N_{\bm{A}'}(i, K)}$}
		\; \Longrightarrow \;
		T(i, \bm{z}, \bm{A}) = T(i, \bm{z}', \bm{A}').
	\end{align*}
\end{assumption}

The assumption states that the instrumental exposure of each unit depends only on the unit's own $K$-neighborhood.
This would be a mild requirement that most practical IEMs should satisfy.

We introduce the concept of ANI, which originates from \cite{leung2022causal}.
Let $N_{\bm{A}}^c(i,s) \coloneqq N_n \setminus N_{\bm{A}}(i,s)$ denote the set of units who are more than distance $s$ away from $i$.
Writing $\bm{Z}'$ as an independent copy of $\bm{Z}$, define $\bm{Z}_i^{(s)} \coloneqq (\bm{Z}_{N_{\bm{A}}(i,s)}, \bm{Z}'_{N_{\bm{A}}^c(i,s)})$ by combining the sub-vector of $\bm{Z}$ on $N_{\bm{A}}(i,s)$ and that of $\bm{Z}'$ on $N_{\bm{A}}^c(i,s)$.
Let $\theta_{n,s}^{\mathrm{ADE}} \coloneqq \max \{ \max_{i \in S_n} \bbE | y_i(\bm{Z}) - y_i(\bm{Z}_i^{(s)}) |, \; \max_{i \in S_n} \bbE | D_i(\bm{Z}) - D_i(\bm{Z}_i^{(s)}) | \}$.
This measures the intensity of interference with units that are at least $s$ distance away.
By Assumption \ref{as:outcome}, $\theta_{n,s}^{\mathrm{ADE}}$ is uniformly bounded in $n$ and $s$.

\begin{assumption}[ANI 1] \label{as:ANI1}
	$\sup_{n \in \mathbb{N}} \theta_{n,s}^{\mathrm{ADE}} \to 0$ as $s \to \infty$.
\end{assumption}

The ANI assumption says that spillover effects from units that are sufficiently far away should be sufficiently small.
In particular, those not connected with $i$ do not affect the outcome and the treatment response of $i$.
Thus, the ANI would be reasonable in practical situations where only nearby people have major impacts on one's behavior.
In other words, the assumption precludes situations such as herding behavior or social pressure, whereby one may be influenced by arbitrary unknown others, irrespective of actual connections.
Note that the ANI is much weaker than the commonly used clustered interference assumption that requires $\theta_{n,L}^{\mathrm{ADE}}=0$ for some $L$.

Let $S_{\bm{A}}^{\partial}(i, s) \coloneqq \{ j \in S_n: \ell_{\bm{A}}(i, j) = s \}$ be the subset of $S_n$ that are exactly at distance $s$ from unit $i \in S_n$.
We denote its $k$-th sample moment as $M_{S_n}^{\partial}(s; k) \coloneqq |S_n|^{-1} \sum_{i \in S_n} | S_{\bm{A}}^{\partial}(i, s)|^k$, which measures the denseness of $\bm{A}$ restricted on $S_n$.
When $k = 1$, we write $M_{S_n}^{\partial}(s) = M_{S_n}^{\partial}(s; 1)$.
Further, letting $\lfloor \cdot \rfloor$ indicate the floor function, define
\begin{align}\label{eq:tildetheta}
	\tilde \theta_{n,s}^{\mathrm{ADE}} \coloneqq
	\begin{cases*}
		\theta_{n, \lfloor s/2 \rfloor}^{\mathrm{ADE}} & \text{for $s > 2 \max\{ K, 1 \}$} \\
		1 & \text{otherwise} 
	\end{cases*}
	.
\end{align}

\begin{assumption}[Weak dependence 1]\label{as:weakLLN1}
	(i) $\max_{1 \le s \le 2K} M_{S_n}^{\partial}(s) = O(1)$, where $K$ is as given in Assumption \ref{as:exposure}.
	(ii) $|S_n|^{-1} \sum_{s = 1}^{n - 1} M_{S_n}^{\partial}(s) \tilde \theta_{n,s}^{\mathrm{ADE}} = o(1)$.
\end{assumption}

Assumption \ref{as:weakLLN1}(i) rules out that there are a non-negligible proportion of units whose $2K$ neighborhoods may grow to infinity as $n$ increases.
If one assumes that each individual can hold only a limited number of interacting partners, Assumption \ref{as:weakLLN1}(i) is satisfied with $M_{S_n}^{\partial}(s; k) < \infty$ for all $s,k < \infty$.
However, for example, it is violated if the network is a complete graph.
Assumption \ref{as:weakLLN1}(ii) is analogous to Assumption 5 of \cite{leung2022causal} and Assumption 3.2 of \cite{kojevnikov2021limit}.
This assumption restricts the rate of convergence of $\tilde \theta_{n,s}^{\mathrm{ADE}}$.
For example, in the case of a ring network, we can see that $M_{S_n}^{\partial}(s) \le 2$ for all $s$, and Assumption \ref{as:weakLLN1}(ii) is reduced to the condition $|S_n|^{-1} \sum_{s = 1}^{n - 1} \tilde \theta_{n, s}^{\mathrm{ADE}} = o(1)$.

The following theorem establishes the consistency of the ADE estimators.

\begin{theorem}\label{thm:consistency1}
	Suppose that Assumptions \ref{as:dgp} -- \ref{as:weakLLN1} hold.
	Then, we have (i) $\hat{\mathrm{ADEY}}_{S_n}(t) - \mathrm{ADEY}_{S_n}(t) \stackrel{p}{\to} 0$ and (ii) $\hat{\mathrm{ADED}}_{S_n}(t) - \mathrm{ADED}_{S_n}(t) \stackrel{p}{\to} 0$. 
	Additionally, if Assumptions \ref{as:relevance1} -- \ref{as:restrict1} hold, we have (iii) $\hat{\mathrm{LADE}}_{S_n}(t) - \mathrm{LADE}_{S_n}(t) \stackrel{p}{\to} 0$.
\end{theorem}

\begin{remark}[Rate of convergence] \label{remark:rate}
	The convergence rates of the proposed estimators are determined by the convergence rate given in Assumption \ref{as:weakLLN1}(ii);
	see Lemma B.2.
	In particular, $\sqrt{|S_n|}$-consistency can be achieved if Assumption \ref{as:weakLLN1}(ii) is strengthened to $\sum_{s = 1}^{n - 1} M_{S_n}^{\partial}(s) \tilde \theta_{n,s}^{\mathrm{ADE}} = O(1)$.
\end{remark}

Next, we investigate the asymptotic distributions of the ADE estimators.
For each $i \in S_n$, let
\begin{align*}
	V_i^{\mathrm{ADEY}} 
	& \coloneqq W_i^Y - \frac{\bar \mu_{S_n}^Y(1, t)}{p_{S_n}(1, t)} W_i^Z + \frac{\bar \mu_{S_n}^Y(0, t)}{p_{S_n}(0, t)} W_i^{1-Z}, 
	\qquad 
	V_i^{\mathrm{ADED}}
	\coloneqq W_i^D - \frac{\bar \mu_{S_n}^D(1, t)}{p_{S_n}(1, t)} W_i^Z + \frac{\bar \mu_{S_n}^D(0, t)}{p_{S_n}(0, t)} W_i^{1-Z}, \\
	V_i^{\mathrm{LADE}}
	& \coloneqq \frac{1}{\mathrm{ADED}_{S_n}(t)} V_i^{\mathrm{ADEY}}  - \frac{\mathrm{ADEY}_{S_n}(t)}{[\mathrm{ADED}_{S_n}(t)]^2} V_i^{\mathrm{ADED}},
\end{align*}
where
\begin{align} \label{eq:W}
	\begin{split}
		W_i^Y 
		& \coloneqq Y_i \left[ \frac{\bm{1}\{ Z_i = 1, T_i = t\}}{p_{S_n}(1, t)} - \frac{\bm{1}\{ Z_i = 0, T_i = t\}}{p_{S_n}(0, t)} \right],
		\qquad 
		W_i^Z 
		\coloneqq \bm{1}\{ Z_i = 1, T_i = t \}, \\
		W_i^D 
		& \coloneqq D_i \left[ \frac{\bm{1}\{ Z_i = 1, T_i = t\}}{p_{S_n}(1, t)} - \frac{\bm{1}\{ Z_i = 0, T_i = t\}}{p_{S_n}(0, t)} \right],
		\qquad 
		W_i^{1-Z}
		\coloneqq \bm{1}\{ Z_i = 0, T_i = t \}.
	\end{split}
\end{align}
For notational simplicity, we suppress the dependence of $V$'s and $W$'s on the IEM value $t$; the same notation applies to other variables introduced below.
In the proof of the theorem presented below, we will show that the asymptotic distribution of $\hat{\mathrm{ADEY}}_{S_n}(t) - \mathrm{ADEY}_{S_n}(t)$ can be obtained by that of $|S_n|^{-1} \sum_{i \in S_n}(V_i^{\mathrm{ADEY}}  - \bbE [V_i^{\mathrm{ADEY}}])$.
Similar results hold for the other cases.
Let $(\sigma_{S_n}^\mathrm{ADEY})^2 \coloneqq \Var[ |S_n|^{-1/2} \sum_{i \in S_n} V_i^{\mathrm{ADEY}} ]$, and similarly we define $( \sigma_{S_n}^\mathrm{ADED} )^2$ and $( \sigma_{S_n}^\mathrm{LADE} )^2$.

To derive the asymptotic distributions, we employ the central limit theorem (CLT) for $\psi$-weakly dependent processes in \cite{kojevnikov2021limit} (see Definition B.1).
Under Assumptions \ref{as:dgp} -- \ref{as:ANI1}, for each $V_i = V_i^{\mathrm{ADEY}}$, $V_i^{\mathrm{ADED}}$, and $V_i^{\mathrm{LADE}}$, we show that $\{ V_i \}_{i \in S_n}$ is a $\psi$-weakly dependent process with the dependence coefficients $\{ \tilde \theta_{n,s}^{\mathrm{ADE}} \}_{s \ge 0}$.
Then, we can apply their CLT to our context with additional restrictions on the network structure. 
Let $S_{\bm{A}}(i, s) \coloneqq \{ j \in S_n : \ell_{\bm{A}}(i, j) \le s \}$ and $\Delta_{S_n}(s, m; k) \coloneqq |S_n|^{-1} \sum_{i \in S_n} \max_{j \in S^\partial_{\bm{A}}(i, s)} |S_{\bm{A}}(i, m) \setminus S_{\bm{A}}(j, s - 1)|^k$, where $S_{\bm{A}}(j, s - 1) = \varnothing$ if $s = 0$.
This is the $k$-th sample moment of the maximum number (over $j$'s at distance $s$ from $i$) of units who are within distance $m$ from $i$ but at least distance $s$ apart from $j$.
Note that $\Delta_{S_n}(s, m; k)$ increases as $m$ becomes larger, but at the same time decreases fast to zero as $s$ grows because $S_{\bm{A}}(j, s - 1)$ tends to become large quickly;
for example, if all units have approximately $L$ links, $|S_{\bm{A}}(j, s - 1)| = O(L^{s-1})$.
In addition, we define $c_{S_n}(s, m; k) \coloneqq \inf_{\alpha > 1} [\Delta_{S_n}(s, m; k \alpha)]^{\frac{1}{\alpha}} [ M_{S_n}^{\partial} (s; \alpha/(\alpha - 1) ) ]^{1 - \frac{1}{\alpha}}$.
This quantity measures the denseness of the network, which plays an important role in establishing the CLT. 

\begin{assumption}[Weak dependence 2] \label{as:weakCLT1}
	For each $\sigma_{S_n} = \sigma_{S_n}^\mathrm{ADEY}$, $\sigma_{S_n}^\mathrm{ADED}$, and $\sigma_{S_n}^\mathrm{LADE}$, there exist some positive sequence $m_n \to \infty$ and a constant $0 < \varepsilon < 1$ such that for each $k \in \{ 1, 2 \}$, (i) $|S_n|^{-k/2} \sigma_{S_n}^{-(2 + k)} \sum_{s=0}^{n-1} c_{S_n}(s, m_n; k) (\tilde \theta_{n,s}^{\mathrm{ADE}})^{1-\varepsilon} \to 0$ and (ii) $|S_n|^{k/2}\sigma_{S_n}^{-k} (\tilde \theta_{n,m_n}^{\mathrm{ADE}})^{1 - \varepsilon} \to 0$.
\end{assumption}

This corresponds to Assumption 3.4 of \cite{kojevnikov2021limit}.\footnote{
	Note that Assumption \ref{as:weakCLT1} is weaker than Assumption 3.4 of \cite{kojevnikov2021limit}. This comes from the following two facts.
	First, the $\psi$-weak dependent processes considered here are uniformly bounded by Assumptions \ref{as:outcome} and \ref{as:overlap}, while \cite{kojevnikov2021limit} only assume the existence of $4 + \varepsilon$ moments of them.
	Second, they consider a more general form of $\psi$-function than ours.
	See Assumption 2.1 of their paper and Lemma B.3.
}
Note that it restricts not only the network structure but also our choice of sub-population $S_n$.
In particular, when there exist constants $\underline{C}$ and $\bar C$ such that $0 < \underline{C} \le \sigma_{S_n} \le \bar C < \infty$ for all sufficiently large $n$, Assumption \ref{as:weakCLT1} can be reduced to (i) $|S_n|^{-k/2} \sum_{s=0}^{n-1} c_{S_n}(s, m_n; k) (\tilde \theta_{n, s}^{\mathrm{ADE}})^{1 - \varepsilon} \to 0$ and (ii) $|S_n|^{k/2} (\tilde \theta_{n, m_n}^{\mathrm{ADE}})^{1 - \varepsilon} \to 0$.

The following theorem shows that the ADE estimators are asymptotically normal.

\begin{theorem}\label{thm:normal1}
	Suppose that Assumptions \ref{as:dgp} -- \ref{as:weakCLT1} hold.
	Then, we have
	\begin{align*}
		\begin{array}{cl}
			\text{(i)}   & \displaystyle \frac{\sqrt{|S_n|} \left( \hat{\mathrm{ADEY}}_{S_n}(t) - \mathrm{ADEY}_{S_n}(t) \right) }{\sigma_{S_n}^{\mathrm{ADEY}}} \stackrel{d}{\to} \mathrm{Normal}(0, 1)\\
			\text{(ii)}  & \displaystyle \frac{\sqrt{|S_n|} \left( \hat{\mathrm{ADED}}_{S_n}(t) - \mathrm{ADED}_{S_n}(t) \right) }{\sigma_{S_n}^{\mathrm{ADED}}} \stackrel{d}{\to} \mathrm{Normal}(0, 1)
		\end{array}
	\end{align*}
	provided that $(\sigma_{S_n}^{\mathrm{ADEY}})^{-1} = O(1)$ and $(\sigma_{S_n}^{\mathrm{ADED}})^{-1} = O(1)$.
	Additionally, if Assumptions \ref{as:relevance1} -- \ref{as:restrict1} hold, we have
	\begin{align*}
		\begin{array}{cl}
			\text{(iii)}   & \displaystyle \frac{\sqrt{|S_n|} \left( \hat{\mathrm{LADE}}_{S_n}(t) - \mathrm{LADE}_{S_n}(t) \right) }{\sigma_{S_n}^{\mathrm{LADE}}} \stackrel{d}{\to} \mathrm{Normal}(0, 1),
		\end{array}
	\end{align*}
	provided that
	\begin{align} \label{eq:remainder}
		\frac{1}{\sigma_{S_n}^{\mathrm{LADE}}} = O(1),
		\qquad
		\frac{\sigma_{S_n}^{\mathrm{ADEY}} \sigma_{S_n}^{\mathrm{ADED}}}{\sqrt{|S_n|} \sigma_{S_n}^{\mathrm{LADE}}} = o(1),
		\qquad 
		\frac{(\sigma_{S_n}^{\mathrm{ADED}})^2}{\sqrt{|S_n|} \sigma_{S_n}^{\mathrm{LADE}}} = o(1).
	\end{align}
\end{theorem}

The conditions in \eqref{eq:remainder} are fairly mild, which are satisfied especially with $\sqrt{|S_n|}$-consistency.

In Appendix C, we consider inference methods based on network HAC estimation and a wild bootstrap approach.
It is shown that the HAC estimators have asymptotic biases due to the fact that we cannot estimate heterogeneous means in the asymptotic variances in Theorem \ref{thm:normal1}.
This is a well-known issue in the design-based uncertainty framework (cf. \citealp{imbens2015causal}).

\subsubsection{Average indirect effects} \label{subsec:indirect2}

Next, we focus on the estimators of $\mathrm{AIEY}_{S_n}$, $\mathrm{ADED}_{S_n}$, and $\mathrm{LAIE}_{S_n}$;
the discussion of $\mathrm{AIED}_{S_n}$ is analogous.
The asymptotic properties of these estimators can be derived in the same way as in the previous subsection, but we impose an additional condition on the denseness of the network and an ANI condition slightly different from Assumption \ref{as:ANI1}.
Let $\theta_{n,s}^{\mathrm{AIE}} \coloneqq \max \{ \max_{i \in S_n} \max_{j \in \mathcal{E}_i} \bbE | y_j(\bm{Z}) - y_j(\bm{Z}_i^{(s)}) |, \;\; \max_{i \in S_n} \bbE | D_i(\bm{Z}) - D_i(\bm{Z}_i^{(s)}) | \}$.
Here, $\bbE | y_j(\bm{Z}) - y_j(\bm{Z}_i^{(s)}) |$ measures to what extent the outcome of $j \in \mathcal{E}_i$ is affected by the IVs of units that are apart from $i$ more than $s$ distance.
Define $\tilde \theta_{n,s}^{\mathrm{AIE}}$ in the same way as in \eqref{eq:tildetheta}.

\begin{assumption}[ANI 2] \label{as:ANI2}
	$\sup_{n \in \mathbb{N}} \theta_{n,s}^{\mathrm{AIE}} \to 0$ as $s \to \infty$.
\end{assumption}

\begin{assumption}[Weak dependence 3] \label{as:weakLLN2}
	Assumption \ref{as:weakLLN1}(i) -- (ii) hold when $\tilde \theta_{n,s}^{\mathrm{ADE}}$ is replaced by $\tilde \theta_{n,s}^{\mathrm{AIE}}$.
	Additionally, (iii) $\max_{i \in S_n} |\mathcal{E}_i| = O(1)$.
\end{assumption}

Similar to Assumption \ref{as:weakLLN1}, Assumption \ref{as:weakLLN2} restricts the ``sparseness'' of the network in several ways. 
For example, Assumption \ref{as:weakLLN2}(i) and (iii) are fulfilled if the degree in the network $\bm{A}$ is uniformly bounded in $i \in N_n$ and $n \in \mathbb{N}$.
Moreover, Assumption \ref{as:weakLLN2}(ii) is also satisfied if $|S_n|^{-1} \sum_{s = 1}^{n-1} \tilde \theta_{n,s}^{\mathrm{AIE}} = o(1)$ additionally holds.
However, Assumption \ref{as:weakLLN2} rules out, for example, a small-world network, where any pair of units are connected within a short distance.

Given these assumptions, we can prove the following consistency results.

\begin{theorem}\label{thm:consistency2}
	Suppose that Assumptions \ref{as:dgp} -- \ref{as:exposure} and \ref{as:ANI2} -- \ref{as:weakLLN2} hold.
	Then, we have (i) $\hat{\mathrm{AIEY}}_{S_n} - \mathrm{AIEY}_{S_n} \stackrel{p}{\to} 0$ and (ii) $\hat{\mathrm{ADED}}_{S_n} - \mathrm{ADED}_{S_n} \stackrel{p}{\to} 0$.
	Additionally, if Assumptions \ref{as:relevance2} -- \ref{as:restrict2} hold, we have (iii) $\hat{\mathrm{LAIE}}_{S_n} - \mathrm{LAIE}_{S_n} \stackrel{p}{\to} 0$.
\end{theorem}

\begin{remark}[Denseness] \label{remark:dense}
	Assumption \ref{as:weakLLN2}(iii) requires that the size of the interference set is uniformly bounded in $i \in S_n$ and $n \in \mathbb{N}$, which plays an essential role in the proof of Theorem \ref{thm:consistency2}.
	If the network at hand is denser than that considered in Assumption \ref{as:weakLLN2}, $\mathcal{E}_i$ may grow with the sample size and our AIE estimators may not achieve the consistency; 
	see Proposition 5 of \cite{li2022random} for a related result.
	Hence, we should be cautious about the denseness of the network especially when estimating the AIE parameters.
\end{remark}

To discuss the asymptotic normality results, for each $i \in S_n$, define
\begin{align*}
	V_{\mathcal{E}_i}^{\mathrm{AIEY}}
	& \coloneqq W_{{\mathcal{E}_i}}^Y - \frac{ \bar \mu_{S_n}^Y(1; \mathcal{E}) }{ p_{S_n}(1) } W_{\mathcal{E}_i}^Z + \frac{ \bar \mu_{S_n}^Y(0; \mathcal{E}) }{ p_{S_n}(0) } W_{\mathcal{E}_i}^{1-Z},
	\qquad 
	V_{\mathcal{E}_i}^{\mathrm{ADED}}
	\coloneqq W_{{\mathcal{E}_i}}^D - \frac{ \bar \mu_{S_n}^D(1) }{ p_{S_n}(1) } W_{\mathcal{E}_i}^Z + \frac{ \bar \mu_{S_n}^D(0) }{ p_{S_n}(0) } W_{\mathcal{E}_i}^{1-Z}, \\
	V_{\mathcal{E}_i}^{\mathrm{LAIE}}
	& \coloneqq \frac{ 1 }{ \mathrm{ADED}_{S_n} } V_{\mathcal{E}_i}^{\mathrm{AIEY}} - \frac{ \mathrm{AIEY}_{S_n} }{ \mathrm{ADED}_{S_n}^2 } V_{\mathcal{E}_i}^{\mathrm{ADED}},
\end{align*}
where
\begin{align} \label{eq:AIE_W}
	\begin{split}
		W_{\mathcal{E}_i}^Y 
		& \coloneqq \sum_{j \in \mathcal{E}_i} Y_j \left[ \frac{ \bm{1}\{ Z_i = 1 \} }{ p_{S_n}(1) } - \frac{ \bm{1}\{ Z_i = 0 \} }{ p_{S_n}(0) } \right],
		\qquad 
		W_{\mathcal{E}_i}^Z
		\coloneqq \bm{1}\{ Z_i = 1 \}, \\
		W_{\mathcal{E}_i}^D 
		& \coloneqq D_i \left[ \frac{ \bm{1}\{ Z_i = 1 \} }{ p_{S_n}(1) } - \frac{ \bm{1}\{ Z_i = 0 \} }{ p_{S_n}(0) } \right], 
		\qquad 
		W_{\mathcal{E}_i}^{1-Z} 
		\coloneqq \bm{1}\{ Z_i = 0 \}.
	\end{split}
\end{align}
Let $(\sigma_{S_n}^\mathrm{AIEY})^2 \coloneqq \Var[ |S_n|^{-1/2} \sum_{i \in S_n} V_{\mathcal{E}_i}^{\mathrm{AIEY}} ]$, and define $( \sigma_{S_n}^\mathrm{LAIE} )^2$ analogously.
With an abuse of notation, let us denote $( \sigma_{S_n}^\mathrm{ADED} )^2 \coloneqq \Var[ |S_n|^{-1/2} \sum_{i \in S_n} V_{\mathcal{E}_i}^{\mathrm{ADED}} ]$.

\begin{assumption}[Weak dependence 4] \label{as:weakCLT2}
	For each $\sigma_{S_n} = \sigma_{S_n}^\mathrm{AIEY}$, $\sigma_{S_n}^\mathrm{ADED}$, and $\sigma_{S_n}^\mathrm{LAIE}$, Assumption \ref{as:weakCLT1} holds when $\tilde \theta_{n,s}^{\mathrm{ADE}}$ is replaced by $\tilde \theta_{n,s}^{\mathrm{AIE}}$.
\end{assumption}

The following theorem presents the asymptotic normality results.

\begin{theorem}\label{thm:normal2}
	Suppose that Assumptions \ref{as:dgp} -- \ref{as:exposure} and \ref{as:ANI2} -- \ref{as:weakCLT2} hold.
	Then, we have
	\begin{align*}
		\begin{array}{cl}
			\text{(i)} & \displaystyle \frac{\sqrt{|S_n|} \left( \hat{\mathrm{AIEY}}_{S_n} - \mathrm{AIEY}_{S_n} \right) }{\sigma_{S_n}^{\mathrm{AIEY}}} \stackrel{d}{\to} \mathrm{Normal}(0, 1)\\
			\text{(ii)} & \displaystyle \frac{\sqrt{|S_n|} \left( \hat{\mathrm{ADED}}_{S_n} - \mathrm{ADED}_{S_n} \right) }{\sigma_{S_n}^{\mathrm{ADED}}} \stackrel{d}{\to} \mathrm{Normal}(0, 1)
		\end{array}
	\end{align*}
	provided that $(\sigma_{S_n}^{\mathrm{AIEY}})^{-1} = O(1)$ and $(\sigma_{S_n}^{\mathrm{ADED}})^{-1} = O(1)$.
	Additionally, if Assumptions \ref{as:relevance2} -- \ref{as:restrict2} hold, we have
	\begin{align*}
		\begin{array}{cl}
			\text{(iii)} & \displaystyle \frac{\sqrt{|S_n|} \left( \hat{\mathrm{LAIE}}_{S_n} - \mathrm{LAIE}_{S_n} \right) }{\sigma_{S_n}^{\mathrm{LAIE}}} \stackrel{d}{\to} \mathrm{Normal}(0, 1),
		\end{array}
	\end{align*}
	provided that the conditions in \eqref{eq:remainder} hold when $\sigma_{S_n}^{\mathrm{ADEY}}$ and $\sigma_{S_n}^{\mathrm{LADE}}$ are replaced by $\sigma_{S_n}^{\mathrm{AIEY}}$ and $\sigma_{S_n}^{\mathrm{LAIE}}$, respectively.
\end{theorem}

\section{Numerical Illustrations} \label{sec:numerical}

\subsection{Monte Carlo simulation} \label{subsec:simulation}

We investigate the finite sample properties of our methods using a set of Monte Carlo experiments.
We conduct these experiments based on an artificial ring-shaped network and on a real students' friendship network separately.
To save space, the detailed experimental setups and simulation results are summarized in Appendix G.

The main findings are as follows:
First, regardless of whether the IEM is correctly- or mis-specified, our estimators work satisfactorily well with sufficiently small biases.
Second, the RMSE values for the LADE can be significantly large in some cases.
This is because in some situations, the estimated ADED is nearly zero, resulting in extremely large LADE estimates.
This phenomenon is not persistent in other DGPs where the ADED is sufficiently away from zero.
Third, overall, we can observe that the empirical coverage rates are reasonably close to the nominal 95\% level for both HAC and bootstrap estimators in the experiments based on the artificial network.
For the experiments based on the real network, the coverage rates tend to be slightly below the nominal level.
Considering that these confidence intervals contain non-negligible asymptotic biases, the above results should be attributable to this bias to some extent.

\subsection{Empirical illustration} \label{subsec:application}

We apply the proposed methods to the data from \citeauthor{paluck2016changing}'s (\citeyear{paluck2016changing}) field experiment on anti-conflict intervention programs at American middle schools.
During the 2012-2013 school year, the research team organized intervention meetings to help students identify common conflict behaviors in their schools and instruct them on behavioral strategies to mitigate conflicts.

The data include $n = 24,471$ students in 56 public middle schools in the state of New Jersey.
A group of students (called {\it seed-eligible students}) were non-randomly selected by the research team, and half of these students (called {\it seed students} or {\it treatment-eligible students}) were randomly invited to join the program.

The students' social networks were measured by asking them to nominate up to 10 students in their school with whom they had spent time in person or online in the past few weeks.
We construct a symmetric adjacency matrix $\bm{A}$ by treating the pair of students as friends if either student nominated the other, as in \cite{aronow2017estimating}.

In our analysis, $Z_i \in \{ 0, 1 \}$ indicates whether student $i$ received an invitation to the program (i.e., whether student $i$ was a seed student), and $D_i \in \{ 0, 1 \}$ represents the participation in the intervention program (i.e., whether student $i$ attended at least one meeting).
Let $Y_i \in \{ 0, 1 \}$ be an indicator for the wearing of a program wristband given by the treated students as a reward to students for engaging in conflict-mitigating behaviors.
This is regarded as a proxy variable of student's willingness to endorse anti-conflict norms and behaviors, and the same outcome variable is used in \cite{aronow2017estimating} and \cite{leung2022causal}.
We consider the following two IEMs:  $T_{1i} = \bm{1}\{ \sum_{j \neq i} A_{ij} Z_j > 0 \}$ and $T_{2i} = \bm{1}\{ \sum_{j \neq i} A_{ij} D_j > 0 \}$, respectively labeled as ``IEM1'' and ``IEM2''.
For the estimation of ADEs conditional on $T_i = t$, in view of Assumption \ref{as:dgp}(ii), we focus on the sub-populations $S_n(\delta) = \{ i \in N_n: i$ is seed-eligible and has $\delta$ seed-eligible friend(s) $\}$ for $\delta \in \{ 1, 2, 3 \}$.
Meanwhile, when estimating the AIE parameters, we consider the sub-population $S_n^{\ge 1} = \{ i \in N_n: i$ is seed-eligible and has at least one friend $\}$.

Panels (a) and (b) of Table \ref{table:empirical1} present the ITT estimates for the ADE with the standard errors based on the network HAC estimation using bandwidth $b_n = 2$.\footnote{
	Almost the same HAC estimates were obtained for other bandwidths $b_n \in \{1, 3\}$. 
	In addition, the wild bootstrap produced similar standard errors to those reported here.
	To save space, we omit those results.
}
Receiving an invitation has a statistically significant positive effect on the probability of wearing a wristband, which is consistent with previous findings (e.g., \citealp{aronow2017estimating}; \citealp{leung2022causal}).
For example, the estimate of $\mathrm{ADEY}_{S_n(1)}(0)$ for IEM1 indicates that receiving an invitation leads to about a nine percentage point increase in the probability of wearing a wristband for the seed-eligible students whose seed-eligible friend is not treatment-eligible.
Similarly, the ADED estimates indicate positive effects of receiving an invitation on the probability of participation, which supports the IV relevance condition in Assumption \ref{as:relevance1}.

The LADE estimates and their standard errors based on the HAC estimation are also reported in panels (a) and (b) of Table \ref{table:empirical1}.
For example, the estimate of $\mathrm{LADE}_{S_n(1)}(1)$ based on IEM1 indicates an eighteen percentage point increase in the probability of wearing a wristband for the seed-eligible students who have a treatment-eligible friend.
Importantly, the LADE estimates tend to be larger than the corresponding ITT estimates, implying that the ITT analysis might underestimate the effect of the anti-conflict intervention program.

Nonetheless, we should be cautious in interpreting the LADE estimates because the interpretation of LADE essentially depends on which sufficient condition we consider for Assumption \ref{as:restrict1}.
As discussed in Section \ref{subsec:LADE}, the sufficient condition in \eqref{eq:sufficient3} would be met here, and the LADE aggregates the direct effect from participating in the intervention program and the spillover effect from the student's own treatment eligibility.
However, since one's treatment eligibility seems to have little impact on the others' treatment choice as observed below, the LADE should mainly account for the direct effect of the intervention program.

Panel (c) of Table \ref{table:empirical1} presents the AIE estimates when we set $K = 1$ in line with IEM1 and IEM2.
The estimates of $\mathrm{AIEY}_{S_n}$ and $\mathrm{LAIE}_{S_n}$ are positive and statistically significant, indicating substantial spillover effects of one's treatment eligibility and treatment take-up on others' wristband wearing.
By contrast, the estimated $\mathrm{AIED}_{S_n}$ provides no strong evidence on such spillovers between treatment eligibility and treatment decision.

\begin{table}[t]
	\caption{Empirical results} \label{table:empirical1}
	\begin{footnotesize}
		\begin{subtable}[h]{\textwidth}
			\centering
			\caption{Average direct effects conditional on $T_{1i} = 0$ or $T_{2i} = 0$}
			\begin{tabular}{rrrrcrrcrrcrr}
				\hline
				&\multicolumn{3}{c}{\bfseries }&\multicolumn{1}{c}{\bfseries }&\multicolumn{2}{c}{\bfseries $\mathrm{ADEY}_{S_n}(0)$}&\multicolumn{1}{c}{\bfseries }&\multicolumn{2}{c}{\bfseries $\mathrm{ADED}_{S_n}(0)$}&\multicolumn{1}{c}{\bfseries }&\multicolumn{2}{c}{\bfseries $\mathrm{LADE}_{S_n}(0)$}\tabularnewline
				\cline{6-7} \cline{9-10} \cline{12-13}
				&\multicolumn{1}{c}{$S_n$}&\multicolumn{1}{c}{$|S_n|$}&\multicolumn{1}{c}{$b_n$}&\multicolumn{1}{c}{}&\multicolumn{1}{c}{Estimate}&\multicolumn{1}{c}{SE}&\multicolumn{1}{c}{}&\multicolumn{1}{c}{Estimate}&\multicolumn{1}{c}{SE}&\multicolumn{1}{c}{}&\multicolumn{1}{c}{Estimate}&\multicolumn{1}{c}{SE}\tabularnewline
				\hline
				\multicolumn{13}{l}{\underline{(i) IEM1}} \tabularnewline
				&$S_n(1)$&$1023$&$2$&&$0.089$&$0.025$&&$0.473$&$0.038$&&$0.187$&$0.050$\tabularnewline
				&$S_n(2)$&$ 702$&$2$&&$0.104$&$0.041$&&$0.380$&$0.056$&&$0.273$&$0.100$\tabularnewline
				&$S_n(3)$&$ 315$&$2$&&$0.286$&$0.089$&&$0.571$&$0.115$&&$0.500$&$0.118$\tabularnewline
				\hline
				\multicolumn{13}{l}{\underline{(ii) IEM2}} \tabularnewline
				&$S_n(1)$&$1023$&$2$&&$0.058$&$0.017$&&$0.324$&$0.031$&&$0.179$&$0.047$\tabularnewline
				&$S_n(2)$&$ 702$&$2$&&$0.046$&$0.017$&&$0.158$&$0.028$&&$0.289$&$0.097$\tabularnewline
				&$S_n(3)$&$ 315$&$2$&&$0.057$&$0.027$&&$0.125$&$0.038$&&$0.453$&$0.144$\tabularnewline
				\hline
			\end{tabular}
		\end{subtable}
		\begin{subtable}[h]{\textwidth}
			\centering
			\caption{Average direct effects conditional on $T_{1i} = 1$ or $T_{2i} = 1$}
			\begin{tabular}{rrrrcrrcrrcrr}
				\hline
				&\multicolumn{3}{c}{\bfseries }&\multicolumn{1}{c}{\bfseries }&\multicolumn{2}{c}{\bfseries $\mathrm{ADEY}_{S_n}(1)$}&\multicolumn{1}{c}{\bfseries }&\multicolumn{2}{c}{\bfseries $\mathrm{ADED}_{S_n}(1)$}&\multicolumn{1}{c}{\bfseries }&\multicolumn{2}{c}{\bfseries $\mathrm{LADE}_{S_n}(1)$}\tabularnewline
				\cline{6-7} \cline{9-10} \cline{12-13}
				&\multicolumn{1}{c}{$S_n$}&\multicolumn{1}{c}{$|S_n|$}&\multicolumn{1}{c}{$b_n$}&\multicolumn{1}{c}{}&\multicolumn{1}{c}{Estimate}&\multicolumn{1}{c}{SE}&\multicolumn{1}{c}{}&\multicolumn{1}{c}{Estimate}&\multicolumn{1}{c}{SE}&\multicolumn{1}{c}{}&\multicolumn{1}{c}{Estimate}&\multicolumn{1}{c}{SE}\tabularnewline
				\hline
				\multicolumn{13}{l}{\underline{(i) IEM1}} \tabularnewline
				&$S_n(1)$&$1023$&$2$&&$0.076$&$0.029$&&$0.415$&$0.041$&&$0.182$&$0.065$\tabularnewline
				&$S_n(2)$&$ 702$&$2$&&$0.060$&$0.024$&&$0.431$&$0.042$&&$0.140$&$0.051$\tabularnewline
				&$S_n(3)$&$ 315$&$2$&&$0.063$&$0.035$&&$0.366$&$0.057$&&$0.171$&$0.090$\tabularnewline
				\hline
				\multicolumn{13}{l}{\underline{(ii) IEM2}} \tabularnewline
				&$S_n(1)$&$1023$&$2$&&$0.187$&$0.061$&&$0.917$&$0.029$&&$0.204$&$0.066$\tabularnewline
				&$S_n(2)$&$ 702$&$2$&&$0.118$&$0.049$&&$0.908$&$0.031$&&$0.130$&$0.053$\tabularnewline
				&$S_n(3)$&$ 315$&$2$&&$0.193$&$0.075$&&$0.941$&$0.033$&&$0.205$&$0.082$\tabularnewline
				\hline
			\end{tabular}
		\end{subtable}
		\begin{subtable}[h]{\textwidth}
			\centering
			\caption{Average indirect effects with $K = 1$}
			\begin{tabular}{rrrcrrcrrcrrcrr}
				\hline
				\multicolumn{3}{c}{\bfseries }&\multicolumn{1}{c}{\bfseries }&\multicolumn{2}{c}{\bfseries $\mathrm{AIEY}_{S_n}$}&\multicolumn{1}{c}{\bfseries }&\multicolumn{2}{c}{\bfseries $\mathrm{AIED}_{S_n}$}&\multicolumn{1}{c}{\bfseries }&\multicolumn{2}{c}{\bfseries $\mathrm{ADED}_{S_n}$}&\multicolumn{1}{c}{\bfseries }&\multicolumn{2}{c}{\bfseries $\mathrm{LAIE}_{S_n}$}\tabularnewline
				\cline{5-6} \cline{8-9} \cline{11-12} \cline{14-15}
				\multicolumn{1}{c}{$S_n$}&\multicolumn{1}{c}{$|S_n|$}&\multicolumn{1}{c}{$b_n$}&\multicolumn{1}{c}{}&\multicolumn{1}{c}{Estimate}&\multicolumn{1}{c}{SE}&\multicolumn{1}{c}{}&\multicolumn{1}{c}{Estimate}&\multicolumn{1}{c}{SE}&\multicolumn{1}{c}{}&\multicolumn{1}{c}{Estimate}&\multicolumn{1}{c}{SE}&\multicolumn{1}{c}{}&\multicolumn{1}{c}{Estimate}&\multicolumn{1}{c}{SE}\tabularnewline
				\hline
				$S_n^{\ge 1}$&$2244$&$2$&&$0.113$&$0.045$&&$-0.050$&$0.028$&&$0.416$&$0.030$&&$0.270$&$0.104$\tabularnewline
				\hline
			\end{tabular}
		\end{subtable}
	\end{footnotesize}
\end{table}

\if1\blind
{
	\bigskip
	\begin{center}
		{\large\bf ACKNOWLEDGMENTS}
	\end{center}
	
	The authors are grateful to the co-editor, the associate editors, three anonymous referees, and Ryo Okui for their beneficial comments and suggestions.
	This work was supported by JSPS KAKENHI Grant Numbers 19H01473 and 20K01597.
	The data set used in this study is available through the Inter-university Consortium for Political and Social Research (\citealp{paluck2020changing}).
} \fi

\if0\blind
{
} \fi 

\bigskip
\begin{center}
	{\large\bf DISCLOSURE STATEMENT}
\end{center}

The authors report there are no competing interests to declare.


\bigskip
\begin{center}
	{\large\bf SUPPLEMENTARY MATERIAL}
\end{center}

\begin{description}
	
	\item[Supplement:] The proofs of all technical results and the other supplementary results (PDF file)
	
	\item[Replication:] The ACC form and the codes to reproduce the computational results (ZIP file)
	
\end{description}

\small 
\bibliographystyle{Chicago}
\bibliography{network.bib}

\begin{thebibliography}{}

\bibitem[\protect\citeauthoryear{Abadie, Athey, Imbens, and Wooldridge}{Abadie
  et~al.}{2020}]{abadie2020sampling}
Abadie, A., S.~Athey, G.~W. Imbens, and J.~M. Wooldridge (2020).
\newblock Sampling-based versus design-based uncertainty in regression
  analysis.
\newblock {\em Econometrica\/}~{\em 88\/}(1), 265--296.

\bibitem[\protect\citeauthoryear{Aronow, Eckles, Samii, and Zonszein}{Aronow
  et~al.}{2021}]{aronow2021spillover}
Aronow, P.~M., D.~Eckles, C.~Samii, and S.~Zonszein (2021).
\newblock Spillover effects in experimental data.
\newblock {\em Advances in Experimental Political Science\/}, 289--319.

\bibitem[\protect\citeauthoryear{Aronow and Samii}{Aronow and
  Samii}{2017}]{aronow2017estimating}
Aronow, P.~M. and C.~Samii (2017).
\newblock Estimating average causal effects under general interference, with
  application to a social network experiment.
\newblock {\em The Annals of Applied Statistics\/}~{\em 11\/}(4), 1912--1947.

\bibitem[\protect\citeauthoryear{DiTraglia, Garc{\'\i}a-Jimeno,
  O'Keeffe-O'Donovan, and S{\'a}nchez-Becerra}{DiTraglia
  et~al.}{2023}]{ditraglia2023identifying}
DiTraglia, F.~J., C.~Garc{\'\i}a-Jimeno, R.~O'Keeffe-O'Donovan, and
  A.~S{\'a}nchez-Becerra (2023).
\newblock Identifying causal effects in experiments with spillovers and
  non-compliance.
\newblock {\em Journal of Econometrics\/}~{\em 235\/}(2), 1589--1624.

\bibitem[\protect\citeauthoryear{Dupas}{Dupas}{2014}]{dupas2014short}
Dupas, P. (2014).
\newblock Short-run subsidies and long-run adoption of new health products:
  Evidence from a field experiment.
\newblock {\em Econometrica\/}~{\em 82\/}(1), 197--228.

\bibitem[\protect\citeauthoryear{Forastiere, Airoldi, and Mealli}{Forastiere
  et~al.}{2021}]{forastiere2021identification}
Forastiere, L., E.~M. Airoldi, and F.~Mealli (2021).
\newblock Identification and estimation of treatment and interference effects
  in observational studies on networks.
\newblock {\em Journal of the American Statistical Association\/}~{\em
  116\/}(534), 901--918.

\bibitem[\protect\citeauthoryear{Hong and Raudenbush}{Hong and
  Raudenbush}{2006}]{hong2006evaluating}
Hong, G. and S.~W. Raudenbush (2006).
\newblock Evaluating kindergarten retention policy: A case study of causal
  inference for multilevel observational data.
\newblock {\em Journal of the American Statistical Association\/}~{\em
  101\/}(475), 901--910.

\bibitem[\protect\citeauthoryear{Hoshino and Yanagi}{Hoshino and
  Yanagi}{2023}]{hoshino2023randomization}
Hoshino, T. and T.~Yanagi (2023).
\newblock Randomization test for the specification of interference structure.
\newblock {\em arXiv preprint arXiv:2301.05580\/}.

\bibitem[\protect\citeauthoryear{Hu, Li, and Wager}{Hu
  et~al.}{2022}]{hu2022average}
Hu, Y., S.~Li, and S.~Wager (2022).
\newblock Average direct and indirect causal effects under interference.
\newblock {\em Biometrika\/}~{\em 109\/}(4), 1165--1172.

\bibitem[\protect\citeauthoryear{Hudgens and Halloran}{Hudgens and
  Halloran}{2008}]{hudgens2008toward}
Hudgens, M.~G. and M.~E. Halloran (2008).
\newblock Toward causal inference with interference.
\newblock {\em Journal of the American Statistical Association\/}~{\em
  103\/}(482), 832--842.

\bibitem[\protect\citeauthoryear{Imai, Jiang, and Malani}{Imai
  et~al.}{2021}]{imai2020causal}
Imai, K., Z.~Jiang, and A.~Malani (2021).
\newblock Causal inference with interference and noncompliance in two-stage
  randomized experiments.
\newblock {\em Journal of the American Statistical Association\/}~{\em
  116\/}(534), 632--644.

\bibitem[\protect\citeauthoryear{Imbens and Angrist}{Imbens and
  Angrist}{1994}]{imbens1994identification}
Imbens, G.~W. and J.~D. Angrist (1994).
\newblock Identification and estimation of local average treatment effects.
\newblock {\em Econometrica\/}~{\em 62\/}(2), 467--475.

\bibitem[\protect\citeauthoryear{Imbens and Rubin}{Imbens and
  Rubin}{2015}]{imbens2015causal}
Imbens, G.~W. and D.~B. Rubin (2015).
\newblock {\em Causal inference in statistics, social, and biomedical
  sciences}.
\newblock Cambridge University Press.

\bibitem[\protect\citeauthoryear{Kang and Imbens}{Kang and
  Imbens}{2016}]{kang2016peer}
Kang, H. and G.~Imbens (2016).
\newblock Peer encouragement designs in causal inference with partial
  interference and identification of local average network effects.
\newblock {\em arXiv preprint arXiv:1609.04464\/}.

\bibitem[\protect\citeauthoryear{Kang and Keele}{Kang and
  Keele}{2018}]{kang2018spillover}
Kang, H. and L.~Keele (2018).
\newblock Spillover effects in cluster randomized trials with noncompliance.
\newblock {\em arXiv preprint arXiv:1808.06418\/}.

\bibitem[\protect\citeauthoryear{Kojevnikov}{Kojevnikov}{2021}]{kojevnikov2021bootstrap}
Kojevnikov, D. (2021).
\newblock The bootstrap for network dependent processes.
\newblock {\em arXiv preprint arXiv:2101.12312\/}.

\bibitem[\protect\citeauthoryear{Kojevnikov, Marmer, and Song}{Kojevnikov
  et~al.}{2021}]{kojevnikov2021limit}
Kojevnikov, D., V.~Marmer, and K.~Song (2021).
\newblock Limit theorems for network dependent random variables.
\newblock {\em Journal of Econometrics\/}~{\em 222\/}(2), 882--908.

\bibitem[\protect\citeauthoryear{Leung}{Leung}{2022}]{leung2022causal}
Leung, M.~P. (2022).
\newblock Causal inference under approximate neighborhood interference.
\newblock {\em Econometrica\/}~{\em 90\/}(1), 267--293.

\bibitem[\protect\citeauthoryear{Li and Wager}{Li and
  Wager}{2022}]{li2022random}
Li, S. and S.~Wager (2022).
\newblock Random graph asymptotics for treatment effect estimation under
  network interference.
\newblock {\em The Annals of Statistics\/}~{\em 50\/}(4), 2334--2358.

\bibitem[\protect\citeauthoryear{Li, Ding, Lin, Yang, and Liu}{Li
  et~al.}{2019}]{li2019randomization}
Li, X., P.~Ding, Q.~Lin, D.~Yang, and J.~S. Liu (2019).
\newblock Randomization inference for peer effects.
\newblock {\em Journal of the American Statistical Association\/}~{\em
  114\/}(528), 1651--1664.

\bibitem[\protect\citeauthoryear{Manski}{Manski}{2013}]{manski2013identification}
Manski, C.~F. (2013).
\newblock Identification of treatment response with social interactions.
\newblock {\em The Econometrics Journal\/}~{\em 16\/}(1), S1--S23.

\bibitem[\protect\citeauthoryear{Miguel and Kremer}{Miguel and
  Kremer}{2004}]{miguel2004worms}
Miguel, E. and M.~Kremer (2004).
\newblock Worms: identifying impacts on education and health in the presence of
  treatment externalities.
\newblock {\em Econometrica\/}~{\em 72\/}(1), 159--217.

\bibitem[\protect\citeauthoryear{Paluck, Shepherd, and Aronow}{Paluck
  et~al.}{2016}]{paluck2016changing}
Paluck, E.~L., H.~Shepherd, and P.~M. Aronow (2016).
\newblock Changing climates of conflict: A social network experiment in 56
  schools.
\newblock {\em Proceedings of the National Academy of Sciences\/}~{\em
  113\/}(3), 566--571.

\bibitem[\protect\citeauthoryear{Paluck, Shepherd, and Aronow}{Paluck
  et~al.}{2020}]{paluck2020changing}
Paluck, E.~L., H.~R. Shepherd, and P.~Aronow (2020).
\newblock Changing climates of conflict: A social network experiment in 56
  schools, {New Jersey}, 2012-2013.
\newblock {Inter-university Consortium for Political and Social Research
  [distributor], 2020-09-14. https://doi.org/10.3886/ICPSR37070.v2}.

\bibitem[\protect\citeauthoryear{Rubin}{Rubin}{1980}]{rubin1980discussion}
Rubin, D.~B. (1980).
\newblock Discussion of ``randomization analysis of experimental data in the
  fisher randomization test'' by {D. Basu}.
\newblock {\em Journal of the American Statistical Association\/}~{\em 75},
  591--593.

\bibitem[\protect\citeauthoryear{S{\"a}vje}{S{\"a}vje}{2023}]{savje2023causal}
S{\"a}vje, F. (2023).
\newblock Causal inference with misspecified exposure mappings: separating
  definitions and assumptions.
\newblock {\em Biometrika\/}.
\newblock Forthcoming.

\bibitem[\protect\citeauthoryear{S{\"a}vje, Aronow, and Hudgens}{S{\"a}vje
  et~al.}{2021}]{savje2021average}
S{\"a}vje, F., P.~M. Aronow, and M.~G. Hudgens (2021).
\newblock Average treatment effects in the presence of unknown interference.
\newblock {\em The Annals of Statistics\/}~{\em 49\/}(2), 673--701.

\bibitem[\protect\citeauthoryear{Sobel}{Sobel}{2006}]{sobel2006randomized}
Sobel, M.~E. (2006).
\newblock What do randomized studies of housing mobility demonstrate? causal
  inference in the face of interference.
\newblock {\em Journal of the American Statistical Association\/}~{\em
  101\/}(476), 1398--1407.

\bibitem[\protect\citeauthoryear{Vazquez-Bare}{Vazquez-Bare}{2023}]{vazquez2023causal}
Vazquez-Bare, G. (2023).
\newblock Causal spillover effects using instrumental variables.
\newblock {\em Journal of the American Statistical Association\/}~{\em
  118\/}(543), 1911--1922.

\bibitem[\protect\citeauthoryear{Zelizer}{Zelizer}{2019}]{zelizer2019position}
Zelizer, A. (2019).
\newblock Is position-taking contagious? evidence of cue-taking from two field
  experiments in a state legislature.
\newblock {\em American Political Science Review\/}~{\em 113\/}(2), 340--352.

\end{thebibliography}

\clearpage 

\renewcommand{\thepage}{S\arabic{page}}
\renewcommand{\thetable}{S\arabic{table}}
\setcounter{page}{1}
\setcounter{table}{0}

\begin{center}
	{\Large Supplementary Appendix for ``Causal Inference with Noncompliance and Unknown Interference''}
	
	\textbf{(Not for publication)}
	\bigskip
	
	{\large Tadao Hoshino$^*$ and Takahide Yanagi$^\dagger$}
	\bigskip
	
	$^*$ School of Political Science and Economics, Waseda University.
	
	$^\dagger$ Graduate School of Economics, Kyoto University.
\end{center}

\bigskip
\begin{abstract}
	Appendices \ref{app:proofs} and \ref{sec:lemma} provide the proofs of all technical results in the main text.
	In Appendix \ref{sec:inference}, we develop statistical inference methods based on network HAC estimation and a wild bootstrap approach.
	We consider the identification and estimation of average overall effects and average spillover effects in Appendices \ref{sec:overall} and \ref{sec:spillover}, respectively.
	Appendix \ref{sec:identification2} contains the additional discussion of the identification analysis developed in Section \ref{sec:identification}.
	In Appendix \ref{sec:MC}, we report the results of Monte Carlo experiments.
	Appendix \ref{sec:application2} presents the additional empirical results.
\end{abstract}

\noindent%
\vfill

\newpage
\spacingset{1.9} 


\appendix

\section{Proofs}\label{app:proofs}

\subsection{Proof of Proposition \ref{prop:ITT}}

We prove only the result for the outcome variable, and that for the treatment receipt can be shown in the same manner.
We first note that the observed outcome can be written as 
\begin{align*}
	Y_i = \sum_{z_i=0}^1 \sum_{\bm{z}_{-i} \in \{0, 1\}^{n-1}} \bm{1} \{ Z_i = z_i, \bm{Z}_{-i} = \bm{z}_{-i} \} y_i(z_i, \bm{z}_{-i}).
\end{align*}
We then observe that
\begin{align} \label{eq:muY}
	\begin{split}
		\mu_i^Y(z, t) 
		& = \sum_{\bm{z}_{-i} \in \{0, 1\}^{n-1}} y_i(z, \bm{z}_{-i}) \Pr[\bm{Z}_{-i} = \bm{z}_{-i} \mid Z_i = z, T_i = t ] \\
		& = \sum_{\bm{z}_{-i} \in \{0, 1\}^{n-1}} y_i(z, \bm{z}_{-i}) \pi_i(\bm{z}_{-i}, t),
	\end{split}
\end{align}
where the second line follows from Assumption \ref{as:independence} and the definition of $T_i = T(i, \bm{Z}_{-i}, \bm{A})$.
This equality implies the result for $\mathrm{ADEY}_{S_n}(t)$.
\qed

\subsection{Proof of Proposition \ref{prop:AIE_ITT}}

Observe that $Y_j = \sum_{z_i=0}^1 \sum_{\bm{z}_{-i} \in \{0, 1\}^{n-1}} \bm{1} \{ Z_i = z_i, \bm{Z}_{-i} = \bm{z}_{-i} \} y_j(Z_i = z_i, \bm{Z}_{-i} = \bm{z}_{-i})$.
Then, we have
\begin{align*}
	\mu_{ji}^Y(z)
	& = \sum_{\bm{z}_{-i} \in \{ 0, 1 \}^{n-1}} y_j(Z_i = z, \bm{Z}_{-i} = \bm{z}_{-i}) \Pr[\bm{Z}_{-i} = \bm{z}_{-i} \mid Z_i = z] \\
	& = \sum_{\bm{z}_{-i} \in \{ 0, 1 \}^{n-1}} y_j(Z_i = z, \bm{Z}_{-i} = \bm{z}_{-i}) \pi_i(\bm{z}_{-i}),
\end{align*}
where the second line follow from Assumption \ref{as:independence}.
The completes the proof for $\mathrm{AIEY}_{S_n}$.
The result for $\mathrm{AIED}_{S_n}$ can be shown in the same manner.
\qed 

\subsection{Proof of Theorem \ref{thm:LADE}} \label{subsec:proofLADE}

Observe that $D_i = \sum_{z_i=0}^1 \sum_{\bm{z}_{-i} \in \{0,1\}^{n-1}} \bm{1}\{Z_i = z_i, \bm{Z}_{-i} = \bm{z}_{-i} \} D_i(z_i, \bm{z}_{-i})$.
By Assumption \ref{as:independence}, it holds that   
\begin{align*}
	\mu_i^D(z, t) 
	& = \sum_{\bm{z}_{-i} \in \{0,1\}^{n-1}} D_i(z, \bm{z}_{-i}) \pi_i(\bm{z}_{-i}, t) \\
	& = \sum_{\bm{z}_{-i} \in \{0,1\}^{n-1}} \bm{1} \{ D_i(1, \bm{z}_{-i}) \neq D_i(0, \bm{z}_{-i}) \} D_i(z, \bm{z}_{-i}) \pi_i(\bm{z}_{-i}, t) \\
	& \quad + \sum_{\bm{z}_{-i} \in \{0,1\}^{n-1}} \bm{1} \{ D_i(1, \bm{z}_{-i}) = D_i(0, \bm{z}_{-i}) \} D_i(z, \bm{z}_{-i}) \pi_i(\bm{z}_{-i}, t).
\end{align*}
Thus, Assumption \ref{as:monotone1} implies that
\begin{align*}
	\mathrm{ADED}_{S_n}(t) 
	& = \frac{1}{|S_n|} \sum_{i \in S_n} [ \mu_i^D(1, t)  - \mu_i^D(0, t) ] \\
	& = \frac{1}{|S_n|} \sum_{i \in S_n} \sum_{\bm{z}_{-i} \in \{0,1\}^{n-1}} \bm{1} \{ D_i(1, \bm{z}_{-i}) \neq D_i(0, \bm{z}_{-i}) \} \{ D_i(1, \bm{z}_{-i}) - D_i(0, \bm{z}_{-i}) \} \pi_i(\bm{z}_{-i}, t) \\
	& \quad + \frac{1}{|S_n|} \sum_{i \in S_n} \sum_{\bm{z}_{-i} \in \{0,1\}^{n-1}} \bm{1} \{ D_i(1, \bm{z}_{-i}) = D_i(0, \bm{z}_{-i}) \} \{ D_i(1, \bm{z}_{-i}) - D_i(0, \bm{z}_{-i}) \} \pi_i(\bm{z}_{-i}, t) \\
	& = \frac{1}{|S_n|} \sum_{i \in S_n} \sum_{\bm{z}_{-i} \in \{0,1\}^{n-1}} \mathcal{C}_i(\bm{z}_{-i}) \pi_i(\bm{z}_{-i}, t) \\
	& = \frac{1}{|S_n|} \sum_{i \in S_n} \bbE[ \mathcal{C}_i \mid T_i = t ].
\end{align*}
In the same manner, we can show that 
\begin{align*}
	\mathrm{ADEY}_{S_n}(t) 
	& = \frac{1}{|S_n|} \sum_{i \in S_n} [\mu_i^Y(1, t)  - \mu_i^Y(0, t) ] \\
	& = \frac{1}{|S_n|} \sum_{i \in S_n} \sum_{\bm{z}_{-i} \in \{0,1\}^{n-1}} \{ y_i(1, \bm{z}_{-i}) - y_i(0, \bm{z}_{-i}) \} \bm{1} \{ D_i(1, \bm{z}_{-i}) \neq D_i(0, \bm{z}_{-i}) \} \pi_i(\bm{z}_{-i}, t) \\
	& \quad + \frac{1}{|S_n|} \sum_{i \in S_n} \sum_{\bm{z}_{-i} \in \{0,1\}^{n-1}} \{ y_i(1, \bm{z}_{-i}) - y_i(0, \bm{z}_{-i}) \} \bm{1} \{ D_i(1, \bm{z}_{-i}) = D_i(0, \bm{z}_{-i}) \} \pi_i(\bm{z}_{-i}, t) \\
	& = \frac{1}{|S_n|} \sum_{i \in S_n} \sum_{\bm{z}_{-i} \in \{0,1\}^{n-1}} \{ y_i(1, \bm{z}_{-i}) - y_i(0, \bm{z}_{-i}) \} \mathcal{C}_i(\bm{z}_{-i}) \pi_i(\bm{z}_{-i}, t),
\end{align*}
where the last line follows from Assumptions \ref{as:monotone1} and \ref{as:restrict1}.
Combining these equalities with Assumption \ref{as:relevance1}, we obtain the desired result.
\qed

\begin{remark}[Interpretation of the Wald-type estimand when Assumption \ref{as:restrict1} fails]
	In the absence of Assumption \ref{as:restrict1} (restricted interference), under Assumptions \ref{as:independence} -- \ref{as:monotone1}, it is straightforward to see from the proof of Theorem \ref{thm:LADE} that
	\begin{align*}
		\frac{ \mathrm{ADEY}_{S_n}(t) }{ \mathrm{ADED}_{S_n}(t) }
		= \sum_{i \in S_n} \sum_{\bm{z}_{-i} \in \{0,1\}^{n-1}} \{ y_i(1, \bm{z}_{-i}) - y_i(0, \bm{z}_{-i}) \} \frac{ \pi_i(\bm{z}_{-i}, t) }{ \sum_{i \in S_n} \sum_{\bm{z}_{-i} \in \{0,1\}^{n-1}} \mathcal{C}_i(\bm{z}_{-i}) \pi_i(\bm{z}_{-i}, t) }.
	\end{align*}
	The right-hand side would be more difficult to interpret than the LADE parameter because the weight function here does not generally sum to one, that is,
	\begin{align*}
		\sum_{i \in S_n} \sum_{\bm{z}_{-i} \in \{0,1\}^{n-1}} \frac{ \pi_i(\bm{z}_{-i}, t) }{ \sum_{i \in S_n} \sum_{\bm{z}_{-i} \in \{0,1\}^{n-1}} \mathcal{C}_i(\bm{z}_{-i}) \pi_i(\bm{z}_{-i}, t) }
		\neq 1 
		\qquad \text{in general.}
	\end{align*}
\end{remark}

\subsection{Proof of Theorem \ref{thm:LAIE1}}

Using Assumptions \ref{as:independence} and \ref{as:monotone2}, the same arguments as in the proof of Theorem \ref{thm:LADE} can show that
\begin{align*}
	\mathrm{ADED}_{S_n}
	= \frac{1}{|S_n|} \sum_{i \in S_n} \sum_{\bm{z}_{-i} \in \{ 0, 1\}^{n-1}} \mathcal{C}_i(\bm{z}_{-i}) \pi_i(\bm{z}_{-i})
	= \frac{1}{|S_n|} \sum_{i \in S_n} \bbE[ \mathcal{C}_i ].
\end{align*}
Furthermore, we can see that
\begin{align*}
	& \mathrm{AIEY}_{S_n} \\
	& = \frac{1}{|S_n|} \sum_{i \in S_n} \sum_{j \in \mathcal{E}_i} \bbE[ y_j(Z_i = 1, \bm{Z}_{-i}) - y_j(Z_i = 0, \bm{Z}_{-i}) ] \\
	& = \frac{1}{|S_n|} \sum_{i \in S_n} \sum_{j \in \mathcal{E}_i} \bbE[ \{ y_j(Z_i = 1, \bm{Z}_{-i}) - y_j(Z_i = 0, \bm{Z}_{-i}) \} \cdot \bm{1}\{ D_i(1, \bm{Z}_{-i}) \neq D_i(0, \bm{Z}_{-i}) \} ] \\
	& \quad + \frac{1}{|S_n|} \sum_{i \in S_n} \sum_{j \in \mathcal{E}_i} \bbE[ \{ y_j(Z_i = 1, \bm{Z}_{-i}) - y_j(Z_i = 0, \bm{Z}_{-i}) \} \cdot \bm{1}\{ D_i(1, \bm{Z}_{-i}) = D_i(0, \bm{Z}_{-i}) \} ] \\
	& = \frac{1}{|S_n|} \sum_{i \in S_n} \sum_{j \in \mathcal{E}_i} \bbE[ \{ y_j(Z_i = 1, \bm{Z}_{-i}) - y_j(Z_i = 0, \bm{Z}_{-i}) \} \mathcal{C}_i(\bm{Z}_{-i}) ] \\
	& = \frac{1}{|S_n|} \sum_{i \in S_n} \sum_{\bm{z}_{-i} \in \{ 0, 1 \}^{n-1}} \sum_{j \in \mathcal{E}_i} \{ y_j(Z_i = 1, \bm{Z}_{-i} = \bm{z}_{-i}) - y_j(Z_i = 0, \bm{Z}_{-i} = \bm{z}_{-i}) \} \mathcal{C}_i(\bm{z}_{-i}) \pi_i(\bm{z}_{-i}),
\end{align*}
where the first equality follows from Assumption \ref{as:independence} and the third equality follows from Assumptions \ref{as:monotone2} and \ref{as:restrict2}.
In conjunction with Assumption \ref{as:relevance2}, we obtain the desired result.
\qed

\subsection{Proofs of Theorems \ref{thm:consistency1} and \ref{thm:consistency2}}

The proofs are straightforward from Lemmas \ref{lem:consistency} and \ref{lem:AIE_consistency}.
\qed

\subsection{Proof of Theorem \ref{thm:normal1}}

\paragraph{Proof of result (i).}
Let 
\begin{align*}
	\check \mu_{S_n}^Y(z, t) 
	\coloneqq \frac{1}{|S_n|} \sum_{i \in S_n} \frac{Y_i \cdot \bm{1}\{ Z_i = z, T_i = t \}}{p_{S_n}(z, t)}.
\end{align*}
Observe that
\begin{align*}
	\hat \mu_{S_n}^Y(z, t) 
	& =   \frac{\check \mu_{S_n}^Y(z, t)}{\hat p_{S_n}(z, t)} p_{S_n}(z, t)\\
	& = \check \mu_{S_n}^Y(z,t) - \frac{\check \mu_{S_n}^Y(z, t) }{\hat p_{S_n}(z, t)} [\hat p_{S_n}(z, t) - p_{S_n}(z, t)] \\
	& = \check \mu_{S_n}^Y(z,t) - \frac{\bar \mu_{S_n}^Y(z, t)}{p_{S_n}(z, t)} [ \hat p_{S_n}(z, t) - p_{S_n}(z, t)] - \frac{\check \mu_{S_n}^Y(z,t) - \bar \mu_{S_n}^Y(z, t)}{\hat p_{S_n}(z, t)} [\hat p_{S_n}(z, t) - p_{S_n}(z, t)] \\
	& \quad + \bar \mu_{S_n}^Y(z, t) \left( \frac{1}{p_{S_n}(z, t)} -  \frac{1}{\hat p_{S_n}(z, t)} \right) [ \hat p_{S_n}(z, t) - p_{S_n}(z, t)] \\
	& = \check \mu_{S_n}^Y(z,t) - \frac{\bar \mu_{S_n}^Y(z, t)}{p_{S_n}(z, t)} [ \hat p_{S_n}(z, t) - p_{S_n}(z, t)] + o_P\left( \frac{1}{\sqrt{|S_n|}} \right),
\end{align*}
where the last equality holds from Lemmas \ref{lem:p1} and \ref{lem:consistency} and Assumption \ref{as:overlap}.
Using this, we have
\begin{align}\label{eq:ADEYdecomp}
	\begin{split}
		& \hat{\mathrm{ADEY}}_{S_n}(t) - \mathrm{ADEY}_{S_n}(t) \\
		& = \hat \mu_{S_n}^Y(1, t) - \hat \mu_{S_n}^Y(0, t) - \bar \mu_{S_n}^Y(1, t) + \bar \mu_{S_n}^Y(0, t) \\ 
		& = \check \mu_{S_n}^Y(1, t) - \check \mu_{S_n}^Y(0, t) - \frac{\bar \mu_{S_n}^Y(1, t)}{p_{S_n}(1, t)} [ \hat p_{S_n}(1, t) - p_{S_n}(1, t)] + \frac{\bar \mu_{S_n}^Y(0, t)}{p_{S_n}(0, t)} [ \hat p_{S_n}(0, t) - p_{S_n}(0, t)]\\ 
		& \quad - \bar \mu_{S_n}^Y(1, t) + \bar \mu_{S_n}^Y(0, t) + o_P\left( \frac{1}{\sqrt{|S_n|}} \right) \\
		& = \frac{1}{|S_n|} \sum_{i \in S_n} \left( [W_i^Y - \bbE W_i^Y] - \frac{\bar \mu_{S_n}^Y(1, t)}{p_{S_n}(1, t)}[W_i^Z - \bbE W_i^Z] + \frac{\bar \mu_{S_n}^Y(0, t)}{p_{S_n}(0, t)} [W_i^{1-Z} - \bbE W_i^{1-Z}] \right) + o_P\left( \frac{1}{\sqrt{|S_n|}} \right) \\
		& = \frac{1}{|S_n|} \sum_{i \in S_n} \left( V_i^{\mathrm{ADEY}} - \bbE [V_i^{\mathrm{ADEY}}] \right) + o_P\left( \frac{1}{\sqrt{|S_n|}} \right).
	\end{split}
\end{align}
By Lemma \ref{lem:linearpsi1}, $\{ V_i^{\mathrm{ADEY}} \}_{i \in S_n}$ is $\psi$-weakly dependent with the dependence coefficients $\{ \tilde \theta_{n, s}^{\mathrm{ADE}} \}_{s \ge 0}$.
Then, letting $\tilde G_{S_n}^{\mathrm{ADEY}} \coloneqq |S_n|^{-1/2} \sum_{i \in S_n} (V_i^{\mathrm{ADEY}} - \bbE [V_i^{\mathrm{ADEY}}]) / \sigma_{S_n}^{\mathrm{ADEY}}$, the same arguments as in the proofs of Lemmas A.2 and A.3 of \cite{kojevnikov2021limit} show that there exists a positive constant $C > 0$ such that 
\begin{align*}
	& \sup_{a \in \mathbb{R}} \left| \Pr\left( \tilde G_{S_n}^{\mathrm{ADEY}} \le a \right) - \Phi( a ) \right| \\
	& \le C \sum_{k=1}^2 \left( \sqrt{ \frac{1}{|S_n|^{k/2} (\sigma_{S_n}^{\mathrm{ADEY}})^{2+k}} \sum_{s=0}^{n-1} c_{S_n}(s, m_n; k) (\tilde \theta_{n,s}^{\mathrm{ADE}})^{1-\varepsilon}} + \frac{|S_n|^{k/2}}{(\sigma_{S_n}^{\mathrm{ADEY}})^k} (\tilde \theta_{n,s}^{\mathrm{ADE}})^{1 - \varepsilon} \right),
\end{align*}
where $\Phi$ denotes the cumulative distribution function of $\mathrm{Normal}(0, 1)$, and $m_n$ and $\varepsilon$ are as given in Assumption \ref{as:weakCLT1}.
The right-hand side converges to zero by Assumption \ref{as:weakCLT1}, implying that $\tilde G_{S_n}^{\mathrm{ADEY}} \stackrel{d}{\to} \mathrm{Normal}(0, 1)$.
Thus, we have
\begin{align*}
	\frac{\sqrt{|S_n|}\left( \hat{\mathrm{ADEY}}_{S_n}(t) - \mathrm{ADEY}_{S_n}(t) \right)}{\sigma_{S_n}^{\mathrm{ADEY}}}
	= \tilde G_{S_n}^{\mathrm{ADEY}} + o_P\left( \frac{1}{\sigma_{S_n}^{\mathrm{ADEY}}} \right)
	\stackrel{d}{\to} \mathrm{Normal}(0, 1),
\end{align*}
under the condition $(\sigma_{S_n}^{\mathrm{ADEY}})^{-1} = O(1)$.

\paragraph{Proof of result (ii).}

Result (ii) can be shown in the same manner as in result (i).

\paragraph{Proof of result (iii).}

We can observe that
\begin{align*}
	& \hat{\mathrm{LADE}}_{S_n}(t) - \mathrm{LADE}_{S_n}(t) \\
	& = \frac{1}{\mathrm{ADED}_{S_n}(t)} [\hat{\mathrm{ADEY}}_{S_n}(t) - \mathrm{ADEY}_{S_n}(t)] - \frac{\hat{\mathrm{ADEY}}_{S_n}(t)}{\hat{\mathrm{ADED}}_{S_n}(t) \mathrm{ADED}_{S_n}(t)} [\hat{\mathrm{ADED}}_{S_n}(t) - \mathrm{ADED}_{S_n}(t)] \\
	& = \frac{1}{\mathrm{ADED}_{S_n}(t)} [\hat{\mathrm{ADEY}}_{S_n}(t) - \mathrm{ADEY}_{S_n}(t)] - \frac{\mathrm{ADEY}_{S_n}(t)}{[\mathrm{ADED}_{S_n}(t)]^2} [\hat{\mathrm{ADED}}_{S_n}(t) - \mathrm{ADED}_{S_n}(t)] \\
	& \quad - \left( \frac{\hat{\mathrm{ADEY}}_{S_n}(t)}{\hat{\mathrm{ADED}}_{S_n}(t) \mathrm{ADED}_{S_n}(t)} - \frac{\mathrm{ADEY}_{S_n}(t)}{[\mathrm{ADED}_{S_n}(t)]^2}\right) [\hat{\mathrm{ADED}}_{S_n}(t) - \mathrm{ADED}_{S_n}(t)] \\
	& = \frac{1}{|S_n|}\sum_{i \in S_n}\left( \frac{1}{\mathrm{ADED}_{S_n}(t)} [V_i^\mathrm{ADEY} - \bbE V_i^\mathrm{ADEY}] - \frac{\mathrm{ADEY}_{S_n}(t)}{[\mathrm{ADED}_{S_n}(t)]^2} [V_i^{\mathrm{ADED}} - \bbE V_i^{\mathrm{ADED}}]\right) \\
	& \quad + O_P \left( \frac{\sigma_{S_n}^{\mathrm{ADEY}}\sigma_{S_n}^{\mathrm{ADED}}}{|S_n|} \right) + O_P \left( \frac{(\sigma_{S_n}^{\mathrm{ADED}})^2}{|S_n|} \right) + o_P\left( \frac{1}{\sqrt{|S_n|}}\right) \\
	& = \frac{1}{|S_n|} \sum_{i \in S_n} \left( V_i^{\mathrm{LADE}} - \bbE[V_i^{\mathrm{LADE}}] \right) + O_P \left( \frac{\sigma_{S_n}^{\mathrm{ADEY}}\sigma_{S_n}^{\mathrm{ADED}}}{|S_n|} \right) + O_P \left( \frac{(\sigma_{S_n}^{\mathrm{ADED}})^2}{|S_n|} \right) + o_P\left( \frac{1}{\sqrt{|S_n|}}\right),
\end{align*}
where the third line follows from Assumption \ref{as:relevance1} and results (i) -- (ii).
Here, $V_i^{\mathrm{LADE}}$ is uniformly bounded by Assumptions \ref{as:relevance1}, \ref{as:outcome}, and \ref{as:overlap}, and Lemma \ref{lem:linearpsi1} implies that $\{ V_i^{\mathrm{LADE}} \}_{i \in S_n}$ is $\psi$-weakly dependent with the dependence coefficients $\{ \tilde \theta_{n, s}^{\mathrm{ADE}} \}_{s \ge 0}$.
Then, letting  $\tilde G_{S_n}^{\mathrm{LADE}} \coloneqq |S_n|^{-1/2} \sum_{i \in S_n} (V_i^{\mathrm{LADE}} - \bbE[ V_i^{\mathrm{LADE}} ]) / \sigma_{S_n}^{\mathrm{LADE}}$, the same arguments as in the proof of result (i) show that $\tilde G_{S_n}^{\mathrm{LADE}} \stackrel{d}{\to} \mathrm{Normal}(0, 1)$.
Thus, in conjunction with \eqref{eq:remainder}, we obtain
\begin{align*}
	& \frac{\sqrt{|S_n|} \left( \hat{\mathrm{LADE}}_{S_n}(t) - \mathrm{LADE}_{S_n}(t) \right)}{\sigma_{S_n}^{\mathrm{LADE}}} \\
	& = \tilde G_{S_n}^{\mathrm{LADE}} + O_P \left( \frac{\sigma_{S_n}^{\mathrm{ADEY}}\sigma_{S_n}^{\mathrm{ADED}}}{\sqrt{|S_n|}\sigma_{S_n}^{\mathrm{LADE}}} \right) + O_P \left( \frac{(\sigma_{S_n}^{\mathrm{ADED}})^2}{\sqrt{|S_n|}\sigma_{S_n}^{\mathrm{LADE}}} \right) +  o_P\left( \frac{1}{\sigma_{S_n}^{\mathrm{LADE}}} \right)\\
	& \stackrel{d}{\to} \mathrm{Normal}(0, 1).
\end{align*}
\qed

\subsection{Proof of Theorem \ref{thm:normal2}}

\paragraph{Proof of result (i).}

Let
\begin{align*}
	\check \mu_{S_n}^Y(z; \mathcal{E}) 
	\coloneqq \frac{1}{|S_n|} \sum_{i \in S_n} \sum_{j \in \mathcal{E}_i} \frac{ Y_j \cdot \bm{1}\{ Z_i = z \} }{ p_{S_n}(z) }.
\end{align*}
Using Lemmas \ref{lem:p2} and \ref{lem:AIE_consistency} and Assumptions \ref{as:overlap} and \ref{as:weakLLN2}, the same arguments as in the proof of Theorem \ref{thm:normal1}(i) can show that 
\begin{align*}
	\hat \mu_{S_n}^Y(z; \mathcal{E}) 
	& = \check \mu_{S_n}^Y(z; \mathcal{E}) - \frac{ \bar \mu_{S_n}^Y(z; \mathcal{E}) }{ p_{S_n}(z) } [ \hat p_{S_n}(z) - p_{S_n}(z) ] + o_P\left( \frac{1}{\sqrt{|S_n|}} \right).
\end{align*}
Then, in the same manner as in the proof of Theorem \ref{thm:normal1}(i), we can obtain
\begin{align*}
	\hat{\mathrm{AIEY}}_{S_n} - \mathrm{AIEY}_{S_n}
	& = \frac{1}{|S_n|} \sum_{i \in S_n} \left( V_{\mathcal{E}_i}^{\mathrm{AIEY}} - \bbE [ V_{\mathcal{E}_i}^{\mathrm{AIEY}} ] \right) + o_P\left( \frac{1}{\sqrt{|S_n|}} \right).
\end{align*}
By Lemma \ref{lem:linearpsi2}, $\{ V_{\mathcal{E}_i}^{\mathrm{AIEY}} \}_{i \in S_n}$ is $\psi$-weakly dependent with the dependence coefficients $\{ \tilde \theta_{n,s}^{\mathrm{AIE}} \}_{s \ge 0}$.
Hence, the same arguments as in the proof of Theorem \ref{thm:normal1}(i) can lead to the desired result.

\paragraph{Proofs of results (ii) and (iii).}

The proofs are similar to those of Theorem \ref{thm:normal1}(ii) and (iii) and thus omitted.
\qed

\section{Lemmas} \label{sec:lemma}

\subsection{Lemmas for Theorem \ref{thm:consistency1}}

\begin{lemma} \label{lem:p1}
	Suppose that Assumptions \ref{as:dgp}, \ref{as:exposure}, and \ref{as:weakLLN1}(i) hold.
	Then, we have
	\begin{align*}
		\hat p_{S_n}(z, t) - p_{S_n}(z, t) = O_P\left( \frac{1}{\sqrt{|S_n|}} \right)
	\end{align*}
	for all $z \in \{ 0, 1 \}$ and $t \in \mathcal{T}$.
\end{lemma}

\begin{proof}
	By Assumption \ref{as:dgp}(ii), $\bbE[\hat p_{S_n}(z, t)] = p_{S_n}(z, t)$, and thus it suffices to show that $\Var[\hat p_{S_n}(z, t)] = O(|S_n|^{-1})$.
	Observe that
	\begin{align*}
		& \Var[\hat p_{S_n}(z, t)] \\
		& = \frac{1}{|S_n|^2} \sum_{i \in S_n} \Var[\bm{1}\{ Z_i = z, T_i = t \}] + \frac{1}{|S_n|^2} \sum_{i \in S_n} \sum_{j \in S_n \setminus \{ i \}} \Cov[\bm{1}\{ Z_i = z, T_i = t \}, \bm{1}\{ Z_j = z, T_j = t \}] \\
		& = O\left( \frac{1}{|S_n|} \right) + \frac{1}{|S_n|^2} \sum_{i \in S_n} \sum_{j \in S_n } \sum_{s \ge 1} \bm{1}\{\ell_{\bm{A}}(i, j) = s\} \Cov[\bm{1}\{ Z_i = z, T_i = t \}, \bm{1}\{ Z_j = z, T_j = t \}] \\
		& = O\left( \frac{1}{|S_n|} \right) + \frac{1}{|S_n|^2} \sum_{i \in S_n} \sum_{j \in S_n } \sum_{s = 1}^{2K} \bm{1}\{\ell_{\bm{A}}(i, j) = s\} \Cov[\bm{1}\{ Z_i = z, T_i = t \}, \bm{1}\{ Z_j = z, T_j = t \}],
	\end{align*}
	where the last equality follows from the fact that, for any $i, j \in S_n$ such that $\ell_{\bm{A}}(i, j) > 2K$, $(Z_i, T_i)$ is independent of $(Z_j, T_j)$ by Assumptions \ref{as:dgp}(i) and \ref{as:exposure}.
	By the Cauchy--Schwarz inequality, the second term of the last line is bounded above by $|S_n|^{-1} \sum_{s=1}^{2K} M_{S_n}^{\partial}(s)$ which is $O(|S_n|^{-1})$ by Assumption \ref{as:weakLLN1}(i).
	This completes the proof.
\end{proof}

\begin{lemma} \label{lem:consistency}
	Suppose that Assumptions \ref{as:dgp} -- \ref{as:weakLLN1} hold.
	Then, we have
	\begin{align*}
		\renewcommand{\arraystretch}{2}
		\begin{array}{cl}
			\text{(i)}   & \hat \mu_{S_n}^Y(z, t) - \bar \mu_{S_n}^Y(z, t) = o_P(1), \\
			\text{(ii)}  & \hat \mu_{S_n}^D(z, t) - \bar \mu_{S_n}^D(z, t) = o_P(1), \\
		\end{array}	
	\end{align*}
	for all $z \in \{ 0, 1 \}$ and a given $t \in \mathcal{T}$.
	Further, $\sqrt{|S_n|}$-consistency is achieved if Assumption \ref{as:weakLLN1}(ii) is strengthened to the condition in Remark \ref{remark:rate}.
\end{lemma}

\begin{proof}
	We prove only the first result since the second one can be shown in the same way.
	Observe that
	\begin{align*}
		\hat \mu_{S_n}^Y(z, t)
		& = \check \mu_{S_n}^Y(z, t) - \frac{\check \mu_{S_n}^Y(z, t)}{\hat p_{S_n}(z, t)} [\hat p_{S_n}(z, t) - p_{S_n}(z, t)] \\
		& = \check \mu_{S_n}^Y(z, t) + O_P \left( \frac{1}{\sqrt{|S_n|}} \right),
	\end{align*}
	by Lemma \ref{lem:p1} and Assumptions \ref{as:outcome} and \ref{as:overlap}.
	Here, it is easy to see that $\bbE[ \check \mu_{S_n}^Y(z, t) ] = \bar \mu_{S_n}^Y(z, t)$.
	Further, letting $Q_i^Y \coloneqq Y_i \cdot \bm{1}\{ Z_i = z, T_i = t \} / p_{S_n}(z, t)$, we can see that
	\begin{align*}
		\Var\left[ \check \mu_{S_n}^Y(z, t) \right]
		& = \frac{1}{|S_n|^2} \sum_{i \in S_n} \Var[Q_i^Y] + \frac{1}{|S_n|^2} \sum_{i \in S_n} \sum_{j \in S_n \setminus \{ i \}} \Cov[Q_i^Y, Q_j^Y] \\
		& = O\left( \frac{1}{|S_n|} \right) + \frac{1}{|S_n|^2} \sum_{s = 1}^{n - 1}\sum_{i \in S_n} \sum_{j \in S_n }  \bm{1}\{ \ell_{\bm{A}}(i, j) = s \} \Cov[Q_i^Y, Q_j^Y].
	\end{align*}
	Using Assumptions \ref{as:dgp} -- \ref{as:ANI1}, similar arguments to the proof of Theorem 2 of \cite{leung2022causal} can show that the second term in the last line is bounded above by $C |S_n|^{-1} \sum_{s = 1}^{n - 1} M_{S_n}^{\partial}(s) \tilde \theta_{n,s}^{\mathrm{ADE}}$ for some positive constant $C$.
	Thus, we obtain the desired result by Assumption \ref{as:weakLLN1}(ii) or the condition in Remark \ref{remark:rate} and Chebyshev's inequality.
\end{proof}

\subsection{Lemmas for Theorem \ref{thm:normal1}}

For completeness, we define $\psi$-dependence in line with Definition 2.2 of \cite{kojevnikov2021limit}.
For $d \in \mathbb{N}$, let $\mathcal{L}_d$ be the set of real-valued bounded Lipschitz functions on $\mathbb{R}^d$:
\begin{align*}
	\mathcal{L}_d 
	\coloneqq \{ f:\mathbb{R}^d \to \mathbb{R}: \; \| f \|_{\infty} < \infty, \; \Lip(f) < \infty \},    
\end{align*}
where $\| f \|_{\infty} \coloneqq \sup_{x \in \mathbb{R}^d} |f(x)|$ and $\Lip(f)$ indicates the Lipschitz constant of $f$ (with respect to the Euclidean norm).
We write the distance between subsets $H, H' \subset S_n$ by $\ell_{\bm{A}}(H, H') \coloneqq \min \{ \ell_{\bm{A}}(i,j): i \in H, j \in H' \}$.
For $h, h' \in \mathbb{N}$, denote the collection of pairs $(H, H')$ whose sizes are $h$ and $h'$, respectively, with distance at least $s$ as 
\begin{align*}
	\mathcal{P}_{S_n}(h, h', s) \coloneqq \{ (H, H'): H, H' \subset S_n, |H| = h, |H'| = h', \ell_{\bm{A}}(H, H') \ge s \}.
\end{align*}
For a generic random vector $\bm{W}_{n,i} \in \mathbb{R}^v$, let $\bm{W}_{n, H} = (\bm{W}_{n,i})_{i \in H}$ and $\bm{W}_{n, H'} = (\bm{W}_{n,i})_{i \in H'}$.

\begin{definition}[$\psi$-dependence]\label{def:psi}
	A triangular array $\{ \bm{W}_{n,i} \}_{i \in S_n}$ is called {\it $\psi$-dependent}, if for each $n \in \mathbb{N}$, there exist a sequence of uniformly bounded constants $\{ \tilde \theta_{n,s} \}_{s \ge 0}$ with $\tilde \theta_{n,0} = 1$ and a collection of nonrandom functions $\{ \psi_{h,h'}\}_{h, h' \in \mathbb{N}}$, where $\psi_{h, h'}: \mathcal{L}_{hv} \times \mathcal{L}_{h'v} \to [0, \infty)$, such that for all $s > 0$, $(H, H') \in \mathcal{P}_{S_n}(h, h', s)$, $f \in \mathcal{L}_{hv}$, and $f' \in \mathcal{L}_{h'v}$,
	\begin{align*}
		| \Cov[ f(\bm{W}_{n,H}), f'(\bm{W}_{n,H'}) ] | \le \psi_{h,h'}(f, f') \tilde \theta_{n,s}.
	\end{align*}
	The sequence $\{ \tilde \theta_{n,s} \}_{s \ge 0}$ is called the {\it dependence coefficients} of $\{ \bm{W}_{n,i} \}_{i \in S_n}$.
	Further, if $\sup_{n \in \mathbb{N}} \tilde \theta_{n,s} \to 0$ as $s \to \infty$, we say that $\{ \bm{W}_{n,i} \}_{i \in S_n}$ is {\it $\psi$-weakly dependent}.
\end{definition}

Denote $\bm{W}_i \coloneqq (W_i^Y, W_i^D, W_i^Z, W_i^{1-Z})$, whose elements are as defined in \eqref{eq:W}.
For a subset $H \subset S_n$ with $|H| = h$, we write $\bm{W}_H = (\bm{W}_i)_{i \in H}$.

\begin{lemma}\label{lem:psi1}
	Under Assumptions \ref{as:dgp} -- \ref{as:ANI1}, the triangular array $\{ \bm{W}_i \}_{i \in S_n}$ is $\psi$-weakly dependent with the dependence coefficients $\{ \tilde \theta_{n,s}^{\mathrm{ADE}} \}_{s \ge 0}$ defined by \eqref{eq:tildetheta} and 
	\begin{align*}
		\psi_{h, h'}(f, f') = C [ \| f \|_{\infty} \| f' \|_{\infty} + h \| f' \|_{\infty} \Lip(f) + h' \| f \|_{\infty} \Lip(f') ],
		\qquad 
		\forall h, h' \in \mathbb{N}, f \in \mathcal{L}_{4h}, f' \in \mathcal{L}_{4h'},
	\end{align*}
	with some positive constant $C$.
\end{lemma}

\begin{proof}
	Consider arbitrary $n, h, h' \in \mathbb{N}$, $s > 0$, $(H, H') \in \mathcal{P}_{S_n}(h, h'; s)$, $f \in \mathcal{L}_{4h}$, and $f' \in \mathcal{L}_{4h'}$.
	Let $\xi \coloneqq f(\bm{W}_H)$ and $\zeta \coloneqq f'(\bm{W}_{H'})$.
	Consider two independent copies of $\bm{Z}$, say $\bm{Z}'$ and $\bm{Z}''$.
	For $i \in H$ and $j \in H'$, define $\bm{Z}_i^{(s, \xi)} \coloneqq (\bm{Z}_{N_{\bm{A}}(i,s)}, \bm{Z}'_{N_{\bm{A}}^c(i,s)})$, $\bm{Z}_j^{(s, \zeta)} \coloneqq (\bm{Z}_{N_{\bm{A}}(j,s)}, \bm{Z}''_{N_{\bm{A}}^c(j,s)})$, and
	\begin{align*}
		W_i^{Y, (s, \xi)} &\coloneqq y_i(\bm{Z}_i^{(s, \xi)}) \left[ \frac{\bm{1}\{ Z_i = 1, T(i, \bm{Z}_i^{(s, \xi)}, \bm{A}) = t \}}{p_{S_n}(1, t)} - \frac{\bm{1}\{ Z_i = 0, T(i, \bm{Z}_i^{(s, \xi)}, \bm{A}) = t \}}{p_{S_n}(0, t)} \right], \\
		W_j^{Y, (s, \zeta)} &\coloneqq y_j(\bm{Z}_j^{(s, \zeta)}) \left[ \frac{\bm{1}\{ Z_j = 1, T(j, \bm{Z}_j^{(s, \zeta)}, \bm{A}) = t \}}{p_{S_n}(1, t)} - \frac{\bm{1}\{ Z_j = 0, T(j, \bm{Z}_j^{(s, \zeta)}, \bm{A}) = t \}}{p_{S_n}(0, t)} \right].
	\end{align*}
	We similarly define $W_i^{D, (s, \xi)}$, $W_j^{D, (s, \zeta)}$, $W_i^{Z, (s, \xi)}$, $W_j^{Z, (s, \zeta)}$, $W_i^{1-Z, (s, \xi)}$, and $W_j^{1-Z, (s, \zeta)}$,
	and let
	\begin{alignat*}{3}
		& \bm{W}_i^{(s, \xi)} \coloneqq (W_i^{Y, (s, \xi)}, W_i^{D, (s, \xi)}, W_i^{Z, (s, \xi)}, W_i^{1-Z, (s, \xi)}), & \qquad
		& \bm{W}_{H}^{(s, \xi)} \coloneqq (\bm{W}_i^{(s, \xi)})_{i \in H}, & \qquad 
		& \xi^{(s)} \coloneqq f(\bm{W}_H^{(s, \xi)}), \\
		& \bm{W}_j^{(s, \zeta)} \coloneqq (W_j^{Y, (s, \zeta)}, W_j^{D, (s, \zeta)}, W_j^{Z, (s, \zeta)}, W_j^{1-Z, (s, \zeta)}), & \qquad
		& \bm{W}_{H'}^{(s, \zeta)} \coloneqq (\bm{W}_j^{(s, \zeta)})_{j \in H'}, & \qquad 
		& \zeta^{(s)} \coloneqq f'(\bm{W}_{H'}^{(s, \zeta)}).
	\end{alignat*}
	
	Since $f$ and $f'$ are bounded functions,
	\begin{align*}
		|\Cov(\xi, \zeta)|
		& = |\Cov(\xi, \zeta)| \cdot \bm{1}\{ s \le 2 \max\{ K, 1\} \}  + |\Cov(\xi, \zeta)| \cdot \bm{1}\{ s > 2 \max\{ K, 1\} \} \\
		& \le 2 \| f \|_{\infty} \| f' \|_{\infty} \cdot \bm{1}\{ s \le 2 \max\{ K, 1\} \}  + |\Cov(\xi, \zeta)| \cdot \bm{1}\{ s > 2 \max\{ K, 1\} \}.
	\end{align*}
	For the second term, recall that $\ell_{\bm{A}}(H, H') > 2 \max\{ K, 1 \}$ when $s > 2 \max\{ K, 1 \}$.
	Then, denoting $s' = \lfloor s/2 \rfloor$, Assumptions \ref{as:dgp}(i) and \ref{as:exposure} imply that $\bm{W}_H^{(s', \xi)}$ is independent of $\bm{W}_{H'}^{(s', \zeta)}$.
	From this, we have
	\begin{align*}
		|\Cov(\xi, \zeta)| 
		& \le |\Cov(\xi - \xi^{(s')}, \zeta)| + |\Cov(\xi^{(s')}, \zeta - \zeta^{(s')})| + |\Cov(\xi^{(s')}, \zeta^{(s')})| \\
		& = |\Cov(\xi - \xi^{(s')}, \zeta)| + |\Cov(\xi^{(s')}, \zeta - \zeta^{(s')})| \\
		& \le 2 \| f' \|_{\infty} \bbE|\xi - \xi^{(s')}| + 2 \| f \|_{\infty} \bbE|\zeta - \zeta^{(s')}| \\
		& \le 2 \| f' \|_{\infty} \Lip(f) \bbE \| \bm{W}_H - \bm{W}_H^{(s', \xi)} \| + 2 \| f \|_{\infty} \Lip(f') \bbE \| \bm{W}_{H'} - \bm{W}_{H'}^{(s', \zeta)} \|,
	\end{align*}
	where $\| \cdot \|$ denotes the Euclidean norm.
	Here, by Assumption \ref{as:exposure},
	\begin{align} \label{eq:WYs}
		W_i^Y - W_i^{Y, (s', \xi)}
		& = [y_i(\bm{Z}) - y_i(\bm{Z}_i^{(s', \xi)})] \left( \frac{\bm{1}\{ Z_i = 1, T(i, \bm{Z}, \bm{A}) = t \}}{p_{S_n}(1, t)} - \frac{\bm{1}\{ Z_i = 0, T(i, \bm{Z}, \bm{A}) = t \}}{p_{S_n}(0, t)} \right)
	\end{align}
	and
	\begin{align} \label{eq:WDs}
		W_i^D - W_i^{D, (s', \xi)}
		= [D_i(\bm{Z}) - D_i(\bm{Z}_i^{(s', \xi)})] \left( \frac{\bm{1}\{ Z_i = 1, T(i, \bm{Z}, \bm{A}) = t \}}{p_{S_n}(1, t)} - \frac{\bm{1}\{ Z_i = 0, T(i, \bm{Z}, \bm{A}) = t \}}{p_{S_n}(0, t)} \right).
	\end{align}
	Further, it is easy to see that $W_i^Z - W_i^{Z, (s', \xi)} = 0$ and $W_i^{1-Z} - W_i^{1-Z, (s', \xi)} = 0$ by Assumption \ref{as:exposure}.
	Thus, by Assumptions \ref{as:outcome} and \ref{as:overlap}, $\bbE \| \bm{W}_H - \bm{W}_H^{(s', \xi)} \| \le C h \theta_{n,s'}^{\mathrm{ADE}}$ for some positive constant $C$.
	In the same way, we can see that $\bbE \| \bm{W}_{H'} - \bm{W}_{H'}^{(s', \zeta)} \| \le C h' \theta_{n,s'}^{\mathrm{ADE}}$.
	In conjunction with Assumption \ref{as:ANI1}, this completes the proof.
\end{proof}

The next lemma is immediate from Lemma \ref{lem:psi1} (cf. Lemma 2.1 of \citealp{kojevnikov2021limit}).
Let $\{ \bm{c}_{n,i} \}_{i \in S_n}$ be a sequence of uniformly bounded nonrandom vectors in $\mathbb{R}^4$.

\begin{lemma}\label{lem:linearpsi1}
	Under Assumptions \ref{as:dgp} -- \ref{as:ANI1}, the triangular array $\{ \bm{c}_{n,i}^\top \bm{W}_i \}_{i \in S_n}$ is $\psi$-weakly dependent with the dependence coefficients $\{ \tilde \theta_{n,s}^{\mathrm{ADE}} \}_{s \ge 0}$ defined by \eqref{eq:tildetheta} and
	\begin{align*}
		\psi_{h, h'}(f, f') = C [ \| f \|_{\infty} \| f' \|_{\infty} + h \| f' \|_{\infty} \Lip(f) + h' \| f \|_{\infty} \Lip(f') ],
		\qquad 
		\forall h, h' \in \mathbb{N}, f \in \mathcal{L}_{h}, f' \in \mathcal{L}_{h'},
	\end{align*}
	with some positive constant $C$.
\end{lemma}

\subsection{Lemmas for Theorems \ref{thm:consistency2} and \ref{thm:normal2}}

\begin{lemma} \label{lem:p2}
	Under Assumption \ref{as:dgp}, $\hat p_{S_n}(z) - p_{S_n}(z) = O_P(|S_n|^{-1/2})$ for all $z \in \{ 0, 1 \}$.
\end{lemma}

\begin{proof}
	The desired result follows immediately from the law of large numbers.
\end{proof}

\begin{lemma} \label{lem:AIE_consistency}
	Suppose that Assumptions \ref{as:dgp} -- \ref{as:exposure} and \ref{as:ANI2} -- \ref{as:weakLLN2} hold.
	Then, we have
	\begin{align*}
		\renewcommand{\arraystretch}{2}
		\begin{array}{cl}
			\text{(i)}   & \hat \mu_{S_n}^Y(z; \mathcal{E}) - \bar \mu_{S_n}^Y(z; \mathcal{E}) = o_P(1), \\
			\text{(ii)}  & \hat \mu_{S_n}^D(z) - \bar \mu_{S_n}^D(z) = o_P(1), \\
		\end{array}	
	\end{align*}
	for all $z \in \{ 0, 1 \}$.
	Further, $\sqrt{|S_n|}$-consistency is achieved if the condition in Remark \ref{remark:rate} holds when $\tilde \theta_{n,s}^{\mathrm{ADE}}$ is replaced by $\tilde \theta_{n,s}^{\mathrm{AIE}}$.
\end{lemma}

\begin{proof}
	We focus on proving the first result, since the proof of the second is almost the same. 
	Using Lemma \ref{lem:p2} and Assumptions \ref{as:outcome}, \ref{as:overlap}, and \ref{as:weakLLN2}, we have
	\begin{align*}
		\hat \mu_{S_n}^Y(z; \mathcal{E})
		& = \check \mu_{S_n}^Y(z; \mathcal{E}) - \frac{\check \mu_{S_n}^Y(z; \mathcal{E})}{\hat p_{S_n}(z)} [\hat p_{S_n}(z) - p_{S_n}(z)] \\
		& = \check \mu_{S_n}^Y(z; \mathcal{E}) + O_P \left( \frac{1}{\sqrt{|S_n|}} \right).
	\end{align*}
	Here, we can easily show that $\bbE[ \check \mu_{S_n}^Y(z; \mathcal{E}) ] = \bar \mu_{S_n}^Y(z; \mathcal{E})$.
	Further, denoting $Q_{\mathcal{E}_i}^Y \coloneqq \sum_{j \in \mathcal{E}_i} Y_j \cdot \bm{1}\{ Z_i = z \} / p_{S_n}(z)$, we have
	\begin{align*}
		\Var\left[ \check \mu_{S_n}^Y(z; \mathcal{E}) \right]
		& = \frac{1}{|S_n|^2} \sum_{i \in S_n} \Var[Q_{\mathcal{E}_i}^Y] + \frac{1}{|S_n|^2} \sum_{i_1 \in S_n} \sum_{i_2 \in S_n \setminus \{ i_1 \}} \Cov[Q_{\mathcal{E}_{i_1}}^Y, Q_{\mathcal{E}_{i_2}}^Y].
	\end{align*}
	The first term is $O(|S_n|^{-1})$, since $Q_{\mathcal{E}_i}^Y$ is uniformly bounded in $i \in S_n$ and $n \in \mathbb{N}$ by Assumptions \ref{as:outcome}, \ref{as:overlap}, and \ref{as:weakLLN2}.
	The second term can be rewritten as
	\begin{align*}
		\frac{1}{|S_n|^2} \sum_{s = 1}^{n - 1} \sum_{i_1 \in S_n} \sum_{i_2 \in S_n}  \bm{1}\{ \ell_{\bm{A}}(i_1, i_2) = s \} \Cov[Q_{\mathcal{E}_{i_1}}^Y, Q_{\mathcal{E}_{i_2}}^Y].
	\end{align*}
	Then, using Assumptions \ref{as:dgp} -- \ref{as:exposure} and \ref{as:ANI2} -- \ref{as:weakLLN2}, similar arguments to the proof of Theorem 2 of \cite{leung2022causal} can show that it is bounded above by $C |S_n|^{-1} \sum_{s = 1}^{n - 1} M_{S_n}^{\partial}(s) \tilde \theta_{n,s}^{\mathrm{AIE}}$ for some positive constant $C$.
	Thus, we obtain the desired result by Assumption \ref{as:weakLLN2} or the condition in Remark \ref{remark:rate} with replacing $\tilde \theta_{n,s}^{\mathrm{ADE}}$ by $\tilde \theta_{n,s}^{\mathrm{AIE}}$.
\end{proof}

To state the next lemma, recall the definitions of the quantities discussed in Definition \ref{def:psi}.
Let $\bm{W}_{\mathcal{E}_i} \coloneqq (W_{\mathcal{E}_i}^Y, W_{\mathcal{E}_i}^D, W_{\mathcal{E}_i}^Z, W_{\mathcal{E}_i}^{1-Z})$.
For a subset $H \subset S_n$ with $|H| = h$, we write $\bm{W}_{\mathcal{E}_H} = (\bm{W}_{\mathcal{E}_i})_{i \in H}$.

\begin{lemma} \label{lem:psi2}
	Under Assumptions \ref{as:dgp} -- \ref{as:exposure} and \ref{as:ANI2}, the triangular array $\{ \bm{W}_{\mathcal{E}_i} \}_{i \in S_n}$ is $\psi$-weakly dependent with the dependence coefficients $\{ \tilde \theta_{n,s}^{\mathrm{AIE}} \}_{s \ge 0}$ and 
	\begin{align*}
		& \psi_{h, h'}(f, f') = C \left[ \| f \|_{\infty} \| f' \|_{\infty} + h \left( \max_{i \in H} |\mathcal{E}_i| \right) \| f' \|_{\infty} \Lip(f) + h' \left( \max_{i \in H'} |\mathcal{E}_i| \right) \| f \|_{\infty} \Lip(f') \right], \\
		& \qquad \forall h, h' \in \mathbb{N}, f \in \mathcal{L}_{4h}, f' \in \mathcal{L}_{4h'},
	\end{align*}
	with some positive constant $C$.
	Additionally, if Assumption \ref{as:weakLLN2} holds, $\max_{i \in H} |\mathcal{E}_i| = O(1)$ and $\max_{i \in H'} |\mathcal{E}_i| = O(1)$.
\end{lemma}

\begin{proof}
	The desired result follows from the same arguments as in the proof of Lemma \ref{lem:psi1}, expect that \eqref{eq:WYs} and \eqref{eq:WDs} change to
	\begin{align*}
		W_{\mathcal{E}_i}^Y - W_{\mathcal{E}_i}^{Y,(s',\xi)}
		= \sum_{j \in \mathcal{E}_i} [ y_j(\bm{Z}) - y_j(\bm{Z}_i^{(s',\xi)}) ] \left[ \frac{ \bm{1}\{ Z_i = 1 \} }{ p_{S_n}(1) } - \frac{ \bm{1}\{ Z_i = 0 \} }{ p_{S_n}(0) } \right]
	\end{align*}
	and
	\begin{align*}
		W_{\mathcal{E}_i}^D - W_{\mathcal{E}_i}^{D,(s',\xi)}
		& = [ D_i(\bm{Z}) - D_i(\bm{Z}_i^{(s',\xi)}) ] \left[ \frac{ \bm{1}\{ Z_i = 1 \} }{ p_{S_n}(1) } - \frac{ \bm{1}\{ Z_i = 0 \} }{ p_{S_n}(0) } \right].
	\end{align*}
	Then, by Assumptions \ref{as:outcome} and \ref{as:overlap}, it holds that $\bbE \| \bm{W}_{\mathcal{E}_H} - \bm{W}_{\mathcal{E}_H}^{(s', \xi)} \| \le C h (\max_{i \in H} |\mathcal{E}_i| ) \theta_{n,s'}^{\mathrm{AIE}}$ for some positive constant $C$, where the variable definitions should be clear from the context.
	It is easy to see that $\max_{i \in H} |\mathcal{E}_i| \le \max_{i \in S_n} |\mathcal{E}_i| = O(1)$ by Assumption \ref{as:weakLLN2}.
	The rest of the proof is the same as the proof of Lemma \ref{lem:psi1}.
\end{proof}

Let $\{ \bm{c}_{n,i} \}_{i \in S_n}$ be a sequence of uniformly bounded nonrandom vectors in $\mathbb{R}^4$.

\begin{lemma}\label{lem:linearpsi2}
	Under Assumptions \ref{as:dgp} -- \ref{as:exposure} and \ref{as:ANI2}, the triangular array $\{ \bm{c}_{n,i}^\top \bm{W}_{\mathcal{E}_i} \}_{i \in S_n}$ is $\psi$-weakly dependent with the dependence coefficients $\{ \tilde \theta_{n,s}^{\mathrm{AIE}} \}_{s \ge 0}$ and
	\begin{align*}
		& \psi_{h, h'}(f, f') = C \left[ \| f \|_{\infty} \| f' \|_{\infty} + h \left( \max_{i \in H} |\mathcal{E}_i| \right) \| f' \|_{\infty} \Lip(f) + h' \left( \max_{i \in H'} |\mathcal{E}_i| \right) \| f \|_{\infty} \Lip(f') \right], \\
		& \qquad \forall h, h' \in \mathbb{N}, f \in \mathcal{L}_{4h}, f' \in \mathcal{L}_{4h'},
	\end{align*}
	with some positive constant $C$.
	Additionally, if Assumption \ref{as:weakLLN2} holds, $\max_{i \in H} |\mathcal{E}_i| = O(1)$ and $\max_{i \in H'} |\mathcal{E}_i| = O(1)$.
\end{lemma}


\section{Statistical Inference} \label{sec:inference}

To save space, we focus on the statistical inference on the ADE parameters.
The statistical inference on the other parameters is analogous.

\subsection{Network HAC variance estimator}\label{subsubsec:HAC}

We develop the network HAC estimator and prove its asymptotic property.
Under Assumption \ref{as:ANI1}, we can see that
\begin{align*}
	(\sigma_{S_n}^{\mathrm{ADEY}})^2 
	& = \frac{1}{|S_n|} \sum_{i \in S_n} \sum_{j \in S_n} \Cov\left[ V_i^{\mathrm{ADEY}}, V_j^{\mathrm{ADEY}} \right] \bm{1}\{ \ell_{\bm{A}}(i, j) \le n - 1 \},
\end{align*}
and analogous equalities hold for $(\sigma_{S_n}^{\mathrm{ADED}})^2$ and $(\sigma_{S_n}^{\mathrm{LADE}})^2$.
Then, the infeasible network HAC estimator of $(\sigma_{S_n}^{\mathrm{ADEY}})^2$ is given by
\begin{align*}
	(\tilde \sigma_{S_n}^{\mathrm{ADEY}})^2 
	& \coloneqq \frac{1}{|S_n|} \sum_{i \in S_n} \sum_{j \in S_n} (V_i^{\mathrm{ADEY}} - \bbE [V_i^{\mathrm{ADEY}}]) (V_j^{\mathrm{ADEY}} - \bbE [V_j^{\mathrm{ADEY}}]) \bm{1}\{ \ell_{\bm{A}}(i, j) \le b_n \},
\end{align*}
where $b_n \ge 0$ is a bandwidth parameter that grows as $n \to \infty$.
This estimator is infeasible because both $V_i^{\mathrm{ADEY}}$ and $\bbE [V_i^{\mathrm{ADEY}}]$ are unobservable to us.
For constructing feasible variance estimators, we compute
\begin{align*}
	\hat V_i^{\mathrm{ADEY}}
	& \coloneqq \hat W_i^Y - \frac{\hat \mu_{S_n}^Y(1, t)}{\hat p_{S_n}(1, t)} W_i^Z + \frac{\hat \mu_{S_n}^Y(0, t)}{\hat p_{S_n}(0, t)} W_i^{1-Z},
	\qquad
	\hat V_i^{\mathrm{ADED}}
	\coloneqq \hat W_i^D - \frac{\hat \mu_{S_n}^Y(1, t)}{\hat p_{S_n}(1, t)} W_i^Z + \frac{\hat \mu_{S_n}^Y(0, t)}{\hat p_{S_n}(0, t)} W_i^{1-Z},\\
	\hat V_i^{\mathrm{LADE}}
	& \coloneqq \frac{1}{\hat{\mathrm{ADED}}_{S_n}(t)} \hat V_i^{\mathrm{ADEY}} - \frac{\hat{\mathrm{ADEY}}_{S_n}(t)}{[\hat{\mathrm{ADED}}_{S_n}(t)]^2} \hat V_i^{\mathrm{ADED}},
\end{align*}
where
\begin{align*}
	\hat W_i^Y
	& \coloneqq Y_i \left[ \frac{\bm{1}\{ Z_i = 1, T_i = t \}}{\hat p_{S_n}(1, t)} - \frac{\bm{1}\{ Z_i = 0, T_i = t \}}{\hat p_{S_n}(0, t)} \right], \\
	\hat W_i^D
	& \coloneqq D_i \left[ \frac{\bm{1}\{ Z_i = 1, T_i = t \}}{\hat p_{S_n}(1, t)} - \frac{\bm{1}\{ Z_i = 0, T_i = t \}}{\hat p_{S_n}(0, t)} \right].
\end{align*}
Note that the sample mean of each of $\hat V_i^{\mathrm{ADEY}}$, $\hat V_i^{\mathrm{ADED}}$, and $\hat V_i^{\mathrm{LADE}}$ is zero.
Then, the feasible network HAC estimator is given by
\begin{align*}
	(\hat \sigma_{S_n}^{\mathrm{ADEY}})^2 
	& \coloneqq \frac{1}{|S_n|} \sum_{i \in S_n} \sum_{j \in S_n} \hat V_i^{\mathrm{ADEY}} \hat V_j^{\mathrm{ADEY}} \bm{1}\{ \ell_{\bm{A}}(i, j) \le b_n \},
\end{align*}
and $(\hat \sigma_{S_n}^{\mathrm{ADED}})^2$ and $(\hat \sigma_{S_n}^{\mathrm{LADE}})^2$ are analogously defined. 

Recall that $S_{\bm{A}}(i, s)$ denotes the subset of $S_n$ composed of units within $s$ distance from unit $i$.
We write its $k$-th sample moment as $M_{S_n}(s, k) \coloneqq |S_n|^{-1} \sum_{i \in S_n} |S_{\bm{A}}(i, s)|^k$.
Further, define $\mathcal{J}_{S_n}(s, b_n) \coloneqq \{ (i, j, k, l) \in S_n^4: \ell_{\bm{A}}(i, j) = s, \ell_{\bm{A}}(i, k) \le b_n, \ell_{\bm{A}}(j, l) \le b_n \}$.

\begin{assumption}[Weak dependence 3]\label{as:weakHAC}
	(i) There exists some $0 < \epsilon < 1$ such that $\sum_{s = 1}^{n-1} M_{S_n}^{\partial}(s) (\tilde \theta_{n,s}^{\mathrm{ADE}})^{1-\epsilon} = O(1)$ and $\sum_{s = 0}^{n-1} |\mathcal{J}_{S_n}(s, b_n)| (\tilde \theta_{n,s}^{\mathrm{ADE}})^{1-\epsilon} = o(|S_n|^2)$.
	(ii) $M_{S_n}(b_n, k) = o(|S_n|^{k/2})$ for each $k \in \{ 1, 2 \}$.
\end{assumption}

This assumption restricts both the network structure and the rate of divergence of $b_n$ in a similar manner to Assumption 7 of \cite{leung2022causal} and Assumption 4.1 of \cite{kojevnikov2021limit}.
The first part of Assumption \ref{as:weakHAC}(i) strengthens Assumption \ref{as:weakLLN1}(ii) and ensures $\sqrt{|S_n|}$-consistency of our estimators (see Remark \ref{remark:rate}).
The second part of Assumption \ref{as:weakHAC}(i) corresponds to Assumption 4.1(iii) of \cite{kojevnikov2021limit}.
Assumption \ref{as:weakHAC}(ii) is the same as Assumption 7(b) -- (c) of \cite{leung2022causal}.
Under these conditions, we can derive the probability limits of the infeasible oracle variance estimators, and evaluate the stochastic errors caused by replacing unobserved $V_i^{\mathrm{ADEY}}$, $V_i^{\mathrm{ADED}}$, and $V_i^{\mathrm{LADE}}$ with their estimators $\hat V_i^{\mathrm{ADEY}}$, $\hat V_i^{\mathrm{ADED}}$, and $\hat V_i^{\mathrm{LADE}}$.

Let
\begin{align*}
	B_{S_n}^{\mathrm{ADEY}} 
	& \coloneqq \frac{1}{|S_n|} \sum_{i \in S_n} \sum_{j \in S_n} \bbE[V_i^{\mathrm{ADEY}}] \bbE[V_j^{\mathrm{ADEY}}] \bm{1}\{ \ell_{\bm{A}}(i, j) \le b_n \}.
\end{align*}
We define $B_{S_n}^{\mathrm{ADED}}$ and $B_{S_n}^{\mathrm{LADE}}$ in the same manner.

\begin{theorem} \label{thm:HAC}
	Suppose that Assumptions \ref{as:dgp} -- \ref{as:ANI1} and \ref{as:weakHAC} hold.
	Then, if $b_n \to \infty$, we have (i) $(\hat \sigma_{S_n}^{\mathrm{ADEY}})^2 = (\sigma_{S_n}^{\mathrm{ADEY}})^2 + B_{S_n}^{\mathrm{ADEY}} + o_P(1)$ and (ii) $(\hat \sigma_{S_n}^{\mathrm{ADED}})^2 = (\sigma_{S_n}^{\mathrm{ADED}})^2 + B_{S_n}^{\mathrm{ADED}} + o_P(1)$.
	Additionally, if Assumptions \ref{as:relevance1} -- \ref{as:restrict1} hold, we have (iii) $(\hat \sigma_{S_n}^{\mathrm{LADE}})^2 = (\sigma_{S_n}^{\mathrm{LADE}})^2 + B_{S_n}^{\mathrm{LADE}} + o_P(1)$.
	
\end{theorem}

The proof is provided in the end of this section. 
In the proof, we show that $(\hat \sigma_{S_n}^{\mathrm{ADEY}})^2 = (\tilde \sigma_{S_n}^{\mathrm{ADEY}})^2 + B_{S_n}^{\mathrm{ADEY}} + o_P(1)$, where the infeasible estimator $(\tilde \sigma_{S_n}^{\mathrm{ADEY}})^2$ is consistent for $(\sigma_{S_n}^{\mathrm{ADEY}})^2$.
There is an asymptotic bias term $B_{S_n}^{\mathrm{ADEY}}$ due to the fact that we cannot estimate the heterogeneous mean $\bbE[ V_i^{\mathrm{ADEY}} ]$.
It is well known in the design-based uncertainty framework that heterogeneous means cause standard variance estimators to have asymptotic biases (cf. \citealp{imbens2015causal}).

\subsection{Wild bootstrap}\label{subsubsec:Bootstrap}

Alternatively to the HAC estimator, we can consider using a network-dependent bootstrap method.
Here, we particularly focus on \citeauthor{kojevnikov2021bootstrap}'s (\citeyear{kojevnikov2021bootstrap}) wild bootstrap approach.

For exposition, we focus only on constructing a confidence interval for $\mathrm{ADEY}_{S_n}(t)$.
The following procedure can be applied to the other parameters as well.
As shown in \eqref{eq:ADEYdecomp} in Appendix \ref{app:proofs}, we have
\begin{align*}
	\sqrt{|S_n|}\left( \hat{\mathrm{ADEY}}_{S_n}(t) - \mathrm{ADEY}_{S_n}(t) \right) = \frac{1}{\sqrt{|S_n|}} \sum_{i \in S_n}V_i^{\mathrm{ADEY}} + o_P(1).
\end{align*}
Thus, if it is possible to simulate the distribution of $|S_n|^{-1/2} \sum_{i \in S_n}V_i^{\mathrm{ADEY}}$, we can construct an asymptotically valid confidence interval for $\mathrm{ADEY}_{S_n}(t)$.
To this end, noting that the sample mean of $\hat V_i^{\mathrm{ADEY}}$ over $S_n$ is zero, we construct a bootstrap counterpart $V_i^{*,\mathrm{ADEY}}$ of $V_i^{\mathrm{ADEY}}$ in the following procedure: $V_i^{*,\mathrm{ADEY}} \coloneqq \hat V_i^{\mathrm{ADEY}} R_i$, where $R_i$ is the $i$-th element of the $|S_n| \times 1$ vector $[\Omega_{S_n}(b_n)]^{1/2} \zeta_{S_n}$ with $\Omega_{S_n}(b_n) \coloneqq ( |S_{\bm{A}}(i, b_n) \cap S_{\bm{A}}(j, b_n)| / M_{S_n}(b_n, 1) )_{i,j \in S_n}$, $\zeta_{S_n}$ is an $|S_n| \times 1$ vector of random variables drawn from $\mathrm{Normal}(0, \bm{I}_{|S_n|})$ independently of the data, and $b_n$ is a bandwidth parameter.
Then, by repeatedly drawing $\zeta_{S_n}$ many times, we can obtain the distribution of $|S_n|^{-1/2}\sum_{i \in S_n}V_i^{*,\mathrm{ADEY}}$ conditional on the observed data, which serves as an approximation of the distribution of $|S_n|^{-1/2} \sum_{i \in S_n}V_i^{\mathrm{ADEY}}$.
An intuition for the (first-order) validity of this bootstrap method is as follows. 
Since the conditional expectation of $V_i^{*,\mathrm{ADEY}}$ given the observed data is zero, we have
\begin{align*}
	\Var\left[ \frac{1}{\sqrt{|S_n|}} \sum_{i \in S_n}  V_i^{*,\mathrm{ADEY}} \; \middle| \; \text{ data } \right]
	= \frac{1}{|S_n|}\sum_{i \in S_n} \sum_{j \in S_n} \hat V_i^{\mathrm{ADEY}} \hat V_j^{\mathrm{ADEY}} [\Omega_{S_n}(b_n)]_{i,j}.
\end{align*}
Thus, this is a version of the HAC estimator with kernel $\Omega_{S_n}(b_n)$.
For more details, see \cite{kojevnikov2021bootstrap}.

\subsection{Proof of Theorem \ref{thm:HAC}}

To save space, we prove only the result for $\sigma_{S_n}^{\mathrm{ADEY}}$ (those for $\sigma_{S_n}^{\mathrm{ADED}}$ and $\sigma_{S_n}^{\mathrm{LADE}}$ can be shown in the same manner).
It is easy to see that
\begin{align} \label{eq:HACdecomp}
	\begin{split}
		(\hat \sigma_{S_n}^{\mathrm{ADEY}})^2 
		& = (\tilde \sigma_{S_n}^{\mathrm{ADEY}})^2 + B_{S_n}^{\mathrm{ADEY}}  \\
		& \quad +  \frac{1}{|S_n|} \sum_{i \in S_n} \sum_{j \in S_n} \left( \hat V_i^{\mathrm{ADEY}} \hat V_j^{\mathrm{ADEY}} - V_i^{\mathrm{ADEY}} V_j^{\mathrm{ADEY}} \right) \bm{1}\{ \ell_{\bm{A}}(i, j) \le b_n \} \\
		& \quad + \frac{2}{|S_n|} \sum_{i \in S_n} \sum_{j \in S_n} ( V_i^{\mathrm{ADEY}} - \bbE[V_i^{\mathrm{ADEY}}] ) \bbE[V_j^{\mathrm{ADEY}}] \bm{1}\{ \ell_{\bm{A}}(i, j) \le b_n \},
	\end{split}
\end{align}
where $(\tilde \sigma_{S_n}^{\mathrm{ADEY}})^2$ is the infeasible oracle estimator defined as
\begin{align*}
	(\tilde \sigma_{S_n}^{\mathrm{ADEY}})^2 
	& \coloneqq \frac{1}{|S_n|} \sum_{i \in S_n} \sum_{j \in S_n} \left( V_i^{\mathrm{ADEY}} - \bbE [V_i^{\mathrm{ADEY}}] \right) \left( V_j^{\mathrm{ADEY}} - \bbE [V_j^{\mathrm{ADEY}}] \right) \bm{1}\{ \ell_{\bm{A}}(i, j) \le b_n \}.
\end{align*}
By Lemma \ref{lem:linearpsi1} and Assumptions \ref{as:outcome}, \ref{as:overlap}, and \ref{as:weakHAC}(i), Proposition 4.1 of \cite{kojevnikov2021limit} implies that $(\tilde \sigma_{S_n}^{\mathrm{ADEY}})^2 = (\sigma_{S_n}^{\mathrm{ADEY}})^2 + o_P(1)$.
Thus, we obtain the desired result if the second and third lines of \eqref{eq:HACdecomp} are asymptotically negligible.

For the second line of \eqref{eq:HACdecomp}, observe that
\begin{align} \label{eq:VV}
	\begin{split}
		& \frac{1}{|S_n|} \sum_{i \in S_n} \sum_{j \in S_n} \left( \hat V_i^{\mathrm{ADEY}} \hat V_j^{\mathrm{ADEY}} - V_i^{\mathrm{ADEY}} V_j^{\mathrm{ADEY}} \right) \bm{1}\{ \ell_{\bm{A}}(i, j) \le b_n \} \\
		& = \frac{1}{|S_n|} \sum_{i \in S_n} \sum_{j \in S_n} \left( \hat V_i^{\mathrm{ADEY}} - V_i^{\mathrm{ADEY}} \right) \hat V_j^{\mathrm{ADEY}} \bm{1}\{ \ell_{\bm{A}}(i, j) \le b_n \} \\
		& \quad + \frac{1}{|S_n|} \sum_{i \in S_n} \sum_{j \in S_n} \left( \hat V_j^{\mathrm{ADEY}} - V_j^{\mathrm{ADEY}} \right) V_i^{\mathrm{ADEY}} \bm{1}\{ \ell_{\bm{A}}(i, j) \le b_n \}.
	\end{split}
\end{align}
Since $\max_{j \in S_n} |\hat V_j^{\mathrm{ADEY}}| = O_P(1)$ and $\max_{i \in S_n} |\hat V_i^{\mathrm{ADEY}} - V_i^{\mathrm{ADEY}}| = O_P(|S_n|^{-1/2})$ by Assumptions \ref{as:outcome}, \ref{as:overlap}, and \ref{as:weakHAC}(i) and Lemmas \ref{lem:p1} and \ref{lem:consistency}, we have
\begin{align*}
	& \left| \frac{1}{|S_n|} \sum_{i \in S_n} \sum_{j \in S_n} \left( \hat V_i^{\mathrm{ADEY}} - V_i^{\mathrm{ADEY}} \right) \hat V_j^{\mathrm{ADEY}} \bm{1}\{ \ell_{\bm{A}}(i, j) \le b_n \} \right| \\
	& \le \left( \max_{j \in S_n} | \hat V_j^{\mathrm{ADEY}} | \right) \left( \max_{i \in S_n} | \hat V_i^{\mathrm{ADEY}} - V_i^{\mathrm{ADEY}} | \right) \frac{1}{|S_n|} \sum_{i \in S_n} \sum_{j \in S_n} \bm{1}\{ \ell_{\bm{A}}(i, j) \le b_n \} \\
	& = O_P(1) \cdot O_P\left( \frac{1}{\sqrt{|S_n|}} \right) \cdot M_{S_n}(b_n, 1),
\end{align*}
which is $o_P(1)$ under Assumption \ref{as:weakHAC}(ii).
Similarly, we can show that the second term of \eqref{eq:VV} is $o_P(1)$.
Thus, the second line of \eqref{eq:HACdecomp} is $o_P(1)$.

To evaluate the third line of \eqref{eq:HACdecomp}, let $\kappa_{n,i} \coloneqq \sum_{j \in S_n} \bbE[V_j^{\mathrm{ADEY}}] \bm{1}\{ \ell_{\bm{A}}(i, j) \le b_n \}$.
Using the norm inequality, we have
\begin{align*}
	& \bbE \left| \frac{1}{|S_n|} \sum_{i \in S_n} ( V_i^{\mathrm{ADEY}} - \bbE [V_i^{\mathrm{ADEY}}] ) \kappa_{n,i} \right| \\
	& \le \left( \bbE \left[ \left( \frac{1}{|S_n|} \sum_{i \in S_n} ( V_i^{\mathrm{ADEY}} - \bbE V_i^{\mathrm{ADEY}} ) \kappa_{n,i} \right)^2 \right] \right)^{1/2} \\
	& = \left( \frac{1}{|S_n|^2} \sum_{i \in S_n} \Var[ V_i^{\mathrm{ADEY}} ] \kappa_{n,i}^2 + \frac{1}{|S_n|^2} \sum_{i \in S_n} \sum_{j \in S_n \setminus \{ i \}} \Cov[ V_i^{\mathrm{ADEY}}, V_j^{\mathrm{ADEY}} ] \kappa_{n,i} \kappa_{n,j} \right)^{1/2}.
\end{align*}
Noting that $V_i^{\mathrm{ADEY}}$ is bounded by Assumptions \ref{as:outcome} and \ref{as:overlap}, we have $ |S_n|^{-2} \sum_{i \in S_n} \Var[ V_i^{\mathrm{ADEY}} ] \kappa_{n,i}^2 \le C |S_n|^{-1} M_{S_n}(b_n, 2) = o(1)$ by Assumption \ref{as:weakHAC}(ii).
Further, we can see that
\begin{align*}
	& \left| \frac{1}{|S_n|^2} \sum_{i \in S_n} \sum_{j \in S_n \setminus \{ i \}} \Cov[ V_i^{\mathrm{ADEY}}, V_j^{\mathrm{ADEY}} ] \kappa_{n,i} \kappa_{n,j} \right| \\
	& \le \frac{1}{|S_n|^2} \sum_{s = 1}^{n-1} \sum_{i \in S_n} \sum_{j \in S_n}  \bm{1}\{ \ell_{\bm{A}}(i, j) = s \} |\Cov[ V_i^{\mathrm{ADEY}}, V_j^{\mathrm{ADEY}} ]| \cdot |\kappa_{n,i}| \cdot |\kappa_{n,j}| \\
	& \le \frac{C}{|S_n|^2} \sum_{s = 1}^{n-1} \tilde \theta_{n,s}^{\mathrm{ADE}} \sum_{i \in S_n} \sum_{j \in S_n} \sum_{k \in S_n} \sum_{l \in S_n} \bm{1}\{ \ell_{\bm{A}}(i, j) = s \} \bm{1}\{ \ell_{\bm{A}}(i, k) \le b_n \} \bm{1}\{ \ell_{\bm{A}}(j, l) \le b_n \} \\
	& = \frac{C}{|S_n|^2} \sum_{s = 1}^{n-1} |\mathcal{J}_{S_n}(s, b_n)| \tilde \theta_{n,s}^{\mathrm{ADE}}
	= o(1),
\end{align*}
where the second inequality follows from the fact that $\{ V_i^{\mathrm{ADEY}} \}_{i \in S_n}$ is $\psi$-weakly dependent by Lemma \ref{lem:linearpsi1} and the last line follows from the second part of Assumption \ref{as:weakHAC}(i).
Thus, the third line of \eqref{eq:HACdecomp} is $o_P(1)$.
\qed 


\section{Average Overall Effects} \label{sec:overall}

We consider the identification and estimation of average overall effect (AOE) parameters in a similar spirit to \cite{hu2022average}.

\subsection{Identification}

Let $\bar \mu_{S_n}^Y(z; \bar{\mathcal{E}}) \coloneqq |S_n|^{-1} \sum_{i \in S_n} \sum_{j \in \bar{\mathcal{E}}_i} \mu_{ji}^Y(z)$, where $\bar{\mathcal{E}}_i \coloneqq \mathcal{E}_i \cup \{ i \}$.
The AOE of the IV on the outcome is defined by $\mathrm{AOEY}_{S_n}
\coloneqq \bar \mu_{S_n}^Y(1; \bar{\mathcal{E}}) - \bar \mu_{S_n}^Y(0; \bar{\mathcal{E}})$, which can be rewritten as the sum of the ADE and AIE estimands:
\begin{align*}
	\mathrm{AOEY}_{S_n}
	& = \frac{1}{ |S_n| } \sum_{i \in S_n} [ \mu_i^Y(1) - \mu_i^Y(0) ] + \frac{1}{ |S_n| } \sum_{i \in S_n} \sum_{j \in \mathcal{E}_i} [ \mu_{ji}^Y(1) - \mu_{ji}^Y(0) ] \\
	& = \mathrm{ADEY}_{S_n} + \mathrm{AIEY}_{S_n}.
\end{align*}
Similarly, we define $\bar \mu_{S_n}^D(z; \bar{\mathcal{E}}) \coloneqq |S_n|^{-1} \sum_{i \in S_n} \sum_{j \in \bar{\mathcal{E}}_i} \mu_{ji}^D(z)$ and $\mathrm{AOED}_{S_n} \coloneqq \bar \mu_{S_n}^D(1; \bar{\mathcal{E}}) - \bar \mu_{S_n}^D(0; \bar{\mathcal{E}}) = \mathrm{ADED}_{S_n} + \mathrm{AIED}_{S_n}$.

The following proposition shows that the AOE estimands can be interpreted as the weighted average of the effect of $i$'s IV on the sum of $i$'s own outcome and $j$'s ($j \in \mathcal{E}_i$) outcomes and that on the sum of $i$'s own treatment and $j$'s treatments with the weight equal to $\pi_i(\bm{z}_{-i})$.
The proof is the same as Proposition \ref{prop:AIE_ITT} and thus omitted.

\begin{proposition} \label{prop:AOE_ITT}
	Under Assumption \ref{as:independence},
	\begin{align*}
		\mathrm{AOEY}_{S_n}
		& = \frac{1}{|S_n|} \sum_{i \in S_n} \sum_{\bm{z}_{-i} \in \{ 0, 1 \}^{n-1}} \sum_{j \in \bar{\mathcal{E}}_i} \{ y_j(Z_i = 1, \bm{Z}_{-i} = \bm{z}_{-i}) - y_j(Z_i = 0, \bm{Z}_{-i} = \bm{z}_{-i}) \} \pi_i(\bm{z}_{-i}), \\
		\mathrm{AOED}_{S_n}
		& = \frac{1}{|S_n|} \sum_{i \in S_n} \sum_{\bm{z}_{-i} \in \{ 0, 1 \}^{n-1}} \sum_{j \in \bar{\mathcal{E}}_i} \{ D_j(Z_i = 1, \bm{Z}_{-i} = \bm{z}_{-i}) - D_j(Z_i = 0, \bm{Z}_{-i} = \bm{z}_{-i}) \} \pi_i(\bm{z}_{-i}).
	\end{align*}
\end{proposition}

Next, we turn to the local average overall effect (LAOE), which is defined by
\begin{align*}
	\mathrm{LAOE}_{S_n}
	\coloneqq \sum_{i \in S_n} \sum_{\bm{z}_{-i} \in \{ 0, 1 \}^{n-1}} \sum_{j \in \bar{\mathcal{E}}_i} \{ y_j(Z_i = 1, \bm{Z}_{-i} = \bm{z}_{-i}) - y_j(Z_i = 0, \bm{Z}_{-i} = \bm{z}_{-i}) \} \frac{ \mathcal{C}_i(\bm{z}_{-i}) \pi_i(\bm{z}_{-i}) }{ \sum_{i \in S_n} \bbE[ \mathcal{C}_i ] }.
\end{align*}
The LAOE can be written as the sum of the LADE and LAIE parameters:
\begin{align*}
	\mathrm{LAOE}_{S_n} = \mathrm{LADE}_{S_n} + \mathrm{LAIE}_{S_n}.
\end{align*}
Considering this, it is trivial from Theorems \ref{thm:LADE} and \ref{thm:LAIE1} that the LAOE is identifiable from a Wald type estimand.

\begin{theorem} \label{thm:LAOE}
	Suppose that Assumptions \ref{as:independence}, \ref{as:relevance2}, and \ref{as:monotone2} hold.
	In addition, let Assumption \ref{as:restrict2} hold when $\mathcal{E}_i$ is replaced by $\bar{\mathcal{E}}_i$.
	Then, $\mathrm{LAOE}_{S_n} = \mathrm{AOEY}_{S_n} / \mathrm{ADED}_{S_n}$.
\end{theorem}

\subsection{Estimation and asymptotic theory}

We estimate $\mathrm{AOEY}_{S_n}$ by $\hat{\mathrm{AOEY}}_{S_n} \coloneqq \hat \mu_{S_n}^Y(1; \bar{\mathcal{E}}) - \hat \mu_{S_n}^Y(0; \bar{\mathcal{E}})$, where 
\begin{align*}
	\hat \mu_{S_n}^Y(z; \bar{\mathcal{E}})
	\coloneqq \frac{1}{|S_n|} \sum_{i \in S_n} \sum_{j \in \bar{\mathcal{E}}_i} \frac{ Y_j \cdot \bm{1}\{ Z_i = z\} }{ \hat p_{S_n}(z) }.
\end{align*}
Similarly, we can define the estimators of $\mathrm{AOED}_{S_n}$ and $\bar \mu_{S_n}^D(z; \bar{\mathcal{E}})$, denoted as $\hat{\mathrm{AOED}}_{S_n}$ and $\hat \mu_{S_n}^D(z; \bar{\mathcal{E}})$, respectively.
Then, we have
\begin{align*}
	\hat{\mathrm{LAOE}}_{S_n}
	= \frac{ \hat{\mathrm{AOEY}}_{S_n} }{ \hat{\mathrm{ADED}}_{S_n} },
\end{align*}
where $\hat{\mathrm{ADED}}_{S_n}$ is defined in Section \ref{subsec:estimator}.

The asymptotic properties of these estimators can be derived in the same way as in Section \ref{subsec:asymptotics}.
Given the identification result in Theorem \ref{thm:LAOE}, we focus on the estimation of $\mathrm{AOEY}_{S_n}$ and $\mathrm{LAOE}_{S_n}$;
the case of $\mathrm{AOED}_{S_n}$ is analogous.
Let $\theta_{n,s}^{\mathrm{AOE}} \coloneqq \max \{ \theta_{n,s}^{\mathrm{ADE}}, \theta_{n,s}^{\mathrm{AIE}} \}$, and $\tilde \theta_{n,s}^{\mathrm{AOE}}$ is defined in the same way as in \eqref{eq:tildetheta}.
Suppose that Assumptions \ref{as:ANI1}, \ref{as:weakLLN1}, and \ref{as:weakCLT1} hold when $\tilde \theta_{n,s}^{\mathrm{ADE}}$ is replaced by $\tilde \theta_{n,s}^{\mathrm{AOE}}$.
For each $i \in S_n$, let
\begin{align*}
	V_{\bar{\mathcal{E}}_i}^{\mathrm{AOEY}}
	& \coloneqq W_{{\bar{\mathcal{E}}_i}}^Y - \frac{ \bar \mu_{S_n}^Y(1; \bar{\mathcal{E}}) }{ p_{S_n}(1) } W_{\bar{\mathcal{E}}_i}^Z + \frac{ \bar \mu_{S_n}^Y(0; \bar{\mathcal{E}}) }{ p_{S_n}(0) } W_{\bar{\mathcal{E}}_i}^{1-Z},
	\qquad 
	V_{\bar{\mathcal{E}}_i}^{\mathrm{ADED}}
	\coloneqq W_{{\bar{\mathcal{E}}_i}}^D - \frac{ \bar \mu_{S_n}^D(1) }{ p_{S_n}(1) } W_{\bar{\mathcal{E}}_i}^Z + \frac{ \bar \mu_{S_n}^D(0) }{ p_{S_n}(0) } W_{\bar{\mathcal{E}}_i}^{1-Z}, \\
	V_{\bar{\mathcal{E}}_i}^{\mathrm{LAOE}}
	& \coloneqq \frac{ 1 }{ \mathrm{ADED}_{S_n} } V_{\bar{\mathcal{E}}_i}^{\mathrm{AOEY}} - \frac{ \mathrm{AOEY}_{S_n} }{ \mathrm{ADED}_{S_n}^2 } V_{\bar{\mathcal{E}}_i}^{\mathrm{ADED}},
\end{align*}
where
\begin{align*}
	W_{\bar{\mathcal{E}}_i}^Y 
	& \coloneqq \sum_{j \in \bar{\mathcal{E}}_i} Y_j \left[ \frac{ \bm{1}\{ Z_i = 1 \} }{ p_{S_n}(1) } - \frac{ \bm{1}\{ Z_i = 0 \} }{ p_{S_n}(0) } \right], \\ 
	W_{\bar{\mathcal{E}}_i}^D 
	& \coloneqq D_i \left[ \frac{ \bm{1}\{ Z_i = 1 \} }{ p_{S_n}(1) } - \frac{ \bm{1}\{ Z_i = 0 \} }{ p_{S_n}(0) } \right], \\ 
	W_{\bar{\mathcal{E}}_i}^Z
	& \coloneqq \bm{1}\{ Z_i = 1 \},
	\qquad 
	W_{\bar{\mathcal{E}}_i}^{1-Z} 
	\coloneqq \bm{1}\{ Z_i = 0 \}.
\end{align*}
Let $(\sigma_{S_n}^\mathrm{AOEY})^2 \coloneqq \Var[ |S_n|^{-1/2} \sum_{i \in S_n} V_{\bar{\mathcal{E}}_i}^{\mathrm{AOEY}} ]$, and $( \sigma_{S_n}^\mathrm{LAOE} )^2$ is defined analogously.
Given these definitions, we obtain the asymptotic normality results for the estimators of $\mathrm{AOEY}_{S_n}$ and $\mathrm{LAOE}_{S_n}$ by replacing the variables in Theorem \ref{thm:normal1} with the corresponding variables defined here.


\section{Average Spillover Effects} \label{sec:spillover}

We consider the identification and estimation of average spillover effect (ASE) parameters similar to those considered in \cite{imai2020causal}.

\subsection{Identification}

\subsubsection{Intention-to-treat analysis}

The ASE of the IV on the outcome is defined by $\mathrm{ASEY}_{S_n}(z, t, t') \coloneqq \bar \mu_{S_n}^Y(z, t) - \bar \mu_{S_n}^Y(z, t')$ for $z \in \{ 0, 1 \}$ and $t, t' \in \mathcal{T}$ such that $t \neq t'$, where $\bar \mu_{S_n}^Y(z, t)$ is defined in Section \ref{sec:identification}.
Similarly, we define $\mathrm{ASED}_{S_n}(z, t, t') \coloneqq \bar \mu_{S_n}^D(z, t) - \bar \mu_{S_n}^D(z, t')$.

The next proposition characterizes these estimands in terms of the potential outcome and the potential treatment status under the potential misspecification of IEM.
The proof is trivial from the proof of Proposition \ref{prop:ITT} and thus omitted.

\begin{proposition} \label{prop:ASE_ITT}
	Under Assumption \ref{as:independence},
	\begin{align*}
		\mathrm{ASEY}_{S_n}(z, t, t')
		& = \frac{1}{|S_n|} \sum_{i \in S_n} \sum_{\bm{z}_{-i} \in \{0, 1\}^{n-1}} y_i(z, \bm{z}_{-i}) \{ \pi_i(\bm{z}_{-i}, t) - \pi_i(\bm{z}_{-i}, t') \}, \\
		\mathrm{ASED}_{S_n}(z, t, t')
		& = \frac{1}{|S_n|} \sum_{i \in S_n} \sum_{\bm{z}_{-i} \in \{0, 1\}^{n-1}} D_i(z, \bm{z}_{-i}) \{ \pi_i(\bm{z}_{-i}, t) - \pi_i(\bm{z}_{-i}, t') \}.
	\end{align*}
\end{proposition}

Differently from the ADE estimands considered in Proposition \ref{prop:ITT}, the characterizations of the ASE estimands are somewhat difficult to derive causal interpretation.
This is mainly due to the misspecification of IEM. 
To be more precise, let the pre-specified IEM $T$ be correct in the sense of \eqref{eq:correctIEM}.
Then, we can write the potential outcome and potential treatment status given $Z_i = z$ and $T_i = t$ as $\tilde y_i(z, t)$ and $\tilde d_i(z, t)$, respectively.
Noting that $Y_i = \tilde y_i(Z_i, T_i)$ and $D_i = \tilde d_i(Z_i, T_i)$, it is straightforward to see that the ASE estimands defined by this correct $T$ have the following causal interpretation:
\begin{align*}
	\mathrm{ASEY}_{S_n}(z, t, t')
	= \frac{1}{|S_n|} \sum_{i \in S_n} \{ \tilde y_i(z, t) - \tilde y_i(z, t') \},
	\qquad
	\mathrm{ASED}_{S_n}(z, t, t')
	= \frac{1}{|S_n|} \sum_{i \in S_n} \{ \tilde d_i(z, t) - \tilde d_i(z, t') \}.
\end{align*}

Even without the knowledge of true IEM, the ASE estimands give us some beneficial information regarding the spillover effects; that is, if the ASE for some IEM is non-zero, it indicates the presence of some form of interference (see Remark \ref{rem:spillover}).

\subsubsection{Local average spillover effect}

Given the results in Proposition \ref{prop:ASE_ITT}, under the misspecification of IEM, it would be difficult to identify a meaningful ASE parameter by a Wald-type estimand.
For this reason, we assume that a correct IEM is known to us.

\begin{assumption}[Correct specification] \label{as:correctIEM}
	The IEM $T$ is correct in the sense of \eqref{eq:correctIEM}.
\end{assumption}


For each $z \in \{ 0, 1 \}$ and $t, t' \in \mathcal{T}$ such that $t \neq t'$, let $\tilde{\mathcal{C}}_i(z, t, t') \coloneqq \bm{1}\{ \tilde d_i(z, t) = 1, \tilde d_i(z, t') = 0 \}$ indicate a complier whose treatment status is switched by changing the instrumental exposure's value.
The local average spillover effect (LASE) is defined by
\begin{align*}
	\mathrm{LASE}_{S_n}(z, t, t') 
	\coloneqq \sum_{i \in S_n} \{ \tilde y_i(z, t) - \tilde y_i(z, t') \} \frac{ \tilde{\mathcal{C}}_i(z, t, t') }{ \sum_{i \in S_n} \tilde{\mathcal{C}}_i(z, t, t') }.
\end{align*}
The LASE measures the average causal effect of changing the instrumental exposure on the outcome over the compliers.
By the definition of $\tilde{\mathcal{C}}_i(z, t, t')$, we can see that the LASE captures both the average effect of changing unit's own treatment receipt and the average effect caused by changing the instrumental exposure from $t'$ to $t$.
Note that the definition of LASE assumes the existence of spillovers between the instrumental exposure and the treatment receipt.

\begin{assumption}[Relevance 3] \label{as:relevance3}
	$|S_n|^{-1} \sum_{i \in S_n} \tilde{\mathcal{C}}_i(z, t, t') \ge c$ for a constant $c > 0$.
\end{assumption}

\begin{assumption}[Monotonicity 3] \label{as:monotone3}
	$\tilde d_i(z, t) \ge \tilde d_i(z, t')$ for all $i \in S_n$.
\end{assumption}

\begin{assumption}[Restricted interference 3] \label{as:restrict3}
	For all $i \in S_n$, $\tilde y_i(z, t) = \tilde y_i(z, t')$ holds whenever $\tilde d_i(z, t) = \tilde d_i(z, t')$.
\end{assumption}

Without Assumption \ref{as:relevance3}, the LASE is no longer well-defined.
This assumption may be violated when the treatment status of the focal unit does not depend on the IVs of the others.
Assumption \ref{as:monotone3} restricts the direction of the treatment response to the instrumental exposure.
For example, the assumption can be satisfied with the treatment choice equation in Example \ref{example:latent1}.
Assumption \ref{as:restrict3} rules out the effect of the instrumental exposure on the potential outcome for the units whose treatment statuses do not depend on their instrumental exposures.

The following theorem shows that the LASE is identifiable from a Wald-type estimand.
The proof is given in the end of this section.

\begin{theorem}\label{thm:LASE}
	Under Assumptions \ref{as:correctIEM} -- \ref{as:restrict3}, $\mathrm{LASE}_{S_n}(z, t, t') = \mathrm{ASEY}_{S_n}(z, t, t') / \mathrm{ASED}_{S_n}(z, t, t')$.
\end{theorem}

\subsection{Estimation and asymptotic theory}

We estimate the ASE estimands by $\hat{\mathrm{ASEY}}_{S_n}(z, t, t') \coloneqq \hat \mu_{S_n}^Y(z, t) - \hat \mu_{S_n}^Y(z, t')$ and $\hat{\mathrm{ASED}}_{S_n}(z, t, t') \coloneqq \hat \mu_{S_n}^D(z, t) - \hat \mu_{S_n}^D(z, t')$, where $\hat \mu_{S_n}^Y(z, t)$ and $\hat \mu_{S_n}^D(z, t)$ are defined in Section \ref{subsec:estimator}.
Then, the LASE estimator is given by
\begin{align*}
	\hat{\mathrm{LASE}}_{S_n}(z, t, t')
	\coloneqq \frac{ \hat{\mathrm{ASEY}}_{S_n}(z, t, t') }{ \hat{\mathrm{ASED}}_{S_n}(z, t, t') }.
\end{align*}

In the same manner as in Section \ref{subsec:asymptotics}, we can prove the consistency and asymptotic normality for the ASE estimators.
Specifically, for each $i \in S_n$, define  
\begin{align*}
	V_i^{\mathrm{ASEY}}
	& \coloneqq W_i^Y(z, t, t') - \frac{ \bar \mu_{S_n}^Y(z, t) }{ p_{S_n}(z, t) } W_i(z, t) + \frac{ \bar \mu_{S_n}^Y(z, t') }{ p_{S_n}(z, t') } W_i(z, t'), \\
	V_i^{\mathrm{ASED}}
	& \coloneqq W_i^D(z, t, t') - \frac{ \bar \mu_{S_n}^D(z, t) }{ p_{S_n}(z, t) } W_i(z, t) + \frac{ \bar \mu_{S_n}^D(z, t') }{ p_{S_n}(z, t') } W_i(z, t'), \\
	V_i^{\mathrm{LASE}}
	& \coloneqq \frac{1}{\mathrm{ASED}_{S_n}(z, t, t')} V_i^{\mathrm{ASEY}}  - \frac{\mathrm{ASEY}_{S_n}(z, t, t')}{[\mathrm{ASED}_{S_n}(z, t, t')]^2} V_i^{\mathrm{ASED}},
\end{align*}
where
\begin{align*}
	W_i^Y(z, t, t')
	& \coloneqq Y_i \left[ \frac{ \bm{1}\{ Z_i = z, T_i = t \} }{ p_{S_n}(z ,t) } - \frac{ \bm{1}\{ Z_i = z, T_i = t' \} }{ p_{S_n}(z ,t') } \right], \\
	W_i^D(z, t, t')
	& \coloneqq D_i \left[ \frac{ \bm{1}\{ Z_i = z, T_i = t \} }{ p_{S_n}(z ,t) } - \frac{ \bm{1}\{ Z_i = z, T_i = t' \} }{ p_{S_n}(z ,t') } \right], \\
	W_i(z, t) 
	& \coloneqq \bm{1}\{ Z_i = z, T_i = t \}.
\end{align*}
Let $(\sigma_{S_n}^\mathrm{ASEY})^2 \coloneqq \Var[ |S_n|^{-1/2} \sum_{i \in S_n} V_i^{\mathrm{ASEY}} ]$, and $( \sigma_{S_n}^\mathrm{ASED} )^2$ and $( \sigma_{S_n}^\mathrm{LASE} )^2$ are defined analogously.
Then, the asymptotic normality results for the ASE estimators can be established when the variables in Theorem \ref{thm:normal1} are replaced by the corresponding variables defined here.

\begin{remark}[Testing the presence of spillovers]\label{rem:spillover}
	If there do not exist any spillover effects actually such that $y_i(z_i, \bm{z}_{-i}) = y_i(z_i)$ and $D_i(z_i, \bm{z}_{-i}) = D_i(z_i)$, clearly from Proposition \ref{prop:ASE_ITT}, both $\mathrm{ASEY}_{S_n}(z,t,t')$ and $\mathrm{ASED}_{S_n}(z,t,t')$ should be zero for any IEM.
	In other words, if $\hat{\mathrm{ASEY}}_{S_n}(z,t,t')$ or $\hat{\mathrm{ASED}}_{S_n}(z,t,t')$ is significantly different from zero at least for some arbitrary IEM, we can deduce the existence of some spillover effects even without the knowledge of correct IEM.
\end{remark}

\subsection{Proof of Theorem \ref{thm:LASE}}

Under Assumptions \ref{as:correctIEM} and \ref{as:monotone3}, observe that $\mu_i^D(z, t) - \mu_i^D(z, t') = \tilde d_i(z, t) - \tilde d_i(z, t') = \tilde{\mathcal{C}}_i(z, t, t')$.
Further, by Assumptions \ref{as:correctIEM}, \ref{as:monotone3}, and \ref{as:restrict3}, we have
\begin{align*}
	\mu_i^Y(z, t) - \mu_i^Y(z, t')
	& = \tilde y_i(z, t) - \tilde y_i(z, t') \\
	& = [\tilde y_i(z, t) - \tilde y_i(z, t')] \mathbf{1}\{ \tilde d_i(z, t) \neq \tilde d_i(z, t')  \} \\
	& \quad + [\tilde y_i(z, t) - \tilde y_i(z, t')] \mathbf{1}\{ \tilde d_i(z, t) = \tilde d_i(z, t') \} \\
	& = [\tilde y_i(z, t) - \tilde y_i(z, t')] \tilde{\mathcal{C}}_i(z, t, t').
\end{align*}
In conjunction with Assumption \ref{as:relevance3}, we obtain the desired result.
\qed


\section{Additional Results for Section \ref{sec:identification}} \label{sec:identification2}

\subsection{Additional discussion on compliers} \label{subsec:complier}

We compare the definitions of compliers in this study and \cite{vazquez2023causal}.
For ease of exposition, we focus on the simple setting wherein spillovers occur only within pairs of two units.
We label two units in each pair as units 1 and 2, and we consider the type of unit 1 (the case of unit 2 is analogous).
For each $z_1, z_2 \in \{ 0, 1 \}$, denote the potential treatment status of unit 1 given $Z_1 = z_1$ and $Z_2 = z_2$ as $D_1(z_1, z_2)$.
In this case, the type of unit 1 in our framework is determined from
\begin{align}\label{eq:ourtype}
	\begin{split}
		\mathcal{AT}_1(z_2) 
		& = \bm{1}\{ D_1(1, z_2) = 1, D_1(0, z_2) = 1 \},\\
		\mathcal{C}_1(z_2) 
		& = \bm{1}\{ D_1(1, z_2) = 1, D_1(0, z_2) = 0 \},\\
		\mathcal{NT}_1(z_2) 
		& = \bm{1}\{ D_1(1, z_2) = 0, D_1(0, z_2) = 0 \}.
	\end{split}
\end{align}
By contrast, \cite{vazquez2023causal} considers the following types (cf. Table 1 of the paper):
\begin{align} \label{eq:vtype}
	\begin{split}
		\mathtt{AT}_1^{\mathtt{V}}
		& = \bm{1}\{ D_1(1, 1) = 1, D_1(1, 0) = 1, D_1(0, 1) = 1, D_1(0, 0) = 1 \},\\
		\mathtt{SC}_1^{\mathtt{V}}
		& = \bm{1}\{ D_1(1, 1) = 1, D_1(1, 0) = 1, D_1(0, 1) = 1, D_1(0, 0) = 0 \},\\
		\mathtt{C}_1^{\mathtt{V}}
		& = \bm{1}\{ D_1(1, 1) = 1, D_1(1, 0) = 1, D_1(0, 1) = 0, D_1(0, 0) = 0 \},\\
		\mathtt{GC}_1^{\mathtt{V}}
		& = \bm{1}\{ D_1(1, 1) = 1, D_1(1, 0) = 0, D_1(0, 1) = 0, D_1(0, 0) = 0 \},\\
		\mathtt{NT}_1^{\mathtt{V}}
		& = \bm{1}\{ D_1(1, 1) = 0, D_1(1, 0) = 0, D_1(0, 1) = 0, D_1(0, 0) = 0 \},
	\end{split}
\end{align}
where $\mathtt{AT}_1^{\mathtt{V}}$ stands for an always taker, $\mathtt{SC}_1^{\mathtt{V}}$ a social-interaction complier, $\mathtt{C}_1^{\mathtt{V}}$ a complier, $\mathtt{GC}_1^{\mathtt{V}}$ a group complier, and $\mathtt{NT}_1^{\mathtt{V}}$ a never taker.
Here, we add the superscript $\mathtt{V}$ to highlight that these variables originate from \cite{vazquez2023causal}.
Then, it is straightforward to observe the following relationships between \eqref{eq:ourtype} and \eqref{eq:vtype}:
\begin{align*}
	\mathtt{AT}_1^\mathtt{V} & = \mathcal{AT}_1(1) \times \mathcal{AT}_1(0), \\
	\mathtt{SC}_1^\mathtt{V} & = \mathcal{AT}_1(1) \times \mathcal{C}_1(0), \\
	\mathtt{C}_1^\mathtt{V}  & = \mathcal{C}_1(1) \times \mathcal{C}_1(0), \\
	\mathtt{GC}_1^\mathtt{V} & = \mathcal{C}_1(1) \times \mathcal{NT}_1(0), \\
	\mathtt{NT}_1^\mathtt{V} & = \mathcal{NT}_1(1) \times \mathcal{NT}_1(0).
\end{align*}
Furthermore, under the monotonicity condition of \cite{vazquez2023causal} (cf. Assumption 3 of the paper), we observe that
\begin{align*}
	\mathcal{C}_1(1) 
	& = \bm{1}\{ D_1(1, 1) = 1, D_1(0, 1) = 0 \} \\
	& = \bm{1}\{ D_1(1, 1) = 1, D_1(1, 0) = 1, D_1(0, 1) = 0, D_1(0, 0) = 0 \} \\
	& \quad + \bm{1}\{ D_1(1, 1) = 1, D_1(1, 0) = 0, D_1(0, 1) = 0, D_1(0, 0) = 0 \} \\
	& = \mathtt{C}_1^\mathtt{V} + \mathtt{GC}_1^\mathtt{V}
\end{align*}
and that
\begin{align*}
	\mathcal{C}_1(0) 
	& = \bm{1}\{ D_1(1, 0) = 1, D_1(0, 0) = 0 \} \\
	& = \bm{1}\{ D_1(1, 1) = 1, D_1(1, 0) = 1, D_1(0, 1) = 1, D_1(0, 0) = 0 \} \\
	& \quad + \bm{1}\{ D_1(1, 1) = 1, D_1(1, 0) = 1, D_1(0, 1) = 0, D_1(0, 0) = 0 \} \\
	& = \mathtt{C}_1^\mathtt{V} + \mathtt{SC}_1^\mathtt{V}.
\end{align*}
Hence, the compliers in our framework comprise the compliers and group/social-interaction compliers in \cite{vazquez2023causal}.

\subsection{Testable Implications of \eqref{eq:sufficient1} -- \eqref{eq:sufficient3}} \label{subsec:testable}

\paragraph*{Testable implication of \eqref{eq:sufficient1}}

Under condition \eqref{eq:sufficient1}, Proposition \ref{prop:ITT} implies that $\mathrm{ASED}_{S_n}(z, t, t') = 0$ for all $z \in \{ 0, 1 \}$ and $t, t' \in \mathcal{T}$.
Thus, if $\mathrm{ASED}_{S_n}(z, t, t') \neq 0$ is observed, it indicates a violation of condition \eqref{eq:sufficient1}.

\paragraph*{Testable implication of \eqref{eq:sufficient2}}

Denote $m_i^{(1)}(y, z, t) \coloneqq \bbE[ \bm{1}\{ Y_i \le y \} D_i \mid Z_i = z, T_i = t ]$ for $y \in \mathbb{R}$, $z \in \{ 0, 1 \}$, and $t \in \mathcal{T}$.
Condition \eqref{eq:sufficient2} implies that $m_i^{(1)}(y, z, t) = \sum_{\bm{z}_{-i} \in \{ 0, 1\}^{n-1}} \bm{1}\{ Y_i(1) \le y \} D_i(z, \bm{z}_{-i}) \pi_i(\bm{z}_{-i}, t)$ by Assumption \ref{as:independence}.
As a result, we have $m_i^{(1)}(y, 1, t) - m_i^{(1)}(y, 0, t) \ge 0$ under Assumption \ref{as:monotone1}.
Note that this inequality might not hold without condition \eqref{eq:sufficient2}.
Similarly, letting $m_i^{(0)}(y, z, t) \coloneqq \bbE[ \bm{1}\{ Y_i \le y \} (1 - D_i) \mid Z_i = z, T_i = t ]$, we can show that $m_i^{(0)}(y, 1, t) - m_i^{(0)}(y, 0, t) \le 0$.
Note that although these inequalities cannot be directly tested for each $i$, we can check whether they hold or not on average for some sub-samples.  

\paragraph*{Testable implication of \eqref{eq:sufficient3}}

Let $g_i: \{ 0, 1 \}^{n-1} \to \mathbb{R}_+$ be a known non-negative function of $\bm{D}_{-i}$.
Then, we can show that
\begin{align*}
	\bbE[ D_i g_i(\bm{D}_{-i}) \mid Z_i = 0, T_i = t ]
	& = \sum_{\bm{z}_{-i} \in \{0, 1\}^{n-1}} \mathcal{AT}_i(\bm{z}_{-i}) g_i(\bm{D}_{-i}(Z_i = 0, \bm{Z}_{-i} = \bm{z}_{-i})) \pi_i(\bm{z}_{-i}, t)
\end{align*}
by Assumptions \ref{as:independence} and \ref{as:monotone1}.
Similarly, 
\begin{align*}
	\bbE[ D_i g_i(\bm{D}_{-i}) \mid Z_i = 1, T_i = t ]
	& = \sum_{\bm{z}_{-i} \in \{0, 1\}^{n-1}} \{ \mathcal{C}_i(\bm{z}_{-i}) + \mathcal{AT}_i(\bm{z}_{-i}) \} g_i(\bm{D}_{-i}(Z_i = 1, \bm{Z}_{-i} = \bm{z}_{-i})) \pi_i(\bm{z}_{-i}, t).
\end{align*}
Thus, condition \eqref{eq:sufficient3} leads to
\begin{align*}
	&\bbE[ D_i g_i(\bm{D}_{-i}) \mid Z_i = 1, T_i = t ] - \bbE[ D_i g_i(\bm{D}_{-i}) \mid Z_i = 0, T_i = t ] \\
	& \qquad \qquad =  \sum_{\bm{z}_{-i} \in \{0, 1\}^{n-1}} \mathcal{C}_i(\bm{z}_{-i}) g_i(\bm{D}_{-i}(Z_i = 1, \bm{Z}_{-i} = \bm{z}_{-i})) \pi_i(\bm{z}_{-i}, t) \ge 0.
\end{align*}
A similar argument give that $\bbE[ (1 - D_i) g_i(\bm{D}_{-i}) \mid Z_i = 0, T_i = t ] - \bbE[ ( 1- D_i) g_i(\bm{D}_{-i}) \mid Z_i = 1, T_i = t ]  \ge 0$.
Then, these inequalities provide testable implications of \eqref{eq:sufficient3}.

\subsection{Additional discussion on local average indirect effect} \label{sec:AIE2}

We investigate a causal interpretation of the Wald-type estimand $\mathrm{AIEY}_{S_n} / \mathrm{AIED}_{S_n}$, which differs from the one in Theorem \ref{thm:LAIE1}.

For each given $\bm{z}_{-i} \in \{ 0, 1 \}^{n-1}$, let $\mathcal{C}_j(\bm{Z}_{-i} = \bm{z}_{-i}) \coloneqq \bm{1}\{ D_j(Z_i = 1, \bm{Z}_{-i} = \bm{z}_{-i}) = 1, D_j(Z_i = 0, \bm{Z}_{-i} = \bm{z}_{-i}) = 0 \}$ denote $j$'s compliance status with respect to $i$'s IV.
We write its realization as $\mathcal{C}_j(\bm{Z}_{-i})$.

\begin{assumption}[Relevance 4] \label{as:relevance4}
	$|S_n|^{-1} \sum_{i \in S_n} \sum_{j \in \mathcal{E}_i} \bbE[ \mathcal{C}_j(\bm{Z}_{-i}) ] \ge c$ for a constant $c > 0$.
\end{assumption}

\begin{assumption}[Monotonicity 4] \label{as:monotone4}
	$D_j(Z_i = 1, \bm{Z}_{-i} = \bm{z}_{-i}) \ge D_j(Z_i = 0, \bm{Z}_{-i} = \bm{z}_{-i})$ for all $i \in S_n$, $j \in \mathcal{E}_i$, and $\bm{z}_{-i} \in \{ 0, 1 \}^{n-1}$ such that $\pi_i(\bm{z}_{-i}) > 0$.
\end{assumption}

\begin{assumption}[Restricted interference 4] \label{as:restrict4}
	For all $i \in S_n$, $j \in \mathcal{E}_i$, and $\bm{z}_{-i} \in \{ 0, 1 \}^{n-1}$ such that $\pi_i(\bm{z}_{-i}) > 0$, $y_j(Z_i = 1, \bm{Z}_{-i} = \bm{z}_{-i}) = y_j(Z_i = 0, \bm{Z}_{-i} = \bm{z}_{-i})$ holds whenever $D_j(Z_i = 1, \bm{Z}_{-i} = \bm{z}_{-i}) = D_j(Z_i = 0, \bm{Z}_{-i} = \bm{z}_{-i})$.
\end{assumption}

Assumption \ref{as:relevance4} requires the existence of compliers in this context.
Under Assumption \ref{as:monotone4}, each unit's IV can only affect the treatment status of other units in the non-negative direction.
Assumption \ref{as:restrict2} restricts the interference structure in a different way from Assumption \ref{as:restrict1}.
In the same manner as in Section \ref{subsec:LADE}, we can see that a sufficient condition for Assumption \ref{as:restrict4} is no treatment spillover effect on the outcome, that is, \eqref{eq:sufficient2}.
Meanwhile, \eqref{eq:sufficient1} and \eqref{eq:sufficient3} do not fulfill Assumption \ref{as:restrict4}.
In this sense, Assumption \ref{as:restrict4} is more restrictive than Assumptions \ref{as:restrict1} and \ref{as:restrict2}.

The next theorem shows that the Wald-type estimand $\mathrm{AIEY}_{S_n} / \mathrm{AIED}_{S_n}$ has a causal characterization.
The proof is given at the end of this section.

\begin{theorem} \label{thm:LAIE2}
	Under Assumptions \ref{as:independence} and \ref{as:relevance4} -- \ref{as:restrict4}, $\mathrm{AIEY}_{S_n} / \mathrm{AIED}_{S_n}$ equals to
	\begin{align*}
		\sum_{i \in S_n} \sum_{\bm{z}_{-i} \in \{ 0, 1\}^{n-1}} \sum_{j \in \mathcal{E}_i} \{ y_j(Z_i = 1, \bm{Z}_{-i} = \bm{z}_{-i}) - y_j(Z_i = 0, \bm{Z}_{-i} = \bm{z}_{-i}) \} \frac{ \mathcal{C}_j(\bm{Z}_{-i} = \bm{z}_{-i}) \pi_i(\bm{z}_{-i}) }{ \sum_{i \in S_n} \sum_{j \in \mathcal{E}_i} \bbE[ \mathcal{C}_j(\bm{Z}_{-i}) ] }.
	\end{align*}
\end{theorem}

For a complier $j$ such that $\mathcal{C}_j(\bm{Z}_{-i} = \bm{z}_{-i}) = 1$, observe that
\begin{align*}
	& y_j(Z_i = 1, \bm{Z}_{-i} = \bm{z}_{-i}) - y_j(Z_i = 0, \bm{Z}_{-i} = \bm{z}_{-i}) \\
	& = \underbrace{ Y_j(1, \bm{D}_{-j}(Z_i = 1, \bm{Z}_{-i} = \bm{z}_{-i})) - Y_j(0, \bm{D}_{-j}(Z_i = 1, \bm{Z}_{-i} = \bm{z}_{-i})) }_{\text{effect of changing $j$'s treatment}} \\
	& \quad + \underbrace{ Y_j(0, \bm{D}_{-j}(Z_i = 1, \bm{Z}_{-i} = \bm{z}_{-i})) - Y_j(0, \bm{D}_{-j}(Z_i = 0, \bm{Z}_{-i} = \bm{z}_{-i})) }_{\text{effect of changing others' treatments through $i$'s IV}}.
\end{align*}
As such, the causal parameter in Theorem \ref{thm:LAIE2} can be regarded as the sum of the average effect of changing $j$'s treatment through $i$'s IV on $j$'s outcome and that of changing others' treatments.

\subsection{Proof of Theorem \ref{thm:LAIE2}}

Using Assumption \ref{as:independence}, observe that
\begin{align*}
	\mathrm{AIED}_{S_n}
	& = \frac{1}{|S_n|} \sum_{i \in S_n} \sum_{j \in \mathcal{E}_i} \bbE[ D_j(Z_i = 1, \bm{Z}_{-i}) - D_j(Z_i = 0, \bm{Z}_{-i}) ] \\
	& = \frac{1}{|S_n|} \sum_{i \in S_n} \sum_{j \in \mathcal{E}_i} \bbE[ \bm{1}\{ D_j(Z_i = 1, \bm{Z}_{-i}) \neq D_j(Z_i = 0, \bm{Z}_{-i}) \} ] \\
	& = \frac{1}{|S_n|} \sum_{i \in S_n} \sum_{j \in \mathcal{E}_i} \bbE[ \mathcal{C}_j(\bm{Z}_{-i}) ] \\
	& = \frac{1}{|S_n|} \sum_{i \in S_n} \sum_{\bm{z}_{-i} \in \{ 0, 1\}^{n-1}} \sum_{j \in \mathcal{E}_i} \mathcal{C}_j(\bm{Z}_{-i} = \bm{z}_{-i}) \pi_i(\bm{z}_{-i}),
\end{align*}
where the third equality follows from Assumption \ref{as:monotone4}.
Similarly, we have
\begin{align*}
	& \mathrm{AIEY}_{S_n} \\
	& = \frac{1}{|S_n|} \sum_{i \in S_n} \sum_{j \in \mathcal{E}_i} \bbE[ y_j(Z_i = 1, \bm{Z}_{-i}) - y_j(Z_i = 0, \bm{Z}_{-i}) ] \\
	& = \frac{1}{|S_n|} \sum_{i \in S_n} \sum_{j \in \mathcal{E}_i} \bbE[ \{ y_j(Z_i = 1, \bm{Z}_{-i}) - y_j(Z_i = 0, \bm{Z}_{-i}) \} \cdot \bm{1}\{ D_j(Z_i = 1, \bm{Z}_{-i}) \neq D_j(Z_i = 0, \bm{Z}_{-i}) \} ] \\
	& \quad + \frac{1}{|S_n|} \sum_{i \in S_n} \sum_{j \in \mathcal{E}_i} \bbE[ \{ y_j(Z_i = 1, \bm{Z}_{-i}) - y_j(Z_i = 0, \bm{Z}_{-i}) \} \cdot \bm{1}\{ D_j(Z_i = 1, \bm{Z}_{-i}) = D_j(Z_i = 0, \bm{Z}_{-i}) \} ] \\
	& = \frac{1}{|S_n|} \sum_{i \in S_n} \sum_{j \in \mathcal{E}_i} \bbE[ \{ y_j(Z_i = 1, \bm{Z}_{-i}) - y_j(Z_i = 0, \bm{Z}_{-i}) \} \mathcal{C}_j(\bm{Z}_{-i}) ] \\
	& = \frac{1}{|S_n|} \sum_{i \in S_n} \sum_{\bm{z}_{-i} \in \{ 0, 1 \}^{n-1}} \sum_{j \in \mathcal{E}_i} \{ y_j(Z_i = 1, \bm{Z}_{-i} = \bm{z}_{-i}) - y_j(Z_i = 0, \bm{Z}_{-i} = \bm{z}_{-i}) \} \mathcal{C}_j(\bm{Z}_{-i} = \bm{z}_{-i}) \pi_i(\bm{z}_{-i}),
\end{align*}
where the third equality follows from Assumptions \ref{as:monotone4} and \ref{as:restrict4}.
In conjunction with Assumption \ref{as:relevance4}, we obtain the desired result.
\qed


\section{Monte Carlo Simulation}\label{sec:MC}

We investigate the finite sample properties of the proposed methods.
The next subsection considers a simple ring-shaped network, considering the ease of obtaining the closed-form expressions for the population parameters.
Subsequently, we conduct similar simulation analysis based on real friendship networks of students.

\subsection{Simulations on a ring-shaped network}

Suppose that the individuals are aligned as a circle in the order $i = 1,2, \ldots$, forming a ring-shaped network $\bm{A}$.
Define $G_{ij}^{(L)}$ as a dummy variable indicating whether $j$ is within $L$ path-length from $i$ on $\bm{A}$:
\begin{align*}
	G_{ij}^{(L)} =
	\begin{cases}
		1 & \text{if} \;\; \min\{ | i - j |, | i - j + n| , | i - j - n| \} \le L\\
		0 & \text{otherwise}
	\end{cases}
\end{align*}

In this Monte Carlo analysis, the following two DGPs are considered:
\begin{align*}
	\textbf{DGP 1:} \quad 
	Y_i
	& = \beta_{0i} + \beta_{1i} D_i, \\
	D_i
	& = \mathbf{1}\left\{\gamma_{0i} + \gamma_1 Z_i + \gamma_2 \sum_{j \neq i} G_{ij}^{(L)} Z_j \ge 0 \right\}, 
\end{align*}
where $\beta_{0i}$, $\beta_{1i}$, and $\gamma_{0i}$ are drawn from $\mathrm{Normal}(1, 1)$, $\mathrm{Uniform}(1, 2)$, and $\mathrm{Normal}(-1.5, 1)$, respectively, $(\gamma_1, \gamma_2) = (1, 1)$, and $Z_i$'s are IID $\mathrm{Bernoulli}(0.4)$.
\begin{align*}
	\textbf{DGP 2:} \quad 
	Y_i
	& = \beta_{0i} + \beta_{1i} D_i + \beta_{2i} \sum_{j \neq i} G_{ij}^{(L)} D_j, \\
	D_i
	& = \mathbf{1}\left\{\gamma_{0i} + \gamma_1 Z_i \ge 0 \right\}, 
\end{align*}
where $\beta_{0i}$, $\beta_{1i}$, $\beta_{2i}$, and $\gamma_{0i}$ are drawn from $\mathrm{Normal}(1, 1)$, $\mathrm{Uniform}(1, 2)$,  $\mathrm{Normal}(1, 1)$, and $\mathrm{Normal}(-1, 1)$, respectively, $\gamma_1 = 1$, and $Z_i$'s are IID $\mathrm{Bernoulli}(0.4)$.
The individual-specific coefficients are drawn only once, and they are fixed throughout the simulations.

For both DGPs, we consider two cases $L \in \{2, 3\}$.
For the specification of IEM, we consider two versions for each DGP in which the one is a correctly specified IEM and the other is misspecified: 

\begin{table*}[h]
	\begin{center}
		\begin{tabular}{c|ll}
			DGP & \multicolumn{1}{c}{Correct IEM} & \multicolumn{1}{c}{Inorrect IEM} \\ \hline
			\\[-1em]
			\textbf{1} & $T_i = \sum_{j \neq i} G_{ij}^{(L)}Z_j$ & $T_i = \sum_{j \neq i} G_{ij}^{(1)} Z_j$ \\
			\\[-1em]
			\textbf{2} & $T_i = \sum_{j \neq i} G_{ij}^{(L)}\bm{1} \left\{ \gamma_{0j} + \gamma_1 Z_j \ge 0 \right\}$ & $T_i = \sum_{j \neq i} G_{ij}^{(1)}\bm{1} \left\{ \gamma_{0j} + \gamma_1 Z_j \ge 0 \right\}$
		\end{tabular}
	\end{center}
\end{table*}

As the target parameters of interest, we consider $(\mathrm{ADEY}_{N_n}(2), \mathrm{LADE}_{N_n}(2), \mathrm{AIEY}_{N_n}, \mathrm{LAIE}_{N_n})$ in DGP 1 and $(\mathrm{ADEY}_{N_n}(1), \mathrm{LADE}_{N_n}(1), \mathrm{AIEY}_{N_n}, \mathrm{LAIE}_{N_n})$ in DGP 2.
The closed-form expression for $\mu_i^D(z, t)$, $\mu_i^Y(z, t)$, $\mu_{ji}^D(z)$, and $\mu_{ji}^Y(z)$ under correct and incorrect IEMs are summarized as below: letting $Z \sim \mathrm{Bernoulli}(0.4)$ and $\tilde Z^{(p)} \sim \mathrm{Binomial}(p, 0.4)$, for DGP 1,
\begin{small}\begin{align*}
		\textbf{Correct IEM} \qquad 
		\mu_i^D(z, t)
		& = \bm{1} \left\{ \gamma_{0i} + \gamma_1 z + \gamma_2 t \ge 0 \right\}, \\
		\mu_i^Y(z, t)
		& = \beta_{0i} + \beta_{1i} \bm{1} \left\{ \gamma_{0i} + \gamma_1 z + \gamma_2 t \ge 0 \right\}, \\
		\mu_{ji}^D(z)
		& = \Pr\left( \gamma_1 Z  + \gamma_2 \tilde Z^{(2L - 1)} \ge -(\gamma_{0j} + \gamma_2 z) \right) \quad \text{for $j \in \mathcal{E}_i$,}\\
		\mu_{ji}^Y(z)
		& = \beta_{0j} + \beta_{1j} \Pr\left( \gamma_1 Z  + \gamma_2 \tilde Z^{(2L - 1)} \ge -(\gamma_{0j} + \gamma_2 z) \right) \quad \text{for $j \in \mathcal{E}_i$.}
\end{align*}\end{small}
Similarly,
\begin{small}\begin{align*}
		\textbf{Incorrect IEM} \qquad
		\mu_i^D(z, t)
		& = \Pr\left( \gamma_2 \tilde Z^{(2L - 2)}  \ge -(\gamma_{0i} + \gamma_1 z + \gamma_2 t) \right), \\
		\mu_i^Y(z, t)
		& = \beta_{0i} + \beta_{1i} \Pr\left( \gamma_2 \tilde Z^{(2L - 2)}  \ge -(\gamma_{0i} + \gamma_1 z + \gamma_2 t) \right).
\end{align*}\end{small}
Note that the set $\mathcal{E}_i$ induced by the incorrect IEM is a subset of the true interference set and the neighbors are not distinguished in this DGP, hence the forms of $\mu_{ji}^D(z)$ and $\mu_{ji}^Y(z)$ do not change as long as $j \in \mathcal{E}_i$.
Similarly, for DGP 2, we have
\begin{small}\begin{align*}
		\textbf{Correct IEM} \qquad
		\mu_i^D(z, t)
		& = \bm{1} \left\{ \gamma_{0i} + \gamma_1 z \ge 0 \right\}, \\
		\mu_i^Y(z, t)
		& = \beta_{0i} + \beta_{1i} \bm{1} \left\{ \gamma_{0i} + \gamma_1 z \ge 0 \right\} + \beta_{2i} t, \\
		\mu_{ji}^D(z)
		& = \Pr \left( Z \ge -\gamma_{0j}/\gamma_1  \right) \quad \text{for $j \in \mathcal{E}_i$}, \\
		\mu_{ji}^Y(z)
		& = \beta_{0j} + \beta_{1j} \Pr \left( Z \ge -\gamma_{0j}/\gamma_1  \right) + \beta_{2j} \bm{1} \left\{ \gamma_{0i} + \gamma_1 z \ge 0 \right\} \\
		& \quad + \beta_{2j} \sum_{l \neq j,i} G_{jl}^{(L)} \Pr \left( Z \ge -\gamma_{0l}/\gamma_1  \right) \quad \text{for $j \in \mathcal{E}_i$}, 
\end{align*}\end{small}
and
\begin{small}\begin{align*}
		\textbf{Incorrect IEM} \qquad
		\mu_i^D(z, t)
		& = \bm{1} \left\{ \gamma_{0i} + \gamma_1 z \ge 0 \right\}, \\
		\quad \mu_i^Y(z, t)
		& = \beta_{0i} + \beta_{1i} \bm{1} \left\{ \gamma_{0i} + \gamma_1 z \ge 0 \right\} + \beta_{2i} t + \beta_{2i} \sum_{j \neq i}( G_{ij}^{(L)} - G_{ij}^{(1)}) \Pr \left( Z \ge -\gamma_{0j}/\gamma_1  \right).
\end{align*}\end{small}
The other population parameters can be computed analogously.

The data are generated for two sample sizes $n \in \{500, 1000\}$.
Note that in our DGPs, all individuals have the same network structure and the same distribution of $T_i$.
Thus, we use the whole sample $N_n$ as $S_n$ (see Remark \ref{remark:subpopulation} for a related discussion).
The performance of the estimators is measured in terms of the bias and the root mean squared error (RMSE) based on 1,000 Monte Carlo repetitions.

The results are summarized in Table \ref{table:MC_bias_rmse}.
First, we can see that our estimators perform satisfactorily overall.
Second, in the LADE estimation for DGP 1, the RMSE values tend to be very large.
This is because there are several times in the 1,000 repetitions when the estimated ADED is almost zero (when the ADED estimate was exactly zero, we redrew a new $Z$), resulting in extremely large LADE estimates in such cases.
This phenomenon is not prominent in DGP 2.
Lastly, as the range of interactions ($L$) increases, the RMSE values tend to increase slightly.

\begin{table}[ht]
	\caption{Bias and RMSE}
	\label{table:MC_bias_rmse}
	\begin{footnotesize}
		\begin{center}
			\begin{tabular}{rr|rrrrrrrr}
				\hline
				\textbf{DGP 1} &  & \multicolumn{8}{|c}{\underline{Correct IEM}} \\
				&  & \multicolumn{2}{|c}{ADEY} & \multicolumn{2}{c}{LADE} & \multicolumn{2}{c}{AIEY} & \multicolumn{2}{c}{LAIE} \\
				\multicolumn{1}{c}{$L$} & \multicolumn{1}{c}{$n$} & \multicolumn{1}{|c}{Bias} & \multicolumn{1}{c}{RMSE} & \multicolumn{1}{c}{Bias} & \multicolumn{1}{c}{RMSE} & \multicolumn{1}{c}{Bias} & \multicolumn{1}{c}{RMSE} & \multicolumn{1}{c}{Bias} & \multicolumn{1}{c}{RMSE} \\ 
				\hline
				2 & 500 & -0.0051  & 0.1797  & -0.0094  & 0.7550  & -0.0181  & 0.2224  & 0.2134  & 1.4212  \\
				& 1000 & 0.0025  & 0.1236  & 0.0132  & 0.4954  & -0.0106  & 0.1668  & 0.0740  & 0.9447  \\
				3 & 500 & -0.0044  & 0.1812  & -0.0082  & 0.7404  & -0.0137  & 0.3131  & 0.5096  & 2.5651  \\
				& 1000 & -0.0038  & 0.1302  & -0.0054  & 0.5175  & -0.0024  & 0.2268  & 0.3168  & 1.7220  \\
				\hline
				&  & \multicolumn{8}{|c}{\underline{Incorrect IEM}}  \\
				&  & \multicolumn{2}{|c}{ADEY} & \multicolumn{2}{c}{LADE} & \multicolumn{2}{c}{AIEY} & \multicolumn{2}{c}{LAIE} \\
				\multicolumn{1}{c}{$L$} & \multicolumn{1}{c}{$n$} & \multicolumn{1}{|c}{Bias} & \multicolumn{1}{c}{RMSE} & \multicolumn{1}{c}{Bias} & \multicolumn{1}{c}{RMSE} & \multicolumn{1}{c}{Bias} & \multicolumn{1}{c}{RMSE} & \multicolumn{1}{c}{Bias} & \multicolumn{1}{c}{RMSE} \\ 
				\hline
				2 & 500 & 0.0040  & 0.2503  & -0.6032  & 11.2950  & -0.0080  & 0.1658  & 0.1060  & 0.8640  \\
				& 1000 & -0.0006  & 0.1735  & 0.0459  & 3.5362  & -0.0078  & 0.1181  & 0.0207  & 0.5420  \\
				3 & 500 & 0.0025  & 0.2433  & 0.2594  & 25.0826  & -0.0027  & 0.1654  & 0.1696  & 1.1689  \\
				& 1000 & -0.0003  & 0.1654  & -0.0539  & 14.3886  & -0.0033  & 0.1205  & 0.0806  & 0.7671  \\
				\hline \hline
				\textbf{DGP 2} &  & \multicolumn{8}{|c}{\underline{Correct IEM}} \\
				&  & \multicolumn{2}{|c}{ADEY} & \multicolumn{2}{c}{LADE} & \multicolumn{2}{c}{AIEY} & \multicolumn{2}{c}{LAIE} \\
				\multicolumn{1}{c}{$L$} & \multicolumn{1}{c}{$n$} & \multicolumn{1}{|c}{Bias} & \multicolumn{1}{c}{RMSE} & \multicolumn{1}{c}{Bias} & \multicolumn{1}{c}{RMSE} & \multicolumn{1}{c}{Bias} & \multicolumn{1}{c}{RMSE} & \multicolumn{1}{c}{Bias} & \multicolumn{1}{c}{RMSE} \\ 
				\hline
				2 & 500 & -0.0022  & 0.2172  & 0.0025  & 0.6014  & -0.0364  & 0.4781  & -0.0848  & 1.2994  \\
				& 1000 & -0.0155  & 0.1491  & 0.0105  & 0.4391  & -0.0202  & 0.3302  & -0.0641  & 0.9589  \\
				3 & 500 & 0.0103  & 0.2409  & 0.0141  & 0.6708  & -0.0812  & 0.7975  & -0.2019  & 2.2011  \\
				& 1000 & -0.0297  & 0.1804  & -0.0354  & 0.5362  & -0.0330  & 0.5545  & -0.1047  & 1.6226  \\
				\hline
				&  & \multicolumn{8}{|c}{\underline{Incorrect IEM}}  \\
				&  & \multicolumn{2}{|c}{ADEY} & \multicolumn{2}{c}{LADE} & \multicolumn{2}{c}{AIEY} & \multicolumn{2}{c}{LAIE} \\
				\multicolumn{1}{c}{$L$} & \multicolumn{1}{c}{$n$} & \multicolumn{1}{|c}{Bias} & \multicolumn{1}{c}{RMSE} & \multicolumn{1}{c}{Bias} & \multicolumn{1}{c}{RMSE} & \multicolumn{1}{c}{Bias} & \multicolumn{1}{c}{RMSE} & \multicolumn{1}{c}{Bias} & \multicolumn{1}{c}{RMSE} \\ 
				\hline
				2 & 500 & -0.0229  & 0.3248  & -0.0522  & 0.9437  & -0.0216  & 0.3109  & -0.0557  & 0.8558  \\
				& 1000 & -0.0110  & 0.2138  & -0.0367  & 0.6318  & -0.0094  & 0.2033  & -0.0305  & 0.5978  \\
				3 & 500 & -0.0320  & 0.4182  & -0.0730  & 1.2309  & -0.0300  & 0.3921  & -0.0810  & 1.0937  \\
				& 1000 & -0.0155  & 0.2694  & -0.0514  & 0.8136  & -0.0144  & 0.2505  & -0.0461  & 0.7471  \\
				\hline
			\end{tabular}
		\end{center}
	\end{footnotesize}
\end{table}

We next examine the performance of the HAC estimator and the wild bootstrap approach introduced in Sections \ref{subsubsec:HAC} and \ref{subsubsec:Bootstrap}, respectively.
The DGPs and the target parameters considered are the same as above.
The number of bootstrap repetitions is set to 500.
For each parameter in each setup, we compute the coverage rate of the 95\% confidence interval obtained based on these two approaches.
The bandwidth is chosen from $b_n \in \{L + 1, L + 2\}$ for both approaches.
The results are summarized in Tables \ref{table:MC_covratio_hac} and \ref{table:MC_covratio_boot}.
The performances of the HAC estimator and the wild bootstrap are quite similar.   
They show that the estimated confidence intervals have coverage reasonably close to the nominal probability, except in some special circumstances (e.g., the LADE parameter for DGP 1 under the incorrect IEM).
However, it should be kept in mind that, since our variance estimators contain asymptotic biases as shown in Theorem \ref{thm:HAC}, any over- or under-coverage must be due to that bias to some extent.

\begin{table}[ht]
	\caption{Coverage ratio of the 95\% CI (HAC)}
	\label{table:MC_covratio_hac}
	\begin{footnotesize}
		\begin{center}
			\begin{tabular}{cccc|cccc|cccc}
				\hline
				&  &  &  & \multicolumn{4}{c|}{\underline{Correct IEM}} & \multicolumn{4}{c}{\underline{Incorrect IEM}}  \\
				DGP & $L$ & $n$ & $b_n$ & ADEY & LADE & AIEY & LAIE & ADEY & LADE & AIEY & LAIE \\
				\hline
				1 & 2 & 500 & $L+1$ & 0.935  & 0.961  & 0.938  & 0.944  & 0.924  & 0.990  & 0.932  & 0.934  \\
				&  &  & $L+2$ & 0.930  & 0.960  & 0.941  & 0.941  & 0.912  & 0.987  & 0.934  & 0.936  \\
				&  & 1000 & $L+1$ & 0.950  & 0.966  & 0.937  & 0.944  & 0.944  & 0.982  & 0.933  & 0.945  \\
				&  &  & $L+2$ & 0.946  & 0.964  & 0.938  & 0.944  & 0.943  & 0.983  & 0.930  & 0.947  \\
				& 3 & 500 & $L+1$ & 0.933  & 0.946  & 0.931  & 0.947  & 0.912  & 0.995  & 0.926  & 0.947  \\
				&  &  & $L+2$ & 0.926  & 0.945  & 0.924  & 0.945  & 0.905  & 0.995  & 0.923  & 0.938  \\
				&  & 1000 & $L+1$ & 0.934  & 0.949  & 0.931  & 0.948  & 0.948  & 0.996  & 0.932  & 0.952  \\
				&  &  & $L+2$ & 0.934  & 0.949  & 0.932  & 0.950  & 0.946  & 0.996  & 0.929  & 0.950  \\
				\hline
				2 & 2 & 500 & $L+1$ & 0.931  & 0.948  & 0.939  & 0.944  & 0.950  & 0.953  & 0.937  & 0.936  \\
				&  &  & $L+2$ & 0.934  & 0.943  & 0.933  & 0.939  & 0.935  & 0.949  & 0.939  & 0.937  \\
				&  & 1000 & $L+1$ & 0.944  & 0.957  & 0.947  & 0.946  & 0.944  & 0.938  & 0.950  & 0.942  \\
				&  &  & $L+2$ & 0.939  & 0.953  & 0.945  & 0.945  & 0.943  & 0.938  & 0.945  & 0.945  \\
				& 3 & 500 & $L+1$ & 0.938  & 0.949  & 0.928  & 0.929  & 0.926  & 0.941  & 0.921  & 0.928  \\
				&  &  & $L+2$ & 0.934  & 0.941  & 0.926  & 0.921  & 0.922  & 0.934  & 0.915  & 0.925  \\
				&  & 1000 & $L+1$ & 0.930  & 0.944  & 0.934  & 0.929  & 0.947  & 0.943  & 0.944  & 0.945  \\
				&  &  & $L+2$ & 0.930  & 0.941  & 0.927  & 0.928  & 0.944  & 0.939  & 0.939  & 0.942  \\
				\hline
			\end{tabular}
		\end{center}
	\end{footnotesize}
\end{table}

\begin{table}[ht]
	\caption{Coverage ratio of the 95\% CI (Bootstrap)}
	\label{table:MC_covratio_boot}
	\begin{footnotesize}
		\begin{center}
			\begin{tabular}{cccc|cccc|cccc}
				\hline
				&  &  &  & \multicolumn{4}{c|}{\underline{Correct IEM}} & \multicolumn{4}{c}{\underline{Incorrect IEM}}  \\
				DGP & $L$ & $n$ & $b_n$ & ADEY & LADE & AIEY & LAIE & ADEY & LADE & AIEY & LAIE \\
				\hline
				1 & 2 & 500 & $L+1$ & 0.936  & 0.961  & 0.931  & 0.943  & 0.922  & 0.986  & 0.933  & 0.936  \\
				&  &  & $L+2$ & 0.933  & 0.954  & 0.931  & 0.938  & 0.910  & 0.989  & 0.931  & 0.935  \\
				&  & 1000 & $L+1$ & 0.948  & 0.968  & 0.935  & 0.942  & 0.941  & 0.983  & 0.924  & 0.935  \\
				&  &  & $L+2$ & 0.949  & 0.965  & 0.932  & 0.940  & 0.941  & 0.986  & 0.924  & 0.940  \\
				& 3 & 500 & $L+1$ & 0.935  & 0.944  & 0.925  & 0.950  & 0.911  & 0.998  & 0.923  & 0.945  \\
				&  &  & $L+2$ & 0.926  & 0.939  & 0.915  & 0.947  & 0.905  & 0.995  & 0.916  & 0.940  \\
				&  & 1000 & $L+1$ & 0.937  & 0.950  & 0.921  & 0.946  & 0.947  & 0.995  & 0.933  & 0.950  \\
				&  &  & $L+2$ & 0.934  & 0.945  & 0.921  & 0.944  & 0.940  & 0.995  & 0.930  & 0.949  \\
				\hline
				2 & 2 & 500 & $L+1$ & 0.933  & 0.942  & 0.932  & 0.941  & 0.937  & 0.950  & 0.935  & 0.934  \\
				&  &  & $L+2$ & 0.931  & 0.938  & 0.932  & 0.937  & 0.935  & 0.951  & 0.931  & 0.929  \\
				&  & 1000 & $L+1$ & 0.944  & 0.956  & 0.944  & 0.943  & 0.940  & 0.938  & 0.946  & 0.939  \\
				&  &  & $L+2$ & 0.943  & 0.953  & 0.944  & 0.940  & 0.937  & 0.936  & 0.948  & 0.941  \\
				& 3 & 500 & $L+1$ & 0.935  & 0.950  & 0.927  & 0.924  & 0.920  & 0.941  & 0.923  & 0.927  \\
				&  &  & $L+2$ & 0.929  & 0.943  & 0.926  & 0.919  & 0.917  & 0.935  & 0.924  & 0.925  \\
				&  & 1000 & $L+1$ & 0.930  & 0.942  & 0.931  & 0.924  & 0.939  & 0.941  & 0.943  & 0.940  \\
				&  &  & $L+2$ & 0.929  & 0.938  & 0.928  & 0.925  & 0.942  & 0.938  & 0.937  & 0.943  \\
				\hline
			\end{tabular}
		\end{center}
	\end{footnotesize}
\end{table}

\subsection{Simulations on a real students' network}

Next, we investigate the performance of our estimators and inferential procedures based on a real network dataset that is taken from the same students' friendship data as we used in Section \ref{subsec:application}.
For this dataset, we select the schools whose size (number of respondents) is greater than 500, leaving 13,523 students as our population $N_n$.
As mentioned in the main text, we need to select a subset $S_n$ of the population to maintain the distributional homogeneity of treatment exposures.
In this simulation analysis, we construct $S_n$ by $S_n = \{i \in N_n : \sum_{j \neq i} A_{ij} = 11\}$, which is the largest subgroup among the groups defined by the network degree, with $|S_n| =$ 1,372.
Then, based on this network data, we consider the following DGP:
\begin{align*}
	Y_i
	& = \beta_{0i} + \beta_{1i} D_i, \\
	D_i
	& = \mathbf{1}\left\{\gamma_{0i} + \gamma_1 Z_i + \gamma_{2i} T_i \ge 0 \right\},
\end{align*}
where $T_i = \mathbf{1}\{\sum_{j \neq i} A_{ij} Z_j \ge 5\}$, $\beta_{0i}$, $\beta_{1i}$, $\gamma_{0i}$, and $\gamma_{2i}$ are drawn from $\mathrm{Normal}(0, 1)$, $\mathrm{Normal}(1, 1)$,  $\mathrm{Normal}(-1, 1)$, and $\mathrm{Uniform}(1, 2)$, respectively, $\gamma_1 = 0.7$, and $Z_i$'s are IID $\mathrm{Bernoulli}(0.5)$.
Unlike the previously considered DGPs, it is not easy to derive the closed-form expressions for the population parameters in this case.
Thus, we approximate the population values of the target parameters of interest $(\mathrm{ADEY}_{S_n}(1), \mathrm{LADE}_{S_n}(1), \mathrm{AIEY}_{S_n}, \mathrm{LAIE}_{S_n})$ by Monte Carlo integration, and consider only the case of correct IEM.

Below, we report the performance of our estimators based on 1,000 repetitions and the coverage rate of the confidence intervals constructed by the HAC and wild bootstrap approaches in Table \ref{table:MC_summary_real}, altogether.
The bandwidth is chosen from $b_n \in \{5,6\}$, and the number of bootstrap repetitions is 500.
From these results, we can confirm that our estimators work reasonably well.
The coverage ratios are slightly below the nominal level, but this might be specific to the DGP considered here and also due to the asymptotic biases in the variance estimators.

\begin{table}[ht]
	\caption{Summary of simulation results}
	\label{table:MC_summary_real}
	\begin{footnotesize}
		\begin{center}
			\begin{tabular}{r|rrrrrrrr}
				\hline
				& \multicolumn{8}{|c}{\underline{Estimation performance}}  \\
				& \multicolumn{2}{|c}{ADEY} & \multicolumn{2}{c}{LADE} & \multicolumn{2}{c}{AIEY} & \multicolumn{2}{c}{LAIE} \\
				& \multicolumn{1}{c}{Bias} & \multicolumn{1}{c}{RMSE} & \multicolumn{1}{c}{Bias} & \multicolumn{1}{c}{RMSE} & \multicolumn{1}{c}{Bias} & \multicolumn{1}{c}{RMSE} & \multicolumn{1}{c}{Bias} & \multicolumn{1}{c}{RMSE} \\
				\hline
				& -0.0024  & 0.0886  & -0.0131  & 0.4646  & -0.0102  & 0.2589  & 0.0018  & 1.4445  \\
				\hline
				& \multicolumn{8}{|c}{\underline{Coverage ratio}}  \\
				& \multicolumn{2}{|c}{ADEY} & \multicolumn{2}{c}{LADE} & \multicolumn{2}{c}{AIEY} & \multicolumn{2}{c}{LAIE} \\
				\multicolumn{1}{c}{$b_n$} & \multicolumn{1}{|c}{HAC} & \multicolumn{1}{c}{Bootstrap} & \multicolumn{1}{c}{HAC} & \multicolumn{1}{c}{Bootstrap} & \multicolumn{1}{c}{HAC} & \multicolumn{1}{c}{Bootstrap} & \multicolumn{1}{c}{HAC} & \multicolumn{1}{c}{Bootstrap} \\
				\hline
				5 & 0.933  & 0.920  & 0.941  & 0.934  & 0.926  & 0.912  & 0.935  & 0.922  \\
				6 & 0.931  & 0.920  & 0.941  & 0.932  & 0.926  & 0.912  & 0.932  & 0.921  \\
				\hline
			\end{tabular}
		\end{center}
	\end{footnotesize}
\end{table}


\section{Additional Empirical Results} \label{sec:application2}

The descriptive statistics for the seed-eligible students are summarized in Table \ref{table:descriptive}.
We can see that all students without invitation did not actually join the intervention program, implying that the monotonicity condition in Assumption \ref{as:monotone1} holds.
It is also interesting that not just IEM1 but the distribution of IEM2 is also insensitive to the student's own invitation status.
This would suggest that one's treatment eligibility does not have substantial impacts on the others' treatment choices.
Indeed, we found that the conditional distribution of $\sum_{j \neq i} A_{ij} D_j$ given $Z_i = 1$ is almost identical to that given $Z_i = 0$.

Panel (a) of Table \ref{table:empirical2} presents the ADE estimates obtained by a constant IEM discussed in Remark \ref{remark:constantIEM1}.
Similar to the ADE estimates conditional on $T_i = t$ in Table \ref{table:empirical1}, we can see that receiving an invitation and the participation in the intervention program have statistically significant positive direct effects on the probability of wearing a wristband.

Panel (b) of Table \ref{table:empirical2} provides the results of the AOE estimation discussed in Appendix \ref{sec:overall}.
In line with the AIE estimation in Table \ref{table:empirical1}, we set $K = 1$.
All estimates are positive and statistically significant.
In particular, the results of $\mathrm{AOEY}_{S_n}$ and $\mathrm{LAOE}_{S_n}$ imply substantial overall effects of the treatment eligibility and the program participation on the probability of wearing a wristband.

In Panel (c) of Table \ref{table:empirical2}, we report the estimates and the HAC standard errors for the ASE parameters discussed in Appendix \ref{sec:spillover}.
Because the experimental design implies that $\tilde{\mathcal{C}}_i(0, t, t') = 0$ for all $i$, we focus on the ASE estimates when we set $Z_i = 1$.
We can see that almost all ASE estimates obtained by IEM1 are statistically insignificant.
Moreover, the coexistence of the positive and negative estimates of $\mathrm{ASED}_{S_n}(1, 1, 0)$ is inconsistent with Assumption \ref{as:monotone3}, implying the failure of the identification of $\mathrm{LASE}_{S_n}(1, 1, 0)$ for IEM1.
By contrast, the estimates obtained by IEM2 suggest the presence of significant positive spillover effects on both the outcome and the treatment receipt.
However, because the ASE parameters have clear causal interpretation only when the selected IEM is correct in the sense of Assumption \ref{as:correctIEM}, we should be cautious in interpreting the results obtained by IEM2.

\begin{table}[t]
	\caption{Descriptive statistics for seed-eligible students\label{table:descriptive}} 
	\begin{footnotesize}
		\begin{center}
			\begin{tabular}{rrrr}
				\hline 
				&  & \multicolumn{2}{c}{Treatment eligibility} \\ 
				\cline{3-4}
				& \multicolumn{1}{c}{Overall} & \multicolumn{1}{c}{$Z_i = 0$} & \multicolumn{1}{c}{$Z_i = 1$}\\ 
				\hline 
				\multicolumn{1}{l}{Treatment} &  &  &  \\ 
				$D_i = 0$ & 2,357 (79\%) & 1,491 (100\%) & 866 (58\%) \\ 
				$D_i = 1$ & 626 (21\%) & 0 (0\%) & 626 (42\%) \\ 
				\multicolumn{1}{l}{Outcome} &  &  &  \\ 
				$Y_i = 0$ & 2,684 (90\%) & 1,392 (93\%) & 1,292 (87\%) \\ 
				$Y_i = 1$ & 299 (10\%) & 99 (7\%) & 200 (13\%) \\ 
				\multicolumn{1}{l}{IEM1} &  &  &  \\ 
				$T_{1i} = 0$ & 1,458 (49\%) & 711 (48\%) & 747 (50\%) \\ 
				$T_{1i} = 1$ & 1,525 (51\%) & 780 (52\%) & 745 (50\%) \\ 
				\multicolumn{1}{l}{IEM2} &  &  &  \\ 
				$T_{2i} = 0$ & 2,334 (78\%) & 1,151 (77\%) & 1,183 (79\%) \\ 
				$T_{2i} = 1$ & 649 (22\%) & 340 (23\%) & 309 (21\%) \\ 
				\hline 
			\end{tabular}
		\end{center}
	\end{footnotesize}
\end{table}

\begin{table}[t]
	\caption{Additional empirical results} \label{table:empirical2}
	\begin{footnotesize}
		\begin{subtable}[h]{\textwidth}
			\centering
			\caption{Average direct effects (constant IEM)}
			\begin{tabular}{rrrcrrcrrcrr}
				\hline
				\multicolumn{3}{c}{\bfseries }&\multicolumn{1}{c}{\bfseries }&\multicolumn{2}{c}{\bfseries $\mathrm{ADEY}_{S_n}$}&\multicolumn{1}{c}{\bfseries }&\multicolumn{2}{c}{\bfseries $\mathrm{ADED}_{S_n}$}&\multicolumn{1}{c}{\bfseries }&\multicolumn{2}{c}{\bfseries $\mathrm{LADE}_{S_n}$}\tabularnewline
				\cline{5-6} \cline{8-9} \cline{11-12}
				\multicolumn{1}{c}{$S_n$}&\multicolumn{1}{c}{$|S_n|$}&\multicolumn{1}{c}{$b_n$}&\multicolumn{1}{c}{}&\multicolumn{1}{c}{Estimate}&\multicolumn{1}{c}{SE}&\multicolumn{1}{c}{}&\multicolumn{1}{c}{Estimate}&\multicolumn{1}{c}{SE}&\multicolumn{1}{c}{}&\multicolumn{1}{c}{Estimate}&\multicolumn{1}{c}{SE}\tabularnewline
				\hline
				$S_n^{\ge 1}$&$2244$&$2$&&$0.073$&$0.012$&&$0.416$&$0.030$&&$0.175$&$0.026$\tabularnewline
				\hline
			\end{tabular}
		\end{subtable}
		\hfill \\ \\ \hfill
		\begin{subtable}[h]{\textwidth}
			\centering
			\caption{Average overall effects with $K = 1$}
			\begin{tabular}{rrrcrrcrrcrrcrr}
				\hline
				\multicolumn{3}{c}{\bfseries }&\multicolumn{1}{c}{\bfseries }&\multicolumn{2}{c}{\bfseries $\mathrm{AOEY}_{S_n}$}&\multicolumn{1}{c}{\bfseries }&\multicolumn{2}{c}{\bfseries $\mathrm{AOED}_{S_n}$}&\multicolumn{1}{c}{\bfseries }&\multicolumn{2}{c}{\bfseries $\mathrm{ADED}_{S_n}$}&\multicolumn{1}{c}{\bfseries }&\multicolumn{2}{c}{\bfseries $\mathrm{LAOE}_{S_n}$}\tabularnewline
				\cline{5-6} \cline{8-9} \cline{11-12} \cline{14-15}
				\multicolumn{1}{c}{$S_n$}&\multicolumn{1}{c}{$|S_n|$}&\multicolumn{1}{c}{$b_n$}&\multicolumn{1}{c}{}&\multicolumn{1}{c}{Estimate}&\multicolumn{1}{c}{SE}&\multicolumn{1}{c}{}&\multicolumn{1}{c}{Estimate}&\multicolumn{1}{c}{SE}&\multicolumn{1}{c}{}&\multicolumn{1}{c}{Estimate}&\multicolumn{1}{c}{SE}&\multicolumn{1}{c}{}&\multicolumn{1}{c}{Estimate}&\multicolumn{1}{c}{SE}\tabularnewline
				\hline
				$S_n^{\ge 1}$&$2244$&$2$&&$0.185$&$0.048$&&$0.366$&$0.040$&&$0.416$&$0.030$&&$0.445$&$0.108$\tabularnewline
				\hline
			\end{tabular}
		\end{subtable}
		\hfill \\ \\ \hfill
		\begin{subtable}[h]{\textwidth}
			\centering
			\caption{Average spillover effects conditional on $Z_i = 1$}
			\begin{tabular}{rrrrcrrcrrcrr}
				\hline
				&\multicolumn{3}{c}{\bfseries }&\multicolumn{1}{c}{\bfseries }&\multicolumn{2}{c}{\bfseries $\mathrm{ASEY}_{S_n}(1,1,0)$}&\multicolumn{1}{c}{\bfseries }&\multicolumn{2}{c}{\bfseries $\mathrm{ASED}_{S_n}(1,1,0)$}&\multicolumn{1}{c}{\bfseries }&\multicolumn{2}{c}{\bfseries $\mathrm{LASE}_{S_n}(1,1,0)$}\tabularnewline
				\cline{6-7} \cline{9-10} \cline{12-13}
				&\multicolumn{1}{c}{$S_n$}&\multicolumn{1}{c}{$|S_n|$}&\multicolumn{1}{c}{$b_n$}&\multicolumn{1}{c}{}&\multicolumn{1}{c}{Estimate}&\multicolumn{1}{c}{SE}&\multicolumn{1}{c}{}&\multicolumn{1}{c}{Estimate}&\multicolumn{1}{c}{SE}&\multicolumn{1}{c}{}&\multicolumn{1}{c}{Estimate}&\multicolumn{1}{c}{SE}\tabularnewline
				\hline
				\multicolumn{13}{l}{\underline{(i) IEM1}} \tabularnewline
				&$S_n(1)$&$1023$&$2$&&$ 0.021$&$0.032$&&$-0.058$&$0.048$&&$-0.364$&$0.744$\tabularnewline
				&$S_n(2)$&$ 702$&$2$&&$ 0.011$&$0.042$&&$ 0.052$&$0.062$&&$ 0.210$&$0.719$\tabularnewline
				&$S_n(3)$&$ 315$&$2$&&$-0.159$&$0.098$&&$-0.206$&$0.128$&&$ 0.772$&$0.396$\tabularnewline
				\hline
				\multicolumn{13}{l}{\underline{(ii) IEM2}} \tabularnewline
				&$S_n(1)$&$1023$&$2$&&$0.252$&$0.052$&&$0.594$&$0.044$&&$0.425$&$0.084$\tabularnewline
				&$S_n(2)$&$ 702$&$2$&&$0.232$&$0.043$&&$0.750$&$0.046$&&$0.309$&$0.054$\tabularnewline
				&$S_n(3)$&$ 315$&$2$&&$0.246$&$0.068$&&$0.816$&$0.053$&&$0.302$&$0.083$\tabularnewline
				\hline
			\end{tabular}
		\end{subtable}
	\end{footnotesize}
\end{table}


\end{document}